\documentclass{article}
\usepackage{graphicx} 
\usepackage{caption} 
\usepackage{subcaption} 

\usepackage[english]{babel}
\usepackage{enumerate}
\usepackage{amsthm} 
\usepackage{siunitx} 
\usepackage{multirow} 

\usepackage[a4paper,top=3cm,bottom=3cm,left=3cm,right=3cm,marginparwidth=1.75cm]{geometry}
\usepackage[textsize=small]{todonotes}
\usepackage[numbers]{natbib}
\usepackage{pifont}
\usepackage{bm}
\usepackage{bbm}

\usepackage{amsmath, amssymb}
\usepackage{todonotes}

\makeatletter
\let\save@mathaccent\mathaccent
\newcommand*\if@single[3]{%
	\setbox0\hbox{${\mathaccent"0362{#1}}^H$}%
	\setbox2\hbox{${\mathaccent"0362{\kern0pt#1}}^H$}%
	\ifdim\ht0=\ht2 #3\else #2\fi
}
\newcommand*\rel@kern[1]{\kern#1\dimexpr\macc@kerna}
\newcommand*\widebar[1]{\@ifnextchar^{{\wide@bar{#1}{0}}}{\wide@bar{#1}{1}}}
\newcommand*\wide@bar[2]{\if@single{#1}{\wide@bar@{#1}{#2}{1}}{\wide@bar@{#1}{#2}{2}}}
\newcommand*\wide@bar@[3]{%
	\begingroup
	\def\mathaccent##1##2{%
		\let\mathaccent\save@mathaccent
		\if#32 \let\macc@nucleus\first@char \fi
		\setbox\z@\hbox{$\macc@style{\macc@nucleus}_{}$}%
		\setbox\tw@\hbox{$\macc@style{\macc@nucleus}{}_{}$}%
		\dimen@\wd\tw@
		\advance\dimen@-\wd\z@
		\divide\dimen@ 3
		\@tempdima\wd\tw@
		\advance\@tempdima-\scriptspace
		\divide\@tempdima 10
		\advance\dimen@-\@tempdima
		\ifdim\dimen@>\z@ \dimen@0pt\fi
		\rel@kern{0.6}\kern-\dimen@
		\if#31
		\overline{\rel@kern{-0.6}\kern\dimen@\macc@nucleus\rel@kern{0.4}\kern\dimen@}%
		\advance\dimen@0.4\dimexpr\macc@kerna
		\let\final@kern#2%
		\ifdim\dimen@<\z@ \let\final@kern1\fi
		\if\final@kern1 \kern-\dimen@\fi
		\else
		\overline{\rel@kern{-0.6}\kern\dimen@#1}%
		\fi
	}%
	\macc@depth\@ne
	\let\math@bgroup\@empty \let\math@egroup\macc@set@skewchar
	\mathsurround\z@ \frozen@everymath{\mathgroup\macc@group\relax}%
	\macc@set@skewchar\relax
	\let\mathaccentV\macc@nested@a
	\if#31
	\macc@nested@a\relax111{#1}%
	\else
	\def\gobble@till@marker##1\endmarker{}%
	\futurelet\first@char\gobble@till@marker#1\endmarker
	\ifcat\noexpand\first@char A\else
	\def\first@char{}%
	\fi
	\macc@nested@a\relax111{\first@char}%
	\fi
	\endgroup
}

\newcommand\norm[1]{\left\lVert#1\right\rVert}

\usepackage{amsfonts}
\usepackage{bbm}
\usepackage{enumitem}
\usepackage{graphicx}
\usepackage[colorlinks=true, allcolors=blue]{hyperref}
\usepackage{mathtools}
\usepackage{breqn}
\usepackage{float}
\usepackage{xcolor}
\usepackage{tcolorbox}
\usepackage{soul}

\newtheorem{theorem}{Theorem}[section]
\newtheorem{lemma}[theorem]{Lemma}
\newtheorem{corollary}[theorem]{Corollary}

\newtheorem{proposition}[theorem]{Proposition}

\newtheorem{remark}[theorem]{Remark}

\numberwithin{equation}{section}

\theoremstyle{definition}

\newtheorem{definition}[theorem]{Definition}

\newcommand{\upperRomannumeral}[1]{\uppercase\expandafter{\romannumeral#1}}
\newcommand{\lowerRomannumeral}[1]{\lowercase\expandafter{\romannumeral#1}}

\newcommand{\dd}{\mathrm{d}}

\def\E{{\mathbb{E}}}
\newcommand{\R}{{\mathbb R}}

\newcommand{\bd}{\begin{displaymath}}
\newcommand{\ed}{\end{displaymath}}
\newcommand{\be}{\begin{equation}}
\newcommand{\ee}{\end{equation}}
\newcommand{\bq}{\begin{eqnarray}}
\newcommand{\eq}{\end{eqnarray}}
\newcommand{\bn}{\begin{eqnarray*}}
\newcommand{\en}{\end{eqnarray*}}

\usepackage{mathtools}
\usepackage{authblk}
\mathtoolsset{showonlyrefs}

\DeclareMathOperator*{\esssup}{ess\,sup}
\usepackage[textsize=small]{todonotes}

\title{Fredholm Approach to Nonlinear Propagator Models}

\author[1]{Eduardo Abi Jaber\thanks{EAJ is grateful for the financial support from the Chaires FiME-FDD, Financial Risks, Deep Finance \& Statistics and Machine Learning and systematic methods in finance at Ecole Polytechnique.}}
\author[2]{Alessandro Bondi}
\author[3]{Nathan De Carvalho\thanks{NDC is grateful for the financial support provided by Engie Global Markets.}}
\author[4]{Eyal Neuman}
\author[5]{Sturmius Tuschmann\thanks{ST is supported by the EPSRC Centre for Doctoral Training in Mathematics of Random \mbox{Systems}: Analysis, Modelling and Simulation (EP/S023925/1).}}

\affil[1]{Ecole Polytechnique, CMAP}
\affil[2]{Luiss University, DEF}
\affil[3]{Université Paris Cité, LPSM}
\affil[4,5]{Department of Mathematics, Imperial College London}

\begin{document}

\maketitle

\begin{abstract}
We formulate and solve an optimal trading problem with alpha signals, where transactions induce a nonlinear transient price impact described by a general propagator model, including power-law decay. Using a variational approach, we demonstrate that the optimal trading strategy satisfies a nonlinear stochastic Fredholm equation with both forward and backward coefficients. We prove the existence and uniqueness of the solution under a monotonicity condition reflecting the nonlinearity of the price impact. Moreover, we derive an existence result for the optimal strategy beyond this condition when the underlying probability space is countable. In addition, we introduce a novel iterative scheme and establish its convergence to the optimal trading strategy. Finally, we provide a numerical implementation of the scheme that illustrates its convergence, stability, and the effects of concavity on optimal execution strategies under exponential and power-law decay.
\end{abstract}

\begin{description}
\item[Mathematics Subject Classification (2010):] 93E20, 60H30, 91G80
\item[JEL Classification:] C02, C61, G11
\item[Keywords:] Optimal trading, nonlinear market impact, propagator model, power-law decay, square root-law, Fredholm equations, nonlinear stochastic control 
\end{description}

 
 \section{Introduction}

Price impact refers to the empirical observation that executing a large order adversely affects the price of a risky asset in a persistent manner, leading to less favorable execution prices. Consequently, an agent seeking to liquidate a large order, known as a \emph{metaorder}, must split it into smaller parts, referred to as \emph{child orders}, which are typically executed over a period of hours or days. A fundamental question in this context concerns the impact of a metaorder as a function of its size. A large body of empirical and theoretical work suggests that this relationship is concave and in fact, well approximated by a square-root law, see e.g.~\cite{donier2015fully,durin2023two,PhysRevE.89.042805,Sato_24,TothEtAl,Zarinelli_15} and Chapter 12.3 of \cite{bouchaud_bonart_donier_gould_2018} for a comprehensive introduction.   The square-root price impact law states that the peak impact \( D^{\textrm{peak}} \) for a metaorder of volume \( X \), executed over a time horizon \( T \), is given by  
\begin{equation} \label{sqrt}  
D^{\textrm{peak}} = Y\sigma_T \left(\frac{X}{V_T} \right)^\delta, \quad \textrm{for } X \ll V_T,  
\end{equation}  
where \( Y \) is a constant of order 1, \( V_T \) is the total traded volume in the market, and \( \sigma_T \) is the asset's contemporaneous volatility, both measured over time \( T \). The scaling exponent \( \delta \) is estimated to be in the range \( [0.4, 0.7] \), and is often taken as $0.5$ as an approximation.  
The square-root law has been shown to be a good approximation in various markets (such as equities, foreign exchange, options, and even cryptocurrencies) and is general enough to include different types of market liquidity, broad classes of execution strategies, and various trading frequencies (\cite{Bershova_13,Donier_15,Toth_17,Zarinelli_15}).

While the square-root law provides a simple connection between the metaorder size and its impact on the mid-price, price impact evolves dynamically, and the reaction of the mid-price to the metaorder is mostly transient.  
Propagator models are a central tool for mathematically describing this decay phenomenon. They express price moves in terms of the influence of past trades and therefore capture the decay of price impact after each trade (\cite{bouchaud_bonart_donier_gould_2018, gatheral2010no}).  
The price distortion \( D_t \) is quantified by
\be \label{d_t}  
D^X_t  =\int_0^t G(t,s) dX_s, \quad t \geq 0,  
\ee  
where \( X_t \) describes the amount of shares of the asset executed by time \( t \) and \( G(t,s) \) is a positive semi-definite Volterra kernel, often referred to as a \emph{propagator}.  
\citet{obizhaeva2013optimal} and \citet{garleanu2016dynamic}, along with follow-up papers, assume that the price distortion \( D_t \) decays exponentially over time, that is, $G(t,s) =\kappa e^{-\rho(t-s)} \mathbbm{1}_{\{ t > s \}}$ for some positive constants $\kappa$ and $\rho$.  
On the other hand, \citet{bouchaud_bonart_donier_gould_2018} (see Chapter 13.2.1 and references therein) report on empirical observations that the propagator \( G \) decays as a power-law of the lag, i.e.,  
\begin{equation}\label{power-prop} 
G(t,s) \approx (t-s)^{-\beta} \mathbbm{1}_{\{ t > s \}}, \quad t,  s \geq 0,
\end{equation}
with \( 0<\beta < 1 \), and these results are also supported by theoretical arguments.  
 
While the power-law behavior of price impact decay is strongly supported by market data, results on corresponding optimal liquidation strategies have been scarce. 
Power-law or other general propagators that do not exhibit exponential decay introduce non-Markovianity in the corresponding optimal trading problem, thereby rendering standard stochastic control tools, such as dynamic programming or FBSDEs, inapplicable. Recent progress  has been made in the context of  optimal execution with the linear propagator model, see \cite{GSS} for the case without signals and  \cite{abi2022optimal,abijaber2023equilibrium,abijaberOptPort2024} for the case with stochastic signals. However, incorporating both the nonlinearity arising from the square-root law in \eqref{sqrt} and the power-law decay in \eqref{power-prop} remains a long-standing open problem.

As a first step in merging these models, \citet{alfonsi2010optimal} introduced a nonlinear propagator model, motivated by the local shape density of the limit order book and its reaction to market orders. In their model, the price impact \( I_t \) of a metaorder \( (X_t)_{t\geq 0} \) is given by  
\begin{equation} \label{imp}  
I^X_t= h(D^X_t), \quad t\geq 0,  
\end{equation}
where \( h \) is an increasing, nonnegative, and concave function, and \( D^X \) is defined in \eqref{d_t}. Note that the case \( h(x)=x \) corresponds to the linear propagator model, while the case \( h(x)= \text{sign}(x)|x|^{\delta} \), with \( \delta \in (0,1) \), is inspired by \eqref{sqrt} (see \cite{bouchaud2009markets,curato2017optimal,hey2023trading}). \citet{alfonsi2010optimal}  then solved the associated execution problem, which involves minimizing the expected execution costs, \begin{equation}  
C(X) = \mathbb{E} \left[ \int_{0}^T  h(D^X_t)dX_t  \right],  
\end{equation}  
in the special case where \( G(\cdot) \) is an exponentially decaying function, i.e., the Markovian case, obtaining closed-form, deterministic optimal strategies.  

\citet{hey2023trading} generalized the framework of \cite{alfonsi2010optimal} by incorporating a predictive trading signal (also referred to as \emph{alpha}) into the price process. From a modeling perspective, this means that the unaffected price process is no longer a martingale. The agent therefore faces a tradeoff between exploiting the signal to make an immediate profit and managing price impact costs. The objective functional then changes to maximizing the expected revenue functional:  
\begin{equation} \label{costs-sig}  
 R(X)=   \mathbb{E} \left[\int_{0}^{T} \left( \alpha_{t} - I_{t}^{X} \right) dX_t \right], 
\end{equation}  
where the signal $(\alpha_t)_{t \geq 0} $ captures the drift term information of the unaffected asset price. In addition, inventory risk and a penalty for terminal inventory can be incorporated into \eqref{costs-sig}. The optimal execution strategy with respect to \eqref{costs-sig} is expected to be adapted to the signal's filtration. \citet{hey2023trading} worked within the Markovian setup and derived a first-order condition that characterizes an implicit connection between price distortion and the signal. A more tractable relation between impact and alpha was then established in the case where the function \( h(x) \) in \eqref{imp} follows a power-law, and the optimal price impact was derived in closed form.  
The model in \cite{hey2023trading} can be regarded as a generalization of the exponentially decaying linear propagator model studied in \cite{lehalle2019incorporating, neuman2022optimal}, for which closed-form solutions were derived.
For additional work on cost misspecification in the setting of \citet{hey2023trading}, we refer to \cite{hey2023cost}, and a similar problem for the general linear propagator model was studied in \cite{neuman2023statistical}.

The main goal of this work is to study the optimal execution problem in which the price impact is a nonlinear (concave) function of the distortion as in \eqref{imp}, and the price distortion decays as in the general propagator model \eqref{d_t}. The special case where $G$ is given by \eqref{power-prop} will be of major importance, as it is as close as possible to empirical studies. From a mathematical perspective, this will generalize the work in \cite{abi2022optimal,abijaberOptPort2024} to nonlinear price impact models, and generalize the Markovian nonlinear model in \citet{hey2023trading} to the non-Markovian case. We provide an outline of our main results below. 

\textbf{Existence, uniqueness, and first-order conditions.}  
To characterize the optimal strategy for the aforementioned nonlinear, non-Markovian optimal execution problem, we use a variational approach to derive the first-order condition, which takes the form of a nonlinear stochastic Fredholm equation with both forward and backward terms. Under sublinear growth conditions on $h$ in \eqref{imp} and a monotonicity condition (see \eqref{eq:definition_monotone}), we prove that there exists a unique solution to this Fredholm equation, which is the maximizer of the gain functional (see Theorem \ref{T:mainnonlinear}).  
We then provide two examples where the above conditions are satisfied, one of which is the well-known case of an exponential propagator studied in \cite{hey2023trading}. Theorem \ref{T:mainnonlinear} can be regarded as a generalization of Theorem 4.2 in \cite{hey2023trading} to the non-Markovian setting of general propagators. Indeed, in Proposition \ref{P:one_exponential_case}, we show that the monotonicity condition \eqref{eq:definition_monotone} reduces to the conditions assumed in \cite{hey2023trading}.  

Since the monotonicity condition is quite challenging to verify, we extend the results of Theorem \ref{T:mainnonlinear} by relaxing it. In Theorem \ref{T:existence_beyond_monotonicity}, we prove an existence result for the solution to the execution problem by replacing the monotonicity condition with an assumption of a countable underlying probability space.  
Moreover, we prove that this solution must satisfy the stochastic Fredholm equation. These assumptions allow us to include a large class of propagators in our models, including power-law decay. Our proof presents a novel approach for obtaining the sequential weak lower semi-continuity of the negative of the gain functional \eqref{eq:gain_functional_alpha_impact_cost_trade_of}, which, together with the proof of coercivity, yields the existence of a solution to the Fredholm equation. 
 
\textbf{Numerical scheme and convergence results.}  
We formulate an iterative scheme to solve the nonlinear stochastic Fredholm equation with both forward and backward terms, which was derived from the first-order condition, in order to determine the optimal trading strategy. In Theorem \ref{pp:conv_scheme}, we prove the convergence of the scheme and derive a bound on its convergence rate. 
Our scheme is designed to handle a broad class of kernels and stochastic signals. Numerical methods for nonlinear propagator models remain scarce, to the best of our knowledge, the only existing approach is \cite{brokmann2024efficient}, which employs neural networks to approximate the solution, but has been implemented only for the exponential kernel. In contrast, our approach guarantees convergence with explicit bounds on the convergence rate, and is also implemented for power-law propagators.

Moreover, we illustrate the convergence of the numerical scheme for both deterministic and stochastic signals and for general kernels. In the case of an exponential kernel and vanishing instantaneous impact, we use the explicit impact solution provided by  \citet{hey2023trading}  as a benchmark (see Section \ref{s:numerics}). Additionally, we provide a numerical implementation that showcases the stability of the numerical scheme when approximating the power-law kernel by sums of exponentials.
    
\paragraph{Organization of the paper.}  
In Section \ref{s:problem_formulation_and_main_results}, we introduce the nonlinear propagator model. We state the main existence theorems, present a proposition on well-posedness in the single-exponential case, and prove two stability results. Additionally, we outline the iterative numerical scheme designed to determine the optimal trading strategy.  
Section \ref{s:numerics} focuses on numerical results. We implement our numerical scheme for various scenarios and discuss its numerical convergence.  
Sections \ref{S:proofofoptimality}--\ref{s:proof_existence_uniqueness_beyond_monotonicity} contain the proofs of our main results. Finally, in the appendix additional numerical results are provided.

\section{Problem formulation and main theoretical results} \label{s:problem_formulation_and_main_results}

Fix a finite time horizon $T > 0$ and a filtered probability space $(\Omega, \mathcal F, (\mathcal F_t)_{t \in [0,T]}, \mathbb{P})$ satisfying the usual conditions. We denote by $\dd t$ the Lebesgue measure  on the Borel $\sigma$-algebra $\mathcal{B} \left( [0,T] \right)$, and we denote by $\dd t \otimes \mathbb{P}$ the product measure on the $\sigma$-algebra $\mathcal{B} \left( [0,T] \right) \otimes \mathcal{F}$. We introduce the standard Banach spaces
\begin{equation*} \label{eq:def_mathcal_L_p}
    \mathcal{L}^{p} := \left\{ f: [0,T] \times \Omega \to \R \; \text{prog.~measurable}, \; \E \left[ \int_{0}^{T}  |f_{t}|^{p} \dd t \right] < \infty \right\},  \quad p \geq 1. 
\end{equation*}
For $p=2$, we equip $\mathcal L^2$ with the inner product
\begin{equation*}
    \langle f, g \rangle := \E \left[ \int_{0}^{T}  f_{t} g_{t}  \dd t \right], \quad f,g \in \mathcal{L}^{2},
\end{equation*}
which makes it a Hilbert space with the associated norm $\| f\| := \sqrt{\langle f,f\rangle}$. For $u \in \mathcal{L}^{1}$ and $t \in [0,T]$, we denote by $\E_t u$ the conditional expectation of $u$ with respect to the $\sigma$-algebra $\mathcal{F}_t$.

We consider an agent who wishes to liquidate a given amount of assets $X_{0}>0$ at an admissible trading rate $u \in \mathcal{L}^{2}$. The agent's inventory $X^u$ is therefore given by 
\begin{equation}\label{eq:running_inventory}
    X^u_t := X_{0}+\int_{0}^{t}u_s \dd s, \quad t\in [0,T]. 
\end{equation}  
Notice in particular that 
\[
    \sup_{t\in [0,T]}\mathbb{E}\left[|X^u_t|^2\right]<\infty.
\]
We fix a process $S=(S_t)_{t\in [0,T]}\in \mathcal{L}^{2}$ with terminal random variable $S_T \in L^2(\left(\Omega, \mathcal{F}_{T}, \mathbb{P} \right), \mathbb{R})$ representing the fundamental price of the risky asset. The execution price $S^u=(S^u_t)_{t\in [0,T]}$ is therefore given by, 
\begin{equation}\label{eq:effective_price}
    S^u_t = S_t + \frac{\gamma}{2} u_t + h(Z^u_t), \quad t \in [0,T],
\end{equation}
where $\gamma>0$ is a constant, $h \colon \mathbb{R} \mapsto \mathbb{R}$ is a measurable function satisfying suitable regularity conditions \emph{to be specified} and  $Z^u=(Z^u_t)_{t\in[0,T]}$ is the process in $\mathcal{L}^{2}$ given by 
\begin{equation} \label{eq:def_Z}
    Z^u_t := g_{t} + \left(\mathbf{G}u\right)_{t}, \quad t \in [0,T]. 
\end{equation}
 Here $g \in \mathcal{L}^{2}$ is an input stochastic process independent of $u$ and $\mathbf{G}: \mathcal{L}^{2} \mapsto \mathcal{L}^{2}$ is the linear operator induced by a Volterra kernel $G \colon [0,T]^2 \to \mathbb{R}$. 
 The terms $\frac{\gamma}{2} u$ and $Z^u-g$ in \eqref{eq:effective_price} are, respectively, the instantaneous and transient price impacts of the trading strategy $u$ on the unaffected price process $S$. Financially, $\gamma$ captures slippage costs, i.e., bid-ask spreads. Furthermore, $h$ is the impact function applying a potentially nonlinear effect to the price distortion $\mathbf{G}u$ caused by the trading strategy $u$. When $h \equiv \text{Id}$, our setting reduces to the one in \cite{abi2022optimal,abi2024trading, abijaberOptPort2024}, where linear transient price impact is investigated. We discuss the financial interpretation of $g$ below \eqref{eq:def_impact_costs}. 

Next we present the class of admissible price impact kernels, which are also known as propagators. 
\begin{definition}\label{def:admissible_kernel}
    A Volterra kernel $G \colon [0, T]^2 \to \mathbb{R}$ is said to be \emph{admissible} if
    \begin{equation} \label{eq:constant_norm_of_G}
        C_{G} :=
        \sup_{t\in [0, T]} \int_{0}^{t}|G(t,s)|^2\dd s < \infty.
    \end{equation}
    Any admissible kernel $G$ induces a unique linear and bounded integral operator $\mathbf{G}: \mathcal{L}^{2} \mapsto \mathcal{L}^{2}$ defined by
    \begin{equation*}
        \left(\mathbf{G}u\right)_{t} := \int_{0}^{t} G(t,s) u_{s} \dd s, \quad t \in [0,T], \quad u \in \mathcal{L}^{2},
    \end{equation*}
    such that, by Cauchy-Schwarz's inequality, 
    \begin{equation} \label{eq:estimate_on_admissible_G_operator}
        \| \mathbf{G}u \|^2 \le T C_{G} \| u \|^2, \quad u \in \mathcal{L}^2.
    \end{equation}
    As a consequence of Fubini's theorem and the tower property of the conditional expectation, the unique linear and bounded adjoint operator $\mathbf{G}^{\ast}: \mathcal{L}^{2} \to \mathcal{L}^{2}$ of $\mathbf{G}$ is explicitly given by
    \begin{equation} \label{eq:dual_operator_def}
        \left(\mathbf{G}^{\ast}u\right)_{t} := \int_{t}^{T} G(s,t) \E_{t}[ u_{s} ]\dd s, \quad t \in [0,T], \quad u \in \mathcal{L}^{2},
    \end{equation}
    where $\E_t$ denotes conditional expectation with respect to $\mathcal F_t$. 
     The boundedness of $\mathbf{G}^\ast$ is ensured by conditional Jensen's inequality and Fubini's theorem such that
    \begin{equation} \label{eq:estimate_on_adjoint}
        \| \mathbf{G}^{\ast}u \|^2 \leq T C_{G} \| u \|^2,\quad u \in \mathcal{L}^2,
    \end{equation}
    see for example \cite[Chapter 6, Section 2]{dunford1988linear}.
\end{definition}
We introduce the following revenue-risk functional of the agent,  
\begin{equation} \label{eq:gain_functional}
    \mathcal{J}(u) := \mathbb{E}\bigg[- \int_{0}^{T}S_t^u u_t \dd t+X^u_T S_T - \frac{\phi}{2} \int_{0}^{T}(X^u_t)^2 \dd t - \frac{\varrho}{2} (X^u_T)^2\bigg], \quad u\in \mathcal{L}^{2},
\end{equation}
where $\phi$ and $\varrho$ are nonnegative constants representing the penalty factor on the risk aversion term and on the terminal inventory, respectively. In order for $\mathcal{J}$ in \eqref{eq:gain_functional} to be well-defined we require that $S^u \in \mathcal{L}^{2}$ for every $u \in  \mathcal{L}^{2}$. This is satisfied by enforcing the a  sublinear growth condition on $h$ (see see Definition \ref{def:admissible_impact_function}).
Our objective is to maximize the performance functional $\mathcal{J}$, that is, to find an optimal trading rate $\hat{u} \in \mathcal{L}^{2}$ such that \begin{equation} \label{eq:optimal_strategy}
    \mathcal{J}(\hat{u})= \sup_{u \in \mathcal{L}^{2}} \mathcal{J}(u).
\end{equation}
Inserting the definitions of the final inventory $X_{T}^{u}$ from \eqref{eq:running_inventory} and the effective price $S^{u}$ from \eqref{eq:effective_price} into $\mathcal{J}$ from \eqref{eq:gain_functional}, while using the tower property of the conditional expectation, allows us to rewrite the performance functional as follows
\begin{equation} \label{eq:gain_functional_alpha_impact_cost_trade_of}
    \mathcal{J}(u) = \mathbb{E}\bigg[\int_{0}^{T} \left( \alpha_{t} - I_{t}^{u} \right) u_t \dd t - \frac{\phi}{2} \int_{0}^{T}(X^u_t)^2 \dd t - \frac{\varrho}{2} (X^u_T)^2 \bigg] + X_0 \mathbb{E} \left[ S_{T} \right], \quad u \in \mathcal{L}^{2}.
\end{equation}
This expression of $\mathcal{J}$ features the standard tradeoff between the exploitation of the asset's alpha-signal, 
\begin{equation}\label{eq:def_alpha_signal}
    \alpha_t := \mathbb E_t[S_T-S_t],\quad 0\leq t \leq T,
\end{equation}
and the price impact, 
\begin{equation}\label{eq:def_impact_costs}
    I_{t}^{u} := \frac{\gamma}{2} u_{t} + h \left( Z^u_t \right),  0\leq t \leq T,
\end{equation}
see for example  \cite[Lemma~5.3]{abi2024trading}, \cite[Equation (3.2)]{hey2023trading}.

\textbf{Financial interpretation.}  
The stochastic process $\alpha$ in \eqref{eq:def_alpha_signal} captures the drift of the fundamental price $S$, and it is interpreted as a short-term trading signal (see \cite{lehalle2019incorporating}). The process $g$ in \eqref{eq:def_Z} provides additional flexibility to the model, as it can serve as a proxy for the aggregated price impact generated by the order flow of other market participants trading the same asset.  
It is shown in Theorem 2.13 of \cite{N-V-23}, using game-theoretical arguments, that for $h \equiv \text{Id}$ the aggregated price impact of other players serves as an additional signal competing with $\alpha$. In the case of a  concave $h$ the resulting effective signal is given by $\tilde{\alpha} = \alpha - h(g+\mathbf{G}u)$.

\begin{remark}[Flexibility of the model specification] \label{r:model_flexibility}
    Our model formulation includes some well-known examples for nonlinear price impact models. In particular, it accommodates the model introduced in \cite{alfonsi2010optimal} in the special case where
    \begin{equation}\label{eq:expkern}
        G(t,s) := \xi e^{-x (t-s)} \mathbbm{1}_{\{ t > s \}}, \quad \xi, \; x>0,
    \end{equation}
    as well as the case of nonlinear permanent price impact by setting $G$ to a constant.  Our framework also includes the cases where $G$ is a sum of multiple decaying  exponentials as in \cite{hey2023trading}, or has a power-law decay as illustrated in Section \ref{s:numerics}.
    Finally, taking  the exponential kernel \eqref{eq:expkern} with $\xi = x$, and taking $\gamma \to 0$ and $x \to \infty$ yields a nonlinear instantaneous price impact such that \eqref{eq:def_impact_costs} has the form
    \begin{equation*}
        I_{t}^{u} = h \left( g_t + u_{t} \right),
    \end{equation*}
    which is reminiscent of the impact models in \cite{kolm2014multiperiod,kyle2011market}. 
\end{remark}

\subsection{Main results on existence of the optimal strategy}
Since the performance functional \eqref{eq:gain_functional_alpha_impact_cost_trade_of} is nonlinear and not necessarily convex, one of the main fundamental questions addressed in this section deals with existence and characterization of the optimal strategy. We start by introducing the set of admissible impact functions.
  
\begin{definition} \label{def:admissible_impact_function}
An impact function $h \colon \mathbb{R} \mapsto \mathbb{R}_+$ is said to be admissible if it satisfies the following conditions:
\begin{enumerate}
    \item[\textbf{(i)}] $h$ is differentiable with bounded derivative $h'$, 
    \item[\textbf{(ii)}] $h$ grows strictly sublinearly, i.e., there exists $0 < \zeta < 1$ such that
    \begin{equation*}
        |h(x)| = \mathcal{O} \left( 1+|x|^\zeta \right), \quad x \in \R.
    \end{equation*}
\end{enumerate}
\end{definition}

Given an admissible impact function $h$, using Definition \ref{def:admissible_impact_function}(i), we can define an operator $\mathbf{A} \colon \mathcal{L}^{2} \mapsto \mathcal{L}^{2}$ such that, for every $u\in\mathcal{L}^{2}$ and $t \in [0,T]$,
\begin{equation}\label{eq:def_A_operator}
\begin{aligned}
\!\!\mathbf{A}(u)_{t} & := h(Z^u_t) + (\mathbf{G}^\ast (h'(Z^u)u))_{t}
 \\
&= h \left(g_{t} + \int_0^t G(t,s) u_s \dd s \right) +   \int_{t}^T G(s,t) \E_t\bigg[h'\left({ g_{s} + } \int_0^{s} G(s,r)u_r\dd r\right) u_s \bigg] \dd s,
\end{aligned}
\end{equation}
where  $Z^u$ is defined in \eqref{eq:def_Z}. The following monotonicity condition will play a crucial role in the proof of the existence of a maximizer to \eqref{eq:gain_functional_alpha_impact_cost_trade_of}, 
\begin{equation}\label{eq:definition_monotone}
    \langle u-v,\mathbf A(u)-\mathbf{A}(v) \rangle\ge0,\quad u,\,v\in\mathcal{L}^{2},
\end{equation} 
as it ensures the strong concavity of the functional $\mathcal{J}$, see Lemma \ref{thm:conc_conditions}.

\begin{remark}[Assumptions on the impact function $h$] \label{rem_conditions_h}
    We briefly discuss assumptions (i) and (ii) in Definition \ref{def:admissible_impact_function}. (i) allows for a rigorous G\^ateaux differentiation of the performance functional $\mathcal{J}$ in \eqref{eq:gain_functional} and ensures that the nonlinear operator $\mathbf{A}$ in \eqref{eq:def_A_operator} is well-defined. (ii) ensures that the performance functional is well-defined in $\mathcal{L}^{2}$ and also guarantees  the connectivity of the performance functional, which is needed in order to obtain existence of an optimal trading strategy. 
\end{remark}

\begin{remark}[Linear transient impact]
Note that the case of linear transient price impact, i.e., $h(x) = x$,  \eqref{eq:definition_monotone} is equivalent to requiring the linear integral operator $\mathbf{G}$ in $\mathcal{L}^{2}$ to be positive semi-definite in the following sense, that is, 
    \begin{equation} \label{eq:sdp_operator_def}
        \langle u,\mathbf{G}u \rangle \ge 0, \quad u \in \mathcal{L}^{2} \iff \langle u,\left(\mathbf{G} + \mathbf{G}^{\ast}\right)u \rangle \ge 0, \quad u \in \mathcal{L}^{2},
    \end{equation}
    where the equivalence is a consequence of Fubini's Theorem. Indeed, from  \eqref{eq:def_Z} and \eqref{eq:def_A_operator} it follows that, 
	\begin{equation*}
		\langle u-v,\mathbf A (u)-\mathbf{A}(v) \rangle =
		\langle u-v, \left( \mathbf{G} + \mathbf{G}^{\ast} \right) \left( u - v \right) \rangle.
	\end{equation*}
	Such assumption on $\mathbf G$ is standard to ensure the well-posedness of linear price impact problems and has already appeared in \cite{abijaber2023equilibrium,abi2024trading,abijaberOptPort2024,GSS}.
\end{remark}
Moreover, we introduce the admissible kernel
\begin{equation}\label{eq:kernel_H}
H_{\phi,\varrho}(t,s) := \left(\phi(T-t)+\varrho\right) \mathbbm{1}_{\{ t > s \}}, \quad s,t \in [0,T],
\end{equation}
which captures both the running and terminal inventory soft-penalization terms in the performance functional  \eqref{eq:gain_functional}. In particular, the corresponding operator $\mathbf{H}_{\phi,\varrho}$ in $\mathcal{L}^2$ is positive semi-definite in the sense of \eqref{eq:sdp_operator_def} as shown in Lemma \ref{thm:conc_conditions}.

Our first main result characterizes the optimal strategy in terms of the first-order condition (FOC) in the Gâteaux derivative sense.

\begin{theorem}\label{T:mainnonlinear} Let $h$ be an admissible function as in Definition \ref{def:admissible_impact_function}. Assume that the operator $\mathbf{A}$ in \eqref{eq:def_A_operator} satisfies the monotonicity condition \eqref{eq:definition_monotone}. Then, there exists a unique admissible optimal control $\hat{u} \in \mathcal{L}^{2}$ satisfying \eqref{eq:optimal_strategy}. In particular, the optimal control $\hat{u}$ is the unique $\mathcal{L}^{2}-$valued solution to the nonlinear stochastic Fredholm equation
\begin{equation}\label{eq:nonlinearfredholm}
    \gamma \hat{u}_t + \left(\mathbf A(\hat{u})\right)_{t} + \left( \mathbf{H}_{\phi,\varrho}\hat{u} \right)_{t}
    + \left( \mathbf{H}_{\phi,\varrho}^{*}\hat{u} \right)_{t} = \alpha_t - X_{0} \left( \phi(T-t) + \varrho \right), \quad \left( \dd t \otimes \mathbb{P} \right)-a.e.,
\end{equation}
where $\alpha$ and $H_{\phi,\varrho}$ are defined in \eqref{eq:def_alpha_signal} and \eqref{eq:kernel_H}, respectively.

\end{theorem}
The proof of Theorem \ref{T:mainnonlinear} is given in Section~\ref{S:proofofoptimality}.

\begin{remark}
Assuming the monotonicity condition \eqref{eq:definition_monotone}, Theorem~\ref{T:mainnonlinear} shows that computing the optimal control reduces to solving the nonlinear stochastic Fredholm equation \eqref{eq:nonlinearfredholm}. In general, such equation does not admit explicit solutions, except for the following two limit cases.
\begin{enumerate}
    \item [\textbf{(i)}] The linear impact case $h(x)=x$. Here, \eqref{eq:nonlinearfredholm} reduces to a linear stochastic Fredholm equation which has recently been explicitly solved  in \cite{abijaber2023equilibrium} in terms of operator resolvents.
    \item [\textbf{(ii)}] The case where the slippage costs $\gamma \downarrow 0$, the coefficients $\phi=\varrho=0$, the process $g\equiv0$ and the kernel $G$ is the exponential (see \eqref{eq:expkern}). Then \eqref{eq:nonlinearfredholm} reduces to 
    \begin{equation*}
        h \left(Z_t \right) +  \mathbb E_t\bigg[\int_t^T  G(s,t) h'(Z_s) \dd X_s\bigg] = \alpha_t, \quad t \in [0,T]. 
    \end{equation*}
   We recall that from \eqref{eq:def_Z} we get $Z_t=\int_0^t G(t,s)\dd X_s$. Then from Proposition \ref{P:one_exponential_case} below it follows that $\mathbf{A}$ is monotone. 
    In this setting, one would not expect $X$ to be absolutely continuous with respect to the Lebesgue measure. In fact, for suitable increasing odd impact functions $h \in C^2$ that are concave in $\mathbb{R}_+$, explicit solutions for the volume impact $Z$ are derived in \cite[Corollary 4.4]{hey2023cost} with the presence of bulk trades at $t=0$ and $t=T$, corresponding to Dirac measure at the extremities of the trading horizon.
\end{enumerate}
\end{remark}

Given an admissible impact function $h$, an admissible kernel $G$ and a process $g \in \mathcal{L}^{2}$, it turns out that verifying that the operator $\mathbf{A}$ is indeed monotone in the sense of \eqref{eq:definition_monotone} can be very challenging. The following proposition, which will be proved in Section \ref{s:sufficient_conditions_A_monotone}, provides sufficient conditions on $h$ such that $\mathbf{A}$ is monotone.  
\begin{proposition} \label{P:one_exponential_case}
    Let $h$ be differentiable with bounded derivative, and consider the exponential Volterra kernel 
    \begin{equation} \label{eq:exponential_kernel_def}
        G(t,s) = \mathbbm{1}_{\{t\ge s\}}ae^{-b(t-s)}, \quad t,s\in [0,T], \quad a>0, \; b \ge 0.
    \end{equation}
  Assume that the following hold: 
    \begin{enumerate}
        \item[\textbf{(i)}] $h$ is nondecreasing on $\mathbb{R}$,
        \item[\textbf{(ii)}] $g \equiv 0$,
        \item[\textbf{(iii)}] $x \mapsto xh'(x)$ is nondecreasing on $\mathbb{R}$.
    \end{enumerate}
    Then $\mathbf{A}$ satisfies the monotonicity property \eqref{eq:definition_monotone}.
    
\end{proposition}
\begin{remark}
If $h$ is  twice differentiable and $h'>0$ in $\mathbb{R} \setminus A$, for some discrete subset $A\subset \mathbb{R}$, then condition (iii) is equivalent to 
    \begin{equation} \label{eq:arrow_pratt_ratio_condition}
      -x\frac{h''(x)}{h'(x)}\le 1,\quad x\in \mathbb{R}\setminus A. 
    \end{equation}
Note that \eqref{eq:arrow_pratt_ratio_condition} coincides with the sufficient condition on concavity given in Section 4.1 of \cite {hey2023trading}.
\end{remark} 

Since the monotonicity condition \eqref{eq:definition_monotone} is quite difficult to verify, we would like to bypass it and still derive the existence of an optimal strategy to the execution problem \eqref{eq:optimal_strategy}. In the following theorem, which is proved in Section \ref{s:proof_existence_beyond_monotonicity}, we establish such result  when considering a countable underlying probability space $\Omega$.

\begin{theorem} \label{T:existence_beyond_monotonicity}
    Assume that $(\Omega, 2^{\Omega},\mathbb{P})$ is a countable probability space, where $2^{\Omega}$ denotes the power set of $\Omega$. Assume that $h : \mathbb{R}\mapsto\mathbb{R}$ is a globally  Lipschitz continuous function of sublinear growth as in Definition \ref{def:admissible_impact_function}(ii).
    Then, there exists a global maximizer $\hat{u} \in \mathcal{L}^{2}$ for the performance functional $\mathcal{J}$ from \eqref{eq:gain_functional_alpha_impact_cost_trade_of}. If, in addition, $h$ satisfies Definition \ref{def:admissible_impact_function}(i), then $\hat{u}$ is a solution to the nonlinear stochastic Fredholm equation \eqref{eq:nonlinearfredholm}.
\end{theorem}
 
Finally, in order to achieve uniqueness of the solution beyond the monotonicity condition \eqref{eq:definition_monotone}, we introduce the space 
\begin{equation*}
    \mathcal{L}^{\infty} \left( \Omega, L^{2}([0,T]) \right) := \left\{ f: [0,T] \times \Omega \to \R \; \text{prog.~measurable}, \; \esssup_{\omega \in \Omega} \norm{f(\omega)}_{L^{2}} < \infty \right\},
\end{equation*}
and state the following theorem, which is proved in Section \ref{s:proof_existence_uniqueness_beyond_monotonicity}.

\begin{theorem} \label{T:existence_uniqueness_beyond_monotonicity}
Suppose that the assumptions of Theorem \ref{T:existence_beyond_monotonicity} are satisfied, and that 
\begin{equation} \label{eq:alpha_infinity}
    \alpha,\,g\in\mathcal{L}^{\infty} ( \Omega, L^{2}([0,T])).
\end{equation}
Moreover, let $h\colon\mathbb{R}\mapsto\mathbb{R}$ be differentiable, with bounded and Lipschitz continuous derivative, and assume that $G$ is positive semi-definite and satisfies 
\[
\sup_{t\in[0,T]}\int_t^T|G(s,t)|^2ds<\infty.
\]
Then, if the slippage parameter $\gamma >0$ is sufficiently large, there exists a unique optimal trading strategy $\hat{u}\in\mathcal{L}^2$ for $\mathcal{J}$ satisfying \eqref{eq:optimal_strategy}.    In particular, $\hat{u}\in \mathcal{L}^{\infty} \left( \Omega, L^{2}([0,T]) \right)$ is the unique solution of the nonlinear stochastic Fredholm equation \eqref{eq:nonlinearfredholm} in $\mathcal{L}^2$.
\end{theorem}

\begin{remark} 
    The countability assumption on the underlying probability space in Theorem \ref{T:existence_beyond_monotonicity} essentially means that the existence of an optimal trading strategy continues to hold for real-world tactical discrete-time trading activities, where states of the world are well captured within such framework. The additional assumption in Theorem \ref{T:existence_uniqueness_beyond_monotonicity} of a sufficiently large slippage parameter with respect to the the trading horizon and other impact model parameters ensures the regularization of the Fredholm equation \eqref{eq:nonlinearfredholm} (see Section \ref{s:proof_existence_uniqueness_beyond_monotonicity} for additional details). 
\end{remark}

\subsection{Stability results for the optimal control}
 In this section, we establish stability results for the optimal trading strategy $\hat{u} \in \mathcal{L}^{2}$ satisfying \eqref{eq:optimal_strategy}, with respect to the transient impact kernel $G$ and the input signals $\alpha$ and $g$. Specifically, the stability properties of the control problem with respect to these input parameters, as stated in Propositions~\ref{pp:stability_with_respect_to_kernel}--\ref{pp:stability_with_respect_to_signals}, provide a theoretical foundation for our numerical implementation in Section~\ref{s:numerics}. There, the power-law decay is approximated by weighted sums of exponential decays, and stochastic signals are discretized on the probability space. The proofs of these results are deferred to Section~\ref{s:proofs_stability_results}.

\begin{proposition} \label{pp:stability_with_respect_to_kernel}
    Let $h$ be differentiable with bounded derivative. For every $n\in \mathbb{N}$, let $G_n$ be an admissible kernel as in Definition \ref{def:admissible_kernel} and define the operator $\mathbf{A}_n $ (resp.~the functional $\mathcal{J}_n$) as in \eqref{eq:def_A_operator} (resp.~\eqref{eq:gain_functional}), with $G_n$ instead of $G$. Suppose that $\mathbf{A}$ and $\big(\mathbf{A}_n\big)_{n\in\mathbb{N}}$ are monotone in the sense of \eqref{eq:definition_monotone}, and that 
    \begin{equation} \label{eq:condition_stability}
        \sup_{t\in [0,T]} \int_{0}^{t}\left|G_n(t,s)-G(t,s)\right|^2\dd s = \emph{o}(1), \quad \text{as } n \to \infty.
    \end{equation}
    Denote by  $\hat{u}_n$ (resp.~$\hat{u}$) the optimal strategies maximizing $\mathcal{J}_n$ (resp.~$\mathcal{J}$). Then the sequence $(\hat{u}_n)_{n \in \mathbb{N}}$ is bounded in $\mathcal{L}^{2}$. If, additionally, the derivative of $h$ is Lipschitz continuous and 
    \begin{equation} \label{eq:assumption_on_optimal_strategy_for_stability}
        \mathbb{E} \Bigg[ \bigg(\int_{0}^{T} \hat{u}_{t}^{2} \dd t \bigg)^{2} \Bigg] < \infty,
    \end{equation}
    then $\lim_{n\to\infty} \hat{u}_n=\hat{u}$ in $\mathcal{L}^{2}$.
\end{proposition}

\begin{remark}
    Consider the space
    \begin{equation*}
        \mathcal{L}^{4} \left( \Omega, L^{2}([0,T]) \right) := \left\{ f: [0,T] \times \Omega \to \R \; \text{prog.~measurable}, \; \mathbb{E}\Big[\norm{f}^4_{L^{2}}\Big] < \infty \right\}.
    \end{equation*}
    Then, an application of the estimate shown in Lemma \ref{a_priori_lemma}, along with the tower property and conditional Jensen's inequality, yields that any solution $\hat{u}$ of \eqref{eq:nonlinearfredholm} belongs to $\mathcal{L}^{4} \left( \Omega, L^{2}([0,T]) \right)$, i.e., assumption \eqref{eq:assumption_on_optimal_strategy_for_stability} is satisfied, whenever $\alpha,\,g\in \mathcal{L}^{4} \left( \Omega, L^{2}([0,T]) \right)$.
\end{remark}

\begin{proposition} \label{pp:stability_with_respect_to_signals}
    Assume that $h$ is Lipschitz continuous and satisfies  Definition \ref{def:admissible_impact_function}(ii). For every $n\in\mathbb{N}$, consider an alpha-signal $\alpha_n\in\mathcal{L}^{2}$ as in \eqref{eq:def_alpha_signal} and a process  $g_n\in\mathcal{L}^2$ as in \eqref{eq:def_Z}. Define the corresponding functional $\mathcal{J}_n $ as in \eqref{eq:gain_functional} with $\alpha_n$ and $g_n$ instead of $\alpha$ and $g$. Suppose that 
    \begin{equation*}
        \|\alpha_n-\alpha\|=\emph{o}(1), \quad\text{and}\quad \|g_n-g\|=\emph{o}(1),\quad \text{as } n \to \infty.
    \end{equation*}
    Then it holds that
    \begin{equation*}
        \sup_{u\in\mathcal{L}^2} \mathcal{J}_n(u) \to \sup_{u\in\mathcal{L}^2} \mathcal{J}(u),\quad \text{as } n \to \infty.
    \end{equation*}
\end{proposition}

\subsection{An iterative numerical scheme}

In this subsection, we propose an iterative scheme which leverages the explicit solutions of linear stochastic Fredholm equations in order to derive a solution to the non-linear Fredholm equation \eqref{eq:nonlinearfredholm}. Recall that by Theorem \ref{T:mainnonlinear}, when the $\mathcal{L}^2$-operator $\mathbf{A}$ is monotone in the sense of \eqref{eq:definition_monotone}, there exists a unique solution to \eqref{eq:nonlinearfredholm} which is the optimal trading strategy with respect to \eqref{eq:optimal_strategy}. To execute our scheme, we subtract from the operator $\mathbf A$ the linear price impact,  
\begin{equation} \label{eq:Atilde}
    \big(\tilde{\mathbf{A}}(u)\big)_{t} := \left(\mathbf{A}(u)\right)_{t}  - \left(\mathbf{G}u\right)_{t} - \left(\mathbf{G}^{*}u\right)_{t}, \quad t \in [0,T], \quad u \in \mathcal{L}^{2},   
\end{equation}
so that  \eqref{eq:nonlinearfredholm}
becomes for for all $t \in [0,T]$, 
\begin{equation} \label{eq:nonlinear_fredholm_for_scheme}
    \gamma u_t + \left(\left(\mathbf{G} + \mathbf{H}_{\phi,\varrho}\right)u\right)_{t} + \left(\left(\mathbf{G} + \mathbf{H}_{\phi,\varrho} \right)^{*}u\right)_{t} = \alpha_t - \big(\tilde{\mathbf A}(u)\big)_{t} - X_{0} \left(\phi(T-t)+\varrho\right).
\end{equation}
The iterative scheme is defined as follows:
\begin{enumerate}
   	\item [\textbf{(i)}] \textbf{Initialization:}
		\begin{equation}\label{scheme_initialization}
			u^{[0]}_t := 0, \quad t\in [0,T].
		\end{equation}
    \item [\textbf{(ii)}]  \textbf{Update:} for $n\geq 1$, having  $u^{[n-1]}$, we derive $u^{[n]}$ using \eqref{eq:nonlinear_fredholm_for_scheme} but fixing the nonlinear term to $\tilde{\mathbf A} (u^{[n-1]})$, that is we solve the linear Fredholm equation, 
    \begin{equation} \label{scheme_update}
        \gamma u^{[n]} + \left(\mathbf{G} + \mathbf{H}_{\phi,\varrho}\right)u^{[n]} + \left(\mathbf{G} + \mathbf{H}_{\phi,\varrho} \right)^{*}u^{[n]} = Y[n-1],
    \end{equation}
    where we define the source-term $Y$ by
    \begin{equation} \label{eq:source_term_Fredholm_scheme}
        Y_{t}[n-1] := \alpha_t - \tilde{\mathbf A}(u^{[n-1]})_t - X_{0} \left(\phi(T-t) + \varrho\right), \quad t \in [0,T].
    \end{equation}
\end{enumerate}
In the following proposition we establish that the iterative scheme \eqref{scheme_initialization}-\eqref{eq:source_term_Fredholm_scheme} converges to the solution of \eqref{eq:nonlinearfredholm} under suitable conditions. The proof is postponed to Section \ref{ss:proof_conv_scheme}.

\begin{proposition}\label{pp:conv_scheme}
    Let $h$ be differentiable with bounded and Lipschitz continuous derivative with a Lipschitz constant $L>0$, and let the operator $\mathbf{G}$ induced by the admissible kernel $G$ be positive semi-definite. 
    Denote by $(u^{[n]})_{n \geq 0}$ the iterations of the scheme \eqref{scheme_initialization}-\eqref{eq:source_term_Fredholm_scheme}. Suppose that there exists a solution $\hat{u}\in \mathcal{L}^{2}$ of \eqref{eq:nonlinearfredholm} such that 
    \begin{equation} \label{eq:definition_ess_sup}
        M_{\gamma}(\hat{u}) := \esssup_{\omega \in \Omega} \int_{0}^{T} |\hat{u}_{t}(\omega)|^{2} \dd t<\infty
    \end{equation}
    and 
    \begin{equation} \label{eq:common_ratio_geometric_sequence}
      \widetilde{C} := 2\sqrt{T C_{G}}\Big(1 + \norm{h'}_\infty + \frac{L}{2} \sqrt{C_{G} M_{\gamma}(\hat{u})} \Big) <\gamma,
    \end{equation} 
    where $C_{G}$ is the constant defined by \eqref{eq:constant_norm_of_G}. Then $$\lim_{n\to\infty} u^{[n]} = \hat{u}\quad \text{in $\mathcal{L}^2$.}$$  Moreover, the convergence rate is bounded by, 
    \begin{equation} \label{eq:convergence_rate}
      \norm{u^{[n]}-\hat{u}}  \le \bigg(\frac{\widetilde{C}}{\gamma}\bigg)^n\norm{\hat{u}}, \quad n \in \mathbb{N}.
    \end{equation}
\end{proposition}
 As a by-product of Proposition \ref{pp:conv_scheme}, there exists at most one solution of \eqref{eq:nonlinearfredholm} that satisfies \eqref{eq:definition_ess_sup} and \eqref{eq:common_ratio_geometric_sequence}. Indeed, if two such solutions exist, then the same sequence $(u^{[n]})_{n\geq 0}$ converges to both solutions, and by uniqueness of the limit, they must be equal.
In the following section, we illustrate the performance of our numerical scheme for various propagators and market scenarios. 

\section{Numerical illustrations} \label{s:numerics}

In this section, we first explain how to implement the numerical scheme \eqref{scheme_initialization}--\eqref{scheme_update} in practice, leveraging the Nyström approximation of linear operators and Least-Squares Monte Carlo (LSMC) to solve at each step the resulting linear stochastic Fredholm equation in \eqref{scheme_update} and how to quantify the numerical error of the scheme. Then, we specify an impact function that satisfies our assumptions, along with signal processes, to illustrate our theoretical findings, including:
\begin{itemize}
    \item[(i)] the convergence of the numerical scheme to an explicit asymptotic solution when the impact function is concave and when $\gamma$ goes to zero;
    
    \item[(ii)] the stability result from Proposition~\ref{pp:stability_with_respect_to_kernel} when approximating the fractional kernel by a sum of exponential decays;
    
    \item[(iii)] the impact of concavity and the comparison between exponential and power-law decay on optimal trading in presence of a ``buy'' signal.
\end{itemize}

\subsection{The iterative scheme in practice} 

Let $n \in \mathbb{N}^{*}$ and consider the scheme \eqref{scheme_initialization}--\eqref{scheme_update}. Given $u^{[n-1]}$, $u^{[n]}$ solves the linear stochastic Fredholm equation \eqref{scheme_update} and is obtained by using the explicit operator formula in \cite[Proposition 5.1]{abijaber2023equilibrium} in terms of $Y[n-1]$ from \eqref{eq:source_term_Fredholm_scheme} and its conditional expectations $\left(\E_{t}Y_{s}[n-1]\right)_{0 \le t < s \le T}$. The only step left is estimating numerically such quantities at each step of the scheme. By the linearity of the conditional expectation, this amounts to estimating for $0 < i < j \leq N-1$
\begin{align} \label{eq:conditional_expectations_to_estimate}
    E_{i,j}[n-1] := & \; \mathbb{E}_{t_{i}} \left[ \widetilde{Y}_{t_{j}}[n-1] \right],\\ \notag
    \; \widetilde{Y}_{t_{j}}[n-1] = & \; \alpha_{t_{j}} - h\left( g_{t_{j}} + \left(\mathbf{G}u^{[n-1]}\right)_{t_{j}} \right) - \mathbf{G}^{\ast} \left( h' \left( g + \mathbf{G}  u^{[n-1]}  \right) u^{[n-1]} \right)_{t_{j}} \\ \notag
    & + \left(\mathbf{G}u^{[n-1]}\right)_{t_{j}} + \left(\mathbf{G}^{\ast}u^{[n-1]}\right)_{t_{j}},
\end{align}
over a discrete time grid $\left\{ 0 = t_{0} < \cdots < t_{N} = T \right\}, \; N \in \mathbb{N}^{*}$, where all the operators are discretized with the Nyström scheme in the same way as in \cite[Section 3.2]{abi2024trading}.

\textbf{Least Square Monte Carlo (LSMC).} For this, we use LSMC techniques in Markovian settings (e.g.~with an Ornstein-Uhlenbeck drift and $G$ given by an exponential, or sum of two exponentials) in the same spirit as \cite[Section 3.3]{abi2024trading}.

Indeed, assume we have $P \in \mathbb{N}^{*}$ (Markovian) regression variables observed at date $t_{i}, \; i \in \{ 1, \cdots, N \}$
$$
\mathbb{X}_{t_{i}} := \left( X_{t_{i}}^{p} \right)_{p \in \{1, \cdots, P\}}
$$
to expand into an orthonormal polynomial basis of maximum degree $d \in \mathbb{N}^{*}$ of the form
\begin{equation} \label{eq:basis_expansion_definition}
    \mathcal{B}^d(\mathbb{X}_{t_{i}}) := \left\{ \Pi_{p=1}^{P} L_{l_{p}} \left( X_{t_{i}}^{p} \right) \; \Big| \; \left( l_{p} \right)_{p \in \{1, \cdots, P\}} \in \mathbb{N}^{P}, \; \sum_{p=1}^{P} l_{p} \leq d \right\},
\end{equation}
where $L_{l_{p}}, \; p \in \mathbb{N}$ denotes the polynomial of degree $l_{p}$ of the basis. Consequently, there are in total
$$
\binom{P+d}{d}
$$
features in the expanded basis \eqref{eq:basis_expansion_definition}, given that we distribute at most 
$d$
degrees among the 
$P$ features, including the possibility of assigning none (see the stars and bars theorem). Finally, the conditional expectations \eqref{eq:conditional_expectations_to_estimate} are estimated by
\begin{equation} \label{eq:estimated_conditional_expectations}
    E_{i,j}[n-1] \approx \langle l_{i,j}^{d}[n-1], \mathcal{B}^d(\mathbb{X}_{t_{i}}) \rangle, 
\end{equation}
where $\left( l_{i,j}^{d}[n-1] \right)_{0 < i < j \leq N-1}$ are obtained by minimizing the Ridge-regularized least squares of the dependent variables $\widetilde{Y}_{t_{j}}[n-1]$ against the explanatory variables $\mathcal{B}^d(\mathbb{X}_{t_{i}})$ over $M$ sample trajectories of the regression variables. 

The impact of the quality of estimation of the conditional expectations in \eqref{eq:estimated_conditional_expectations} on the PnL functional \eqref{eq:gain_functional} as well as the overall convergence of the scheme will be discussed numerically in Section \ref{S:explicitbenchmar} and in Appendix \ref{s:pnl_error_scheme}.

\textbf{Profit and Loss.} From \eqref{eq:gain_functional_alpha_impact_cost_trade_of} we extract  the Profit and Loss (PnL) defined by
\begin{equation*}
    PnL(u) := \mathbb{E}\left[ \int_{0}^{T} \left( \alpha_{t} - I_{t}^{u} \right) u_t \dd t \right], \quad u \in \mathcal{L}^{2}.
\end{equation*}
Numerically, PnL is approximated  over $M \in \mathbb{N}^{*}$ discrete sample trajectories on the uniform time grid with step $\Delta := \frac{T}{N-1} > 0$ by
\begin{equation} \label{eq:empirical_pnl_functional_alpha_impact_cost_trade_of}
    PnL^{N,M}(u) := \frac{\Delta}{M} \sum_{m = 1}^{M} \sum_{i=0}^{N-1} \left( \alpha_{t_{i}}(\omega_{m}) - I_{t_{i}}^{u}(\omega_{m}) \right) u_{t_{i}}(\omega_{m}), \quad u \in \mathcal{L}^{2},
\end{equation}
where $I_{t_{i}}^{u}$ is approximated again by the Nyström scheme.

\textbf{Error metric when solving the nonlinear Fredholm equation \eqref{eq:nonlinearfredholm}.} After $n \in \mathbb{N}^{*}$ iterations of the scheme \eqref{scheme_initialization}--\eqref{scheme_update}, we define the following error metric applied to $u^{[n]}$:
\begin{align*}
    E \left( u^{[n]} \right) := \E \bigg[\int_0^T \bigg| & \gamma u^{[n]}_t(\omega) + \left(\mathbf A \left(u^{[n]}\right)\right)_{t}(\omega) + \left(\mathbf{H}_{\phi,\varrho}u^{[n]}\right)_{t}(\omega) +  \left(\mathbf{H}_{\phi,\varrho}^{\ast}u^{[n]}\right)_{t}(\omega) \\
    & + X_{0}(\phi(T-t)+\varrho) - \alpha_t(\omega) \bigg|^2 \dd t\bigg].
\end{align*}
As in \eqref{eq:empirical_pnl_functional_alpha_impact_cost_trade_of}, $E(u^{[n]})$ is approximated numerically by
\begin{equation} \label{eq:empirical_error_metric}
    E^{N,M} \left( u^{[n]} \right) := \frac{\Delta}{M} \sum_{m = 1}^{M} E^{N} \left( u^{[n]} \left( \omega_m \right) \right),
\end{equation}
with
\begin{align} \notag
    E^{N} \left( u^{[n]} \left( \omega_m \right) \right) := \sum_{i=0}^{N-1} \bigg| & \gamma u^{[n]}_{t_{i}}(\omega_m) + h \left( g_{t_{i}}(\omega_m) + \left(\mathbf{G} u^{[n]} \right)_{t_{i}}(\omega_m) \right) \\ \notag
    & + \left[ \left( \mathbf{H}_{\phi,\varrho}^{\ast}u^{[n]} \right)_{t_{i}} + \mathbf{G}^{\ast} \left( h' \left( g + \mathbf{G}  u^{[n]}  \right) u^{[n]} \right)_{t_{i}} \right] (\omega_m) \\ \label{eq:empirical_error_metric_per_omega}
    & + \left(\mathbf{H}_{\phi,\varrho}u^{[n]}\right)_{t_{i}}(\omega_m) - \alpha_{t_{i}}(\omega_m) + X_{0}(\phi(T-t_{i})+\varrho) \bigg|^2.
\end{align}
Here,  all the operators are estimated by the Nÿstrom scheme and the conditional expectations by Least-Square Monte Carlo (LSMC) with a Ridge regularization, similarly to \eqref{eq:estimated_conditional_expectations}.

\subsection{Impact function and signal specification} \label{ss:impact_function_signal_specification}

\textbf{Impact function definition.} Let $c \in (0,1]$ and $x_{0} > 0$, and consider the impact function for $x \in \mathbb{R}$ such that
\begin{equation} \label{eq:impact_function_specification}
h_{x_{0},c}(x) := \begin{cases}
    x & \text{ if } |x| \leq x_{0} \\
    \text{sign}(x) \left(\frac{1}{c}|x|x_{0}^{1/c-1}-\left(\frac{1}{c}-1\right)x_{0}^{1/c} \right)^{c} & \text{ if } |x| > x_{0}
\end{cases},
\end{equation}
so that the derivative of the impact function is given explicitly by
\begin{equation} \label{eq:impact_function_specification_derivative}
    h_{x_{0},c}^{'}(x) = \mathbbm{1}_{[-x_{0}, x_{0}]}(x) + x_{0}^{1/c-1} \left(\frac{1}{c}|x|x_{0}^{1/c-1}-\left(\frac{1}{c}-1\right)x_{0}^{1/c} \right)^{c-1} \mathbbm{1}_{(-\infty, -x_{0}) \cup (x_{0}, \infty)}(x).
\end{equation}
Also, its second derivative is given for $x\in \mathbb{R}\setminus \{-x_0,x_0\}$ by
\begin{equation} \label{eq:impact_function_specification_second_derivative}
    h_{x_{0},c}^{''}(x) = \frac{c-1}{c} x_{0}^{2 \left( 1/c-1 \right)} \left(\frac{1}{c}|x|x_{0}^{1/c-1}-\left(\frac{1}{c}-1\right)x_{0}^{1/c} \right)^{c-2} \left( \mathbbm{1}_{(x_{0}, \infty)}(x) - \mathbbm{1}_{(-\infty, -x_{0})} (x) \right).
\end{equation}
Notice that $x_{0}$ controls the neighbourhood width $[-x_{0}, x_{0}]$ around the origin where the impact function coincides with the linear case, while $c$ controls the concavity of the function. Figure \ref{F:impact_function_and_derivative} illustrates the effect of $c$ on the concavity of $h_{x_{0},c}$ and its derivative. From the explicit expressions \eqref{eq:impact_function_specification_derivative} and \eqref{eq:impact_function_specification_second_derivative} above, we readily obtain the following result.

\begin{proposition} \label{P:arrow_pratt_ratio_for_f_x_0_c}
    $h_{x_{0},c}$ given by \eqref{eq:impact_function_specification} is odd, concave in $\mathbb{R}_{+}$ and with bounded derivative. Furthermore, condition \eqref{eq:arrow_pratt_ratio_condition} on the Arrow-Pratt relative measure of risk ratio is satisfied in $\mathbb{R} \setminus \{-x_0,x_0\}$ as soon as
    \begin{equation} \label{eq:arrow_pratt_ratio_condition_satisfied}
        c \geq \frac{1}{2}.
    \end{equation}
\end{proposition}

\begin{proof}
    First, the facts that $h_{x_{0},c}$ is odd and concave in $\mathbb{R}_{+}$ are immediate. Second, given that $0 \le 1-c < 1$, then it is also immediate to note that $h_{x_{0},c}'$ is bounded above by $1$ from \eqref{eq:impact_function_specification_derivative}. Finally, observe that $h_{x_{0},c}' > 0$ and the function
    \begin{equation*}
        R(x) := - x \frac{h_{x_{0},c}''(x)}{h_{x_{0},c}'(x)} = \frac{(1-c)|x|}{|x|-(1-c)x_{0}}\mathbbm{1}_{(-\infty, -x_{0}) \cup (x_{0}, \infty)}(x),\quad x\in\mathbb{R}\setminus \{-x_0,x_0\},
    \end{equation*}
    reaches its supremum at $x= \pm x_{0}$, which yields \eqref{eq:arrow_pratt_ratio_condition_satisfied}.
\end{proof}

\begin{figure}[H]
\begin{center}
\includegraphics[width=6 in,angle=0]{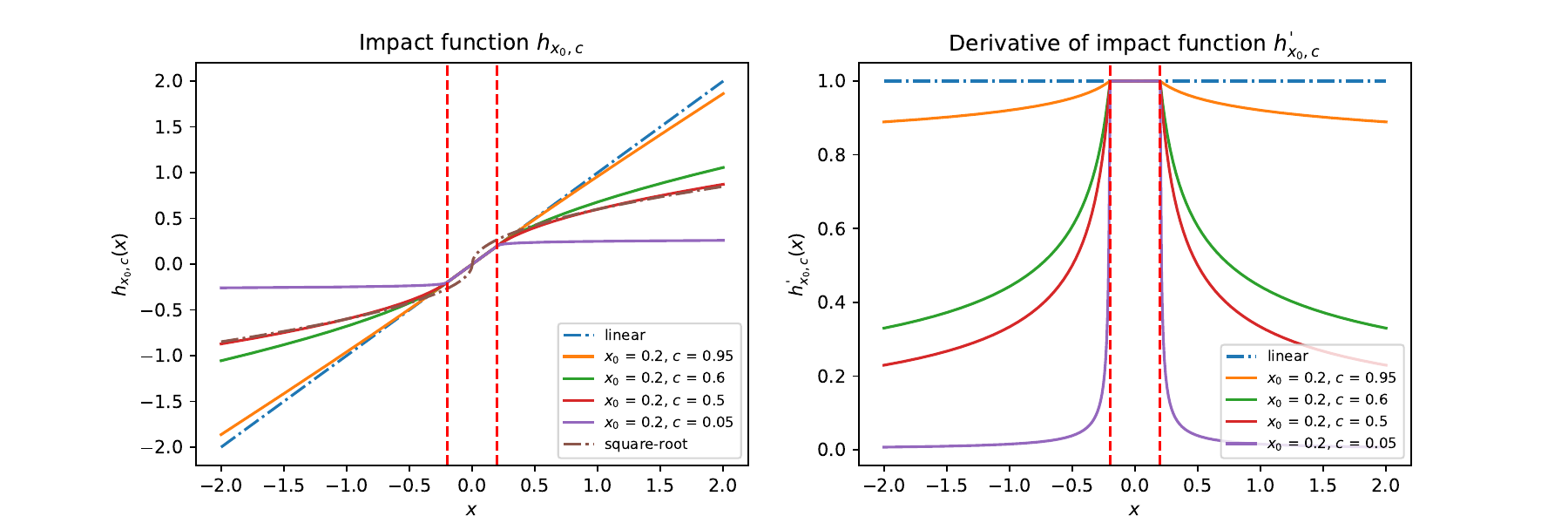}
\caption{Impact function specification \eqref{eq:impact_function_specification} compared to the linear $x \mapsto x$ and the square-root impact functions $x \mapsto \xi \text{sign}(x)\sqrt{|x|}, \; \xi = 0.6$, and its derivative \eqref{eq:impact_function_specification_derivative} for various values of $c$. Note that we don't display the derivative of the square-root impact function $x \in \mathbb{R}^{*} \mapsto \frac{\xi}{\sqrt{2|x|}}$ which explodes when approaching the origin.}
\label{F:impact_function_and_derivative}
\end{center}
\end{figure}

\begin{figure}[H]
\vspace*{-0.4in}
\begin{center}
\includegraphics[width=3 in,angle=0]{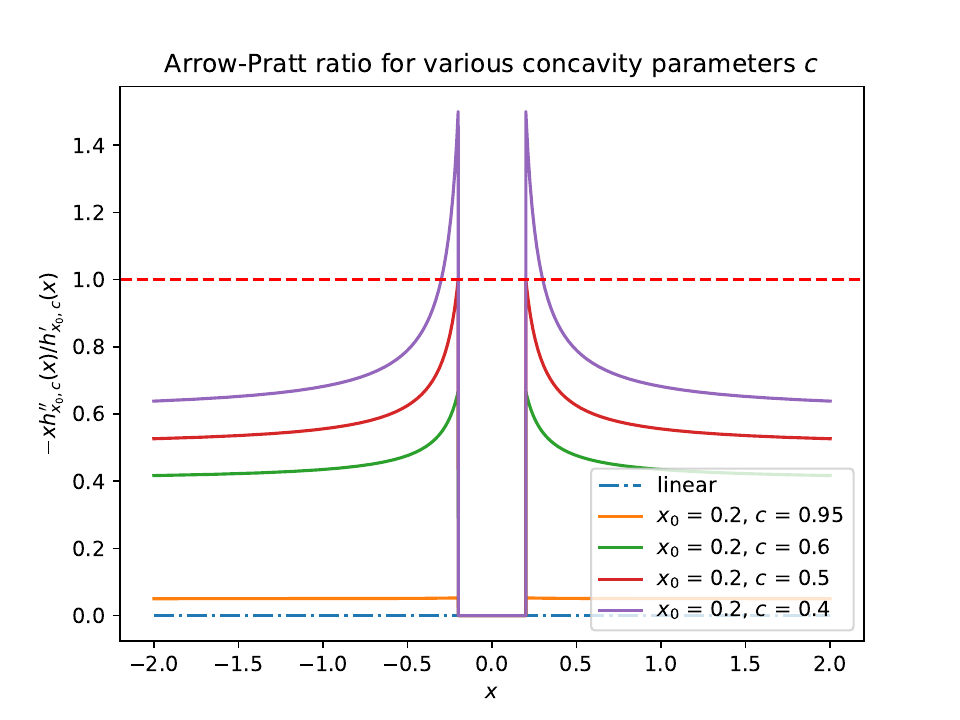}
\caption{Arrow-Pratt relative measure of risk ratio for various values of $c$, see Proposition \ref{P:arrow_pratt_ratio_for_f_x_0_c}.}
\label{F:arrow_pratt_ratio}
\end{center}
\end{figure}

\begin{remark}[Some limit behaviors of the impact function $h_{x_{0},c}$] \label{R:limit_behaviors_impact_function}
    Note in particular that for $c = 1$, the impact function becomes linear while for $c = 0.5$, we get back the impact function from \cite[Example 2.3]{hey2023trading}:
    \begin{equation*}
        h_{x_{0}, 0.5}(x) = \begin{cases}
        x, & \text{ if } |x| \leq x_{0}, \\
        \emph{sign}(x) \sqrt{2|x|x_{0}-x_{0}^{2}}, & \text{ if } |x| > x_{0},
    \end{cases}
    \end{equation*}
    so that the derivative of the impact function is given explicitly by
    \begin{equation*}
        h_{x_{0}, 0.5}^{'}(x) = \mathbbm{1}_{[-x_{0}, x_{0}]}(x) + \frac{x_{0}}{\sqrt{2|x|x_{0}-x_{0}^{2}}} \mathbbm{1}_{(-\infty, -x_{0}) \cup (x_{0}, \infty)}(x).
    \end{equation*}

    Finally, note that $c \leq 1$ ensures that $h_{x_{0}, 0.5}^{'}$ is bounded, while \eqref{eq:arrow_pratt_ratio_condition_satisfied} is satisfied only if the concavity of the transient impact is at most a square root, meaning that Theorem \ref{T:mainnonlinear} can be applied for $\frac{1}{2} \le c \le 1$ and $x_{0} > 0$.
\end{remark}

\textbf{Signals specification.} We specify an Ornstein-Uhlenbeck drift-like price  signal  $\tilde{I}$ such that
\begin{equation} \label{eq:drift_signal_specification}
\dd \tilde I_t = (\tilde \theta - \tilde \kappa \tilde I_t) \dd t + \tilde \xi \dd W_t, \quad \tilde I_{0} \in \mathbb{R}.
\end{equation}
The fundamental alpha-signal $\alpha$ that we consider is then given by
\begin{equation} \label{eq:signal_specification}
\alpha_{t} := \mathbb{E}_{t} \left[ \int_{t}^{T} \tilde I_r \dd r \right] = \left( \tilde I_{t} - \frac{\tilde \theta}{\tilde \kappa} \right) \frac{ 1 - e^{- \tilde \kappa (T-t)}}{\tilde \kappa} + \frac{\tilde \theta}{\tilde \kappa} (T-t), \quad t \in [0,T].
\end{equation}

\subsection{Numerical convergence of the scheme} \label{S:explicitbenchmar}

\textbf{Deterministic signal.} In this case, there is no need to estimate any conditional expectations and the convergence of the scheme is illustrated in Figure \ref{F:convergence_scheme_deterministic_case}. Note that the convergence rate is indeed exponential until machine precision for the various input kernels, which is consistent with Proposition \ref{pp:conv_scheme}, in particular \eqref{eq:convergence_rate}.

\begin{figure}[H]
\begin{center}
\hspace*{-0.3in}
\includegraphics[width=3.5 in,angle=0]{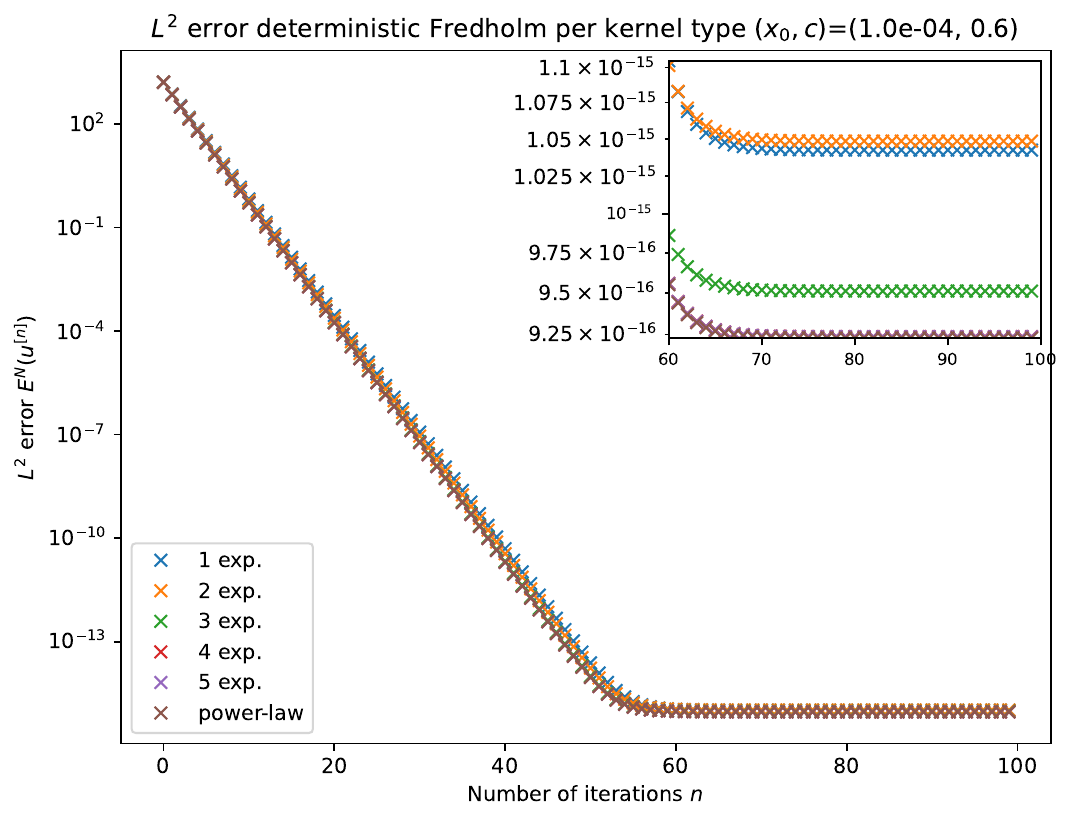}
\caption{Illustration of convergence at exponential rate of the numerical scheme \eqref{scheme_initialization}--\eqref{scheme_update} for various kernels in the deterministic case, with $\gamma = 1$, impact function $h_{x_{0}, c}$ specified in \eqref{eq:impact_function_specification} and deterministic signal given by \eqref{eq:drift_signal_specification}--\eqref{eq:signal_specification} with $\tilde \theta=-40$, $\tilde \kappa=1$, $\tilde \xi = 0$, $\tilde I_0=20$.}
\label{F:convergence_scheme_deterministic_case}
\end{center}
\end{figure}

\textbf{Asymptotic explicit benchmark for one exponential.} Extending the formulation of \cite[Corollary 4.4]{hey2023trading}, with exponential impact time decay $\tau > 0$ and ``Kyle's $\lambda$'' specified as $\lambda = 1 = e^{0}$, to the impact function $h: x \mapsto \text{sign}(x) |x|^{c}, \; x \geq 0$, $c \in [1/2, 1]$ to allow for negative values of the input $x$, we explicitly get  the optimal impact function
\begin{equation*}
    I_{t}^{*} = \frac{1}{1+c} \left( \alpha_{t} + \tau \tilde I_t \right),
\end{equation*}
where recall $\tilde I$ is the stochastic drift price signal. Since $
I_{t}^{*} := \text{sign}(J_{t}^{*}) \left|J_{t}^{*}\right|^{c}$, 
we readily deduce that 
\begin{equation*}
    J_{t}^{*} =  \text{sign}(I_{t}^{*}) \left| I_{t}^{*} \right|^{1/c},
\end{equation*}
which leads to
\begin{equation*}
    J_{t}^{*} = \text{sign} \left( \alpha_{t} + \tau \tilde I_t \right) \left( \frac{1}{1+c} \left| \alpha_{t} + \tau \tilde I_t \right| \right)^{1/c}, \quad t \in (0,T).
\end{equation*}
Now, also considering that $\alpha_{T} = 0$,  the boundary conditions on the impact function $I^{*}$ read $    I_{0}^{*} = 0$, 
   and $I_{T}^{*} = \alpha_T=0$
which translate for $J^{*}$  into $J_{0}^{*} = 0$  and 
   $J_{T}^{*} = 0.$
Thus, the optimal traded volume, which satisfies 
\begin{equation*}
Q_{t}^{*} = \int_{0}^{t} \dd J_{s}^{*} + \int_{0}^{t} \frac{J_{s}^{*}}{\tau} \dd s, \quad t \in (0,T], \quad Q_0=0,
\end{equation*}
is explicitly given by
\begin{equation} \label{eq:explicit_benchmark_smooth}
Q_{t}^{*} = J^{*}_t + \frac{1}{\tau}\int_{0}^{t}  \text{sign} \left( \alpha_{s} + \tau \tilde I_s \right) \bigg( \frac{| \alpha_{s} + \tau \tilde I_s |}{1+c} \bigg)^{1/c} \dd s, \quad t \in (0,T],
\end{equation}
with the bulk trades 
\begin{equation} 
\begin{aligned}\label{eq:explicit_benchmark_bulk_trades}
&\Delta Q_{0} = Q_{0+}-Q_0 = \text{sign} \left( \alpha_{0} + \tau \tilde I_0 \right) \left( \frac{1}{1+c} \left| \alpha_{0} + \tau \tilde I_{0} \right| \right)^{1/c}, \\ &
\Delta Q_{T} = Q_T-Q_{T-}=J_T-J_{T-}= - \text{sign} \left( \tilde I_T \right) \left| \frac{\tau \tilde I_{T} }{1+c} \right|^{1/c}.
\end{aligned}
\end{equation}
Such bulk trades ensure that the boundary conditions are satisfied.

\begin{figure}[H]
\begin{center}
\includegraphics[width=5.8 in,angle=0]{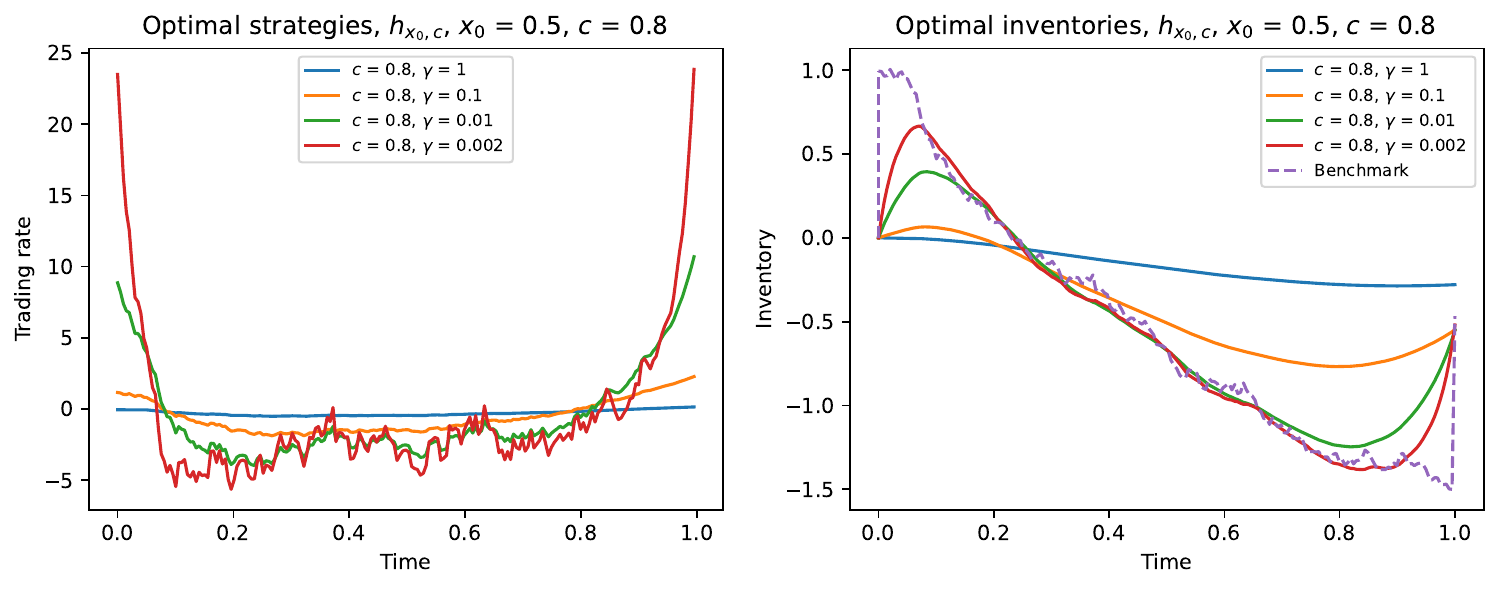}
\caption{Convergence of the numerical solution obtained by the scheme \eqref{scheme_initialization}-\eqref{scheme_update} in the stochastic case to the explicit benchmark \eqref{eq:explicit_benchmark_smooth}-\eqref{eq:explicit_benchmark_bulk_trades} when decreasing $\gamma$ in the concave case, i.e., $\left(x_{0}, c\right) = \left(0.5,0.8\right)$ in the impact function specification \eqref{eq:impact_function_specification}. The stochastic signal is parametrized as the integral of an Ornstein-Uhlenbeck as in \eqref{eq:drift_signal_specification} with volatility $\tilde \xi = 0.5$, mean level $\tilde \theta = -4$, mean-reversion speed $\tilde \kappa=1$ and initial value $\tilde I_0 = 2$. We run $n=30$ iterations for each case, with $N = 200$ time-steps and $M=10000$ sample trajectories. Both the scaling parameter and the mean-reversion rate $\tau$ of the exponential decay are fixed to $1$. The regression basis is set as $\left( u^{[n-1]}, \int_{0}^{\cdot} u^{[n-1]}_{s} \dd s, \int_{0}^{\cdot} e^{-\tilde\kappa(\cdot-s)}u^{[n-1]}_{s} \dd s \right)$ at each iteration and Laguerre polynomials up to degree $d=4$ are selected for the expansion basis. The respective numerical error metrics $E^{N,M}$ from \eqref{eq:empirical_error_metric} are $6e-5$, $1e-3$, $5e-2$ and $1e-1$ for $\gamma = 1, \; 0.1, \; 0.01, \; 0.002$.}
\label{F:concave_stochastic_explicit_inventory}
\end{center}
\end{figure}

\subsection{Impact of power-law decay approximation on trading}

\textbf{Multi-factor approximation of the fractional kernel.} Define the following $\epsilon$-shifted fractional kernel
\begin{equation} \label{eq:def_shifted_frac_kernel}
    G_{\nu, \epsilon}(t) := \xi \left(t+\epsilon\right)^{\nu - 1} \quad \xi>0, \quad \epsilon \geq 0, \quad \nu \in \left( \frac{1}{2}, 1 \right),
\end{equation}
and, for $p \in \mathbb{N}^{*}$, denote the sum of exponential time scales decay kernel by
\begin{equation} \label{eq:def_sum_of_exponentials}
    G_p(t) := \sum_{i=1}^p \xi_i e^{-x_i t},\quad \xi_i,x_i>0, \quad i \in \{ 1, \cdots, p \}.
\end{equation}
Then, given a trading time horizon $T>0$, we aim to approximate the power-law impact decay captured in \eqref{eq:def_shifted_frac_kernel} by a sum of exponential decays in a kernel of the form \eqref{eq:def_sum_of_exponentials}. We minimize the quantity 
\begin{align} \notag
    ||G_{\nu, \epsilon} - G_p||^2_{L^{2}} & := \int_{0}^{T} \left( G_{\nu, \epsilon}(t)-G_p(t) \right)^2 dt \\ \notag
    & = \sum_{i,j=1}^p \frac{\xi_i \xi_j}{x_i + x_j} \left( 1-e^{-(x_i + x_j)T_{l}} \right) + \xi^2 \frac{(T_{l}+\epsilon)^{2\nu - 1} - \epsilon^{2\nu - 1}}{2\nu - 1} \\ \label{eq:loss_approximate_fractional_kernel}
    & \quad - 2\xi\sum_{i=1}^p \frac{\xi_i e^{x_{i} \epsilon}}{x_i^{\nu}} \Gamma \left( \nu \right) \left( P \left( \nu, x_i (T_{l}+\epsilon) \right) - P \left( \nu, x_i \epsilon \right) \right),
\end{align}
with $\Gamma$ and $P$ denoting, respectively, the Gamma function and the regularized lower incomplete Gamma  function, that is,
\begin{equation*}
    \Gamma(a) := \int_0^{\infty} u^{a-1} e^{-u} du, \quad \text{and} \quad P(a,t) := \frac{1}{\Gamma(a)} \int_0^t u^{a-1} e^{-u} du, \quad a >0, \quad t \geq 0.
\end{equation*}
Figure \ref{F:approximated_epsilon_shifted_power_law_kernel_by_exponentials} shows the resulting approximating sum of exponentials kernels \eqref{eq:def_sum_of_exponentials} to fit the $\epsilon$-shifted power-law kernel \eqref{eq:def_shifted_frac_kernel}. The parameters $\left( \xi_{i}, x_{i} \right)_{i \in \{ 1, \cdots, p\}}$ are calibrated incrementally from $1$ to $p$ exponential time scales by minimizing the loss \eqref{eq:loss_approximate_fractional_kernel} on the trading horizon $T=1$ using a standard nonlinear solver. The results of the calibration are displayed in Table~\ref{tab:optimal_weights_mean_reversions_approx_shifted_fractional_kernel}. Notice that fewer exponentials are required to approximate well the power-law when the shift is indeed strictly positive.
\begin{figure}[H]
\begin{center}
\hspace*{-0.3in}
\includegraphics[width=6 in,angle=0]{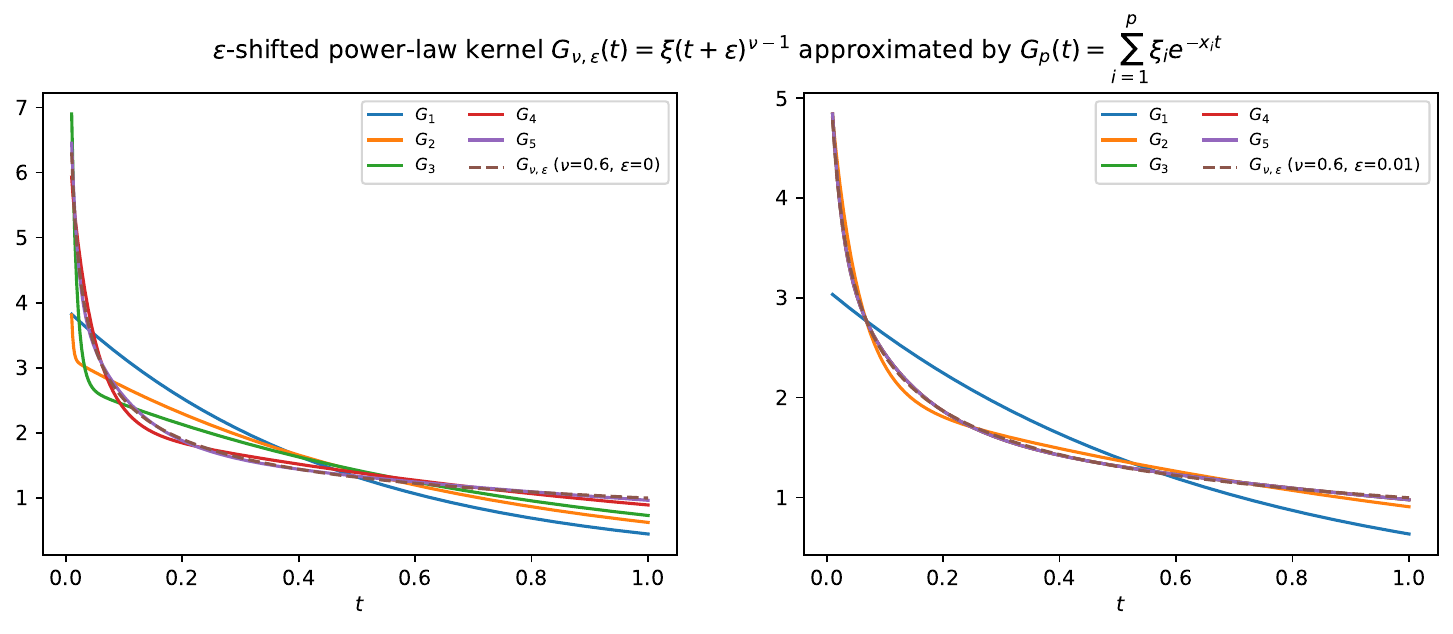}
\caption{Power-law decay approximation by multiple exponential decays: without shift ($\epsilon=0$) on the left, and with shift ($\epsilon>0$) on the right.}
\label{F:approximated_epsilon_shifted_power_law_kernel_by_exponentials}
\end{center}
\end{figure}

\begin{table}[H]
    \centering
    \small 
    \renewcommand{\arraystretch}{1.2} 
    \setlength{\tabcolsep}{3pt} 

    \begin{tabular}{|c!{\vrule width 0.7pt}c|c|c!{\vrule width 1.2pt}c|c|c|}
    \hline
    \multirow{2}{*}{$p$} & \multicolumn{3}{c!{\vrule width 1.2pt}}{$G_{\nu,\epsilon}$, $(\nu, \epsilon) = (0.6, 0)$} & \multicolumn{3}{c|}{$G_{\nu,\epsilon}$, $(\nu, \epsilon) = (0.6, 0.01)$} \\
    \cline{2-7}
    & $\xi_{i}$ & $x_{i}$ & $\|G_{\nu, \epsilon} - G_p\|^2_{L^{2}}$ & $\xi_{i}$ & $x_{i}$ & $\|G_{\nu, \epsilon} - G_p\|^2_{L^{2}}$ \\
    \hline
    1 & $3.907e+00$ & $2.165e+00$ & $1.522e+00$ & $3.082e+00$ & $1.583e+00$ & $1.450e-01$ \\
    \hline
    2 & $\begin{array}{c} 3.180e+00 \\ 2.504e+01 \end{array}$ & $\begin{array}{c} 1.625e+00 \\ 3.600e+02 \end{array}$ & $6.984e-01$ & $\begin{array}{c} 2.074e+00 \\ 3.394e+00 \end{array}$ & $\begin{array}{c} 8.281e-01 \\ 2.114e+01 \end{array}$ & $5.537e-03$ \\
    \hline
    3 & $\begin{array}{c} 2.780e+00 \\ 1.196e+02 \\ 1.267e+01 \end{array}$ & $\begin{array}{c} 1.334e+00 \\ 1.808e+04 \\ 1.115e+02 \end{array}$ & $3.615e-01$ & $\begin{array}{c} 1.704e+00 \\ 2.438e+00 \\ 1.976e+00 \end{array}$ & $\begin{array}{c} 5.578e-01 \\ 5.994e+01 \\ 8.786e+00 \end{array}$ & $2.000e-04$ \\
    \hline
    4 & $\begin{array}{c} 2.166e+00 \\ 1.037e+02 \\ 1.689e+01 \\ 4.709e+00 \end{array}$ & $\begin{array}{c} 8.850e-01 \\ 1.808e+04 \\ 4.993e+02 \\ 2.495e+01 \end{array}$ & $2.681e-01$ & $\begin{array}{c} 1.699e+00 \\ 2.437e+00 \\ 1.574e+00 \\ 4.072e-01 \end{array}$ & $\begin{array}{c} 5.547e-01 \\ 5.999e+01 \\ 8.482e+00 \\ 9.947e+00 \end{array}$ & $1.977e-04$ \\
    \hline
    5 & $\begin{array}{c} 1.818e+00 \\ 9.872e+01 \\ 1.621e+01 \\ 6.018e+00 \\ 2.583e+00 \end{array}$ & $\begin{array}{c} 6.327e-01 \\ 1.812e+04 \\ 8.918e+02 \\ 9.488e+01 \\ 1.113e+01 \end{array}$ & $2.581e-01$ & $\begin{array}{c} 1.697e+00 \\ 2.435e+00 \\ 1.362e+00 \\ 1.748e-02 \\ 6.062e-01 \end{array}$ & $\begin{array}{c} 5.532e-01 \\ 6.002e+01 \\ 8.266e+00 \\ 9.944e+00 \\ 9.955e+00 \end{array}$ & $1.960e-04$ \\
    \hline
    \end{tabular}

    \caption{Comparison of approximations of power-law decay $G_{\nu,\epsilon}$ from \eqref{eq:def_shifted_frac_kernel} with $\nu = 0.6$ and $\epsilon = 0$ or $\epsilon = 0.01$, using sums of exponential decays $G_{p}$ from \eqref{eq:def_sum_of_exponentials}. Optimal weights $\left(\xi_{i}\right)_{i \in \{1,\cdots,p\}}$, mean-reversion rates $\left(x_{i}\right)_{i \in \{1,\cdots,p\}}$, and losses $\|G_{\nu, \epsilon} - G_p\|^2_{L^{2}}$ are shown for $p$ exponential decays in the sum.}
    \label{tab:optimal_weights_mean_reversions_approx_shifted_fractional_kernel}
\end{table}

In Figure \ref{fig:stability_illustration_fractional}, we illustrate the convergence of the empirical optimal PnLs \eqref{eq:empirical_pnl_functional_alpha_impact_cost_trade_of} 
 obtained by solving numerically the nonlinear Fredholm equation \eqref{eq:nonlinearfredholm} by the scheme \eqref{scheme_initialization}--\eqref{scheme_update} with an increasing number of exponential time scales in the approximating kernel \eqref{eq:def_sum_of_exponentials}, 
 to the optimal PnL of the corresponding  asymptotic (shifted-)fractional kernel.
Theoretically, this stability result is investigated in Proposition \ref{pp:stability_with_respect_to_kernel}.  
Note in particular that, in this setting, the ``two exponential time scales'' rule of thumb when approximating the power-law decay seems reasonable. However, five exponential time scales  are preferred to properly mimic the same numerical performances as the power-law, although the empirical optimal PnLs are very close to one another. By contrast, three exponential time scales allow to reproduce the performance of the shifted-power-law kernel.

\begin{figure}[H]
    \centering
    \begin{subfigure}[t]{0.48\textwidth}
        \centering
        \includegraphics[width=\textwidth, trim=0cm 0cm 0cm 0cm, clip]{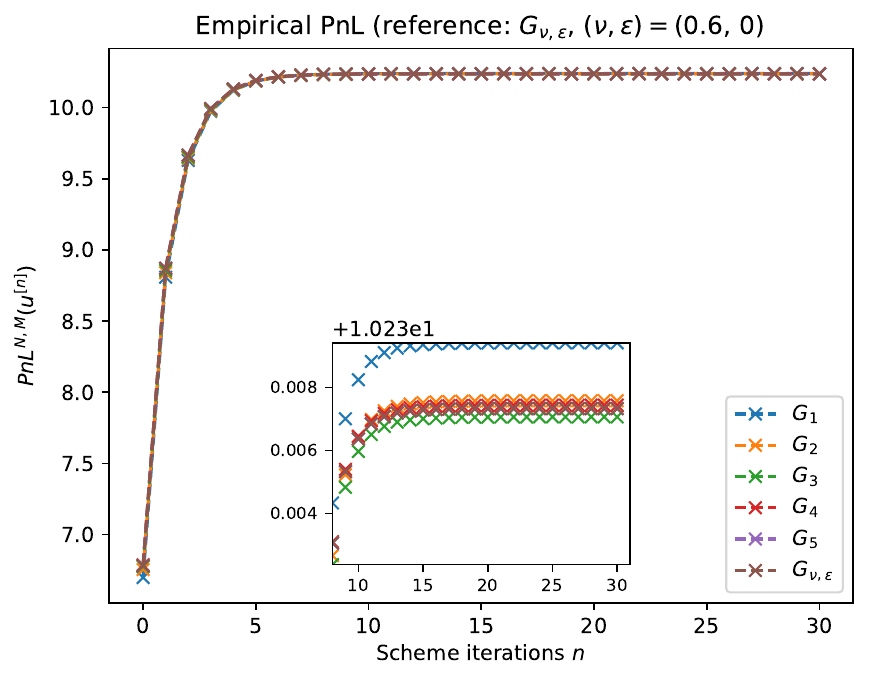} 
    \end{subfigure}
    \hfill 
    \begin{subfigure}[t]{0.48\textwidth}
        \centering
        \includegraphics[width=\textwidth, trim=0cm 0cm 0cm 0cm, clip]{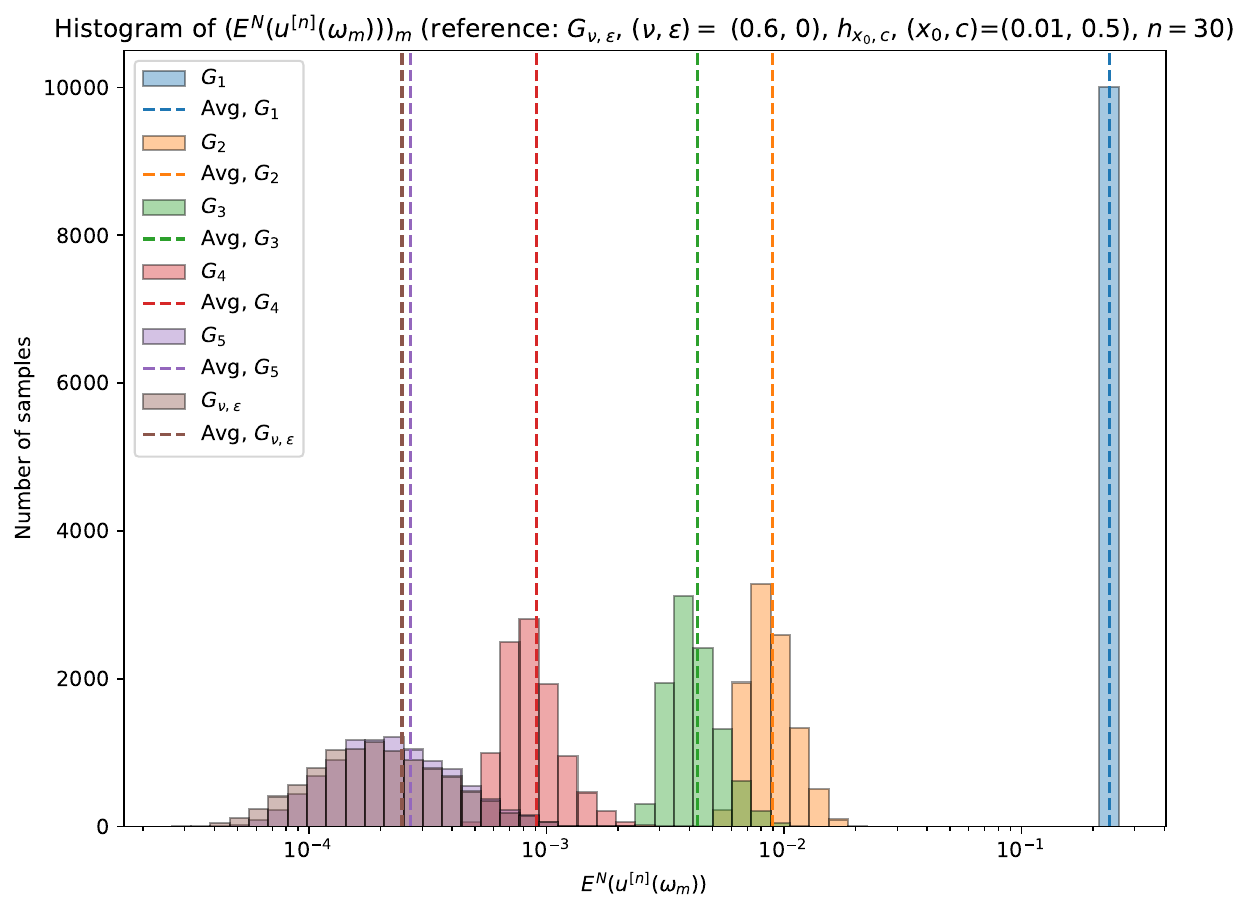} 
    \end{subfigure}
    
    \begin{subfigure}[t]{0.48\textwidth}
        \centering
        \includegraphics[width=\textwidth, trim=0cm 0cm 0cm 0cm, clip]{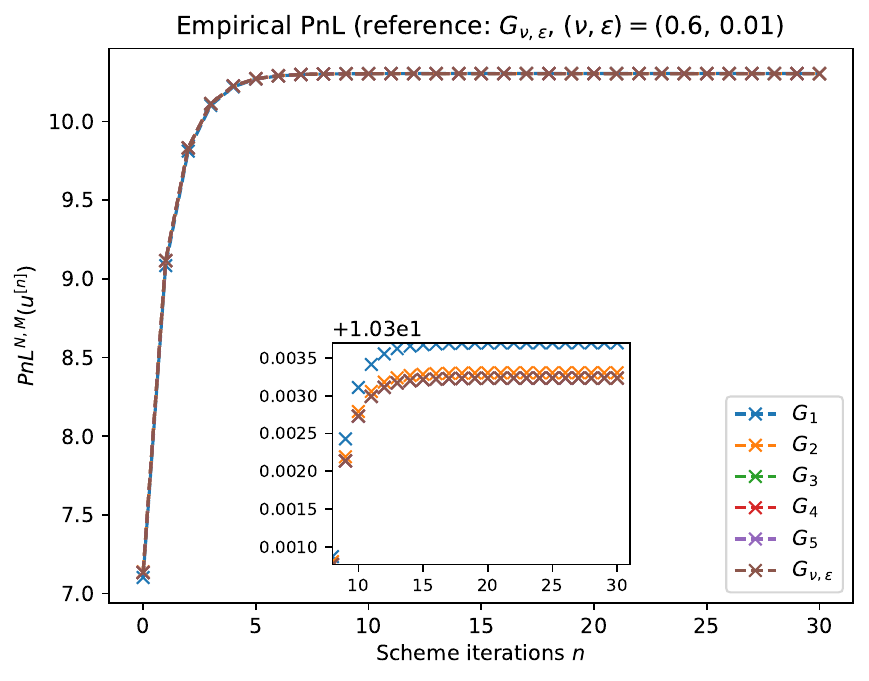} 
    \end{subfigure}
    \hfill 
    \begin{subfigure}[t]{0.48\textwidth}
        \centering
        \includegraphics[width=\textwidth, trim=0cm 0cm 0cm 0cm, clip]{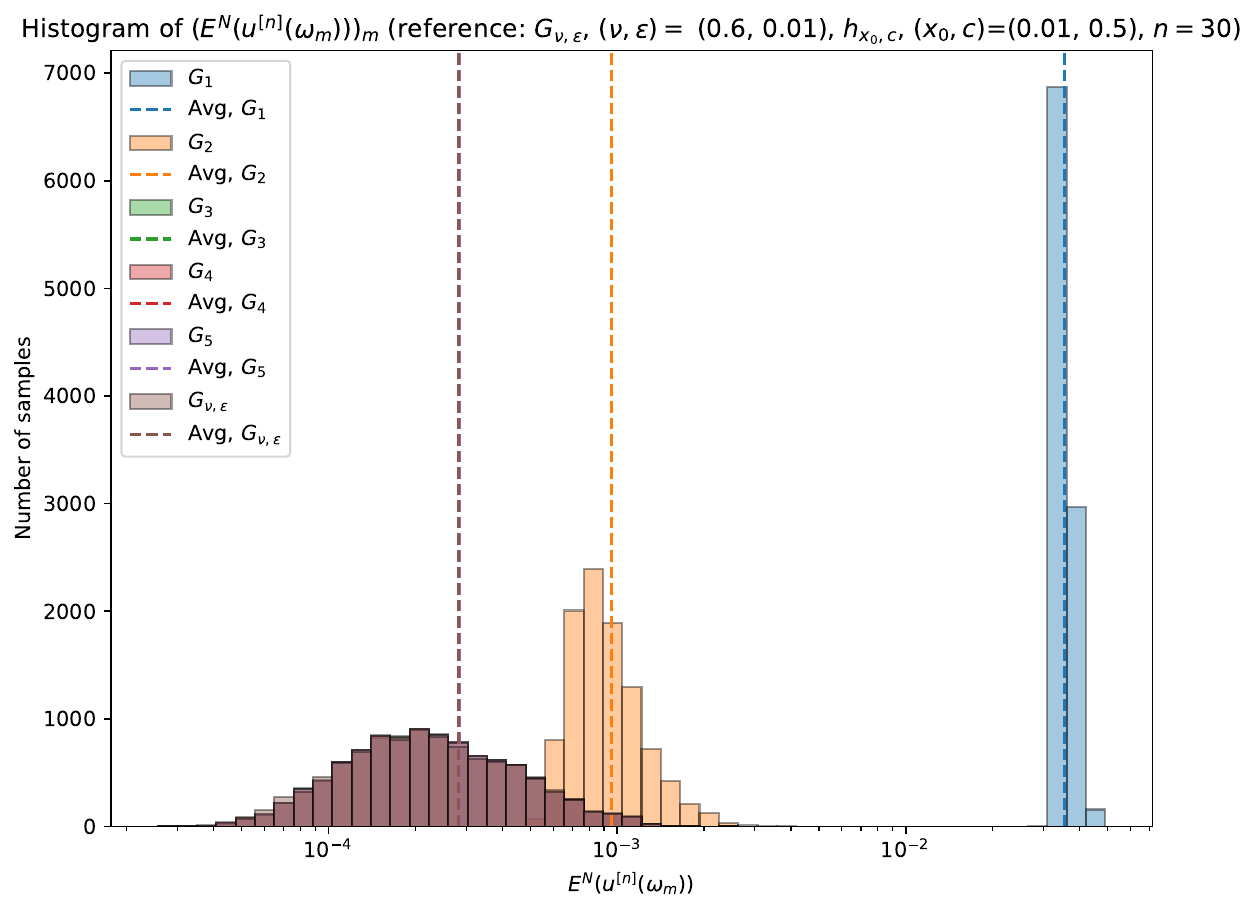} 
    \end{subfigure}
    
    \caption{Impact of the approximation of the (shifted-)fractional decay kernel by one to five exponential time scales, with the parameters  in Table~\ref{tab:optimal_weights_mean_reversions_approx_shifted_fractional_kernel}, on the empirical $PnL^{N,M}$ from \eqref{eq:empirical_pnl_functional_alpha_impact_cost_trade_of} through the scheme iterations (top and bottom left) and the histograms of the numerical errors $\left( E^{N}(u^{[n]}(\omega_{m}) \right)_{m}$ from \eqref{eq:empirical_error_metric_per_omega} (the respective empirical averages $E^{N,M}(u^{[n]})$ from \eqref{eq:empirical_error_metric} are displayed in dashed line) after $n=30$ iterations (top and bottom right). We fix the impact function as an approximating square-root by taking $(x_{0},c)=(0.01,0.5)$ for $h_{x_{0},c}$ in \eqref{eq:impact_function_specification}. In each case, the stochastic drift-signal's parameters are fixed to $\tilde \theta = 40, \; \tilde \kappa = 5, \; \tilde \xi = 5$ and $\tilde{I}_{0} = 10$ in \eqref{eq:drift_signal_specification}. We take $N=100$ time steps, $M=10000$ trajectories, and we consider the regression variables $\left( \alpha, \int_{0}^{.} \alpha_{s} \dd s, \int_{0}^{t} e^{-\tilde \kappa (t-s)} \alpha_{s} \dd s \right)$, Laguerre polynomials up to degree $2$ for the basis expansion for the LSMC with a Ridge regularization set to $1e-6$.}
    \label{fig:stability_illustration_fractional}
\end{figure}

\subsection{Impact of concavity and comparison of exponential versus power-law decay on trading}

The first row of Figure \ref{F:impact_concavity_parameter} displays, in the deterministic case, the influence of the concavity parameter $c \in [0.5,1]$ in the impact function $h_{x_{0},c}$ from \eqref{eq:impact_function_specification} on the optimal trading inventories (upper-left) and the resulting price distortions (upper-right) with a power-law decay kernel when trading a ``buy''-signal, starting and ending at zero inventory by setting $\varrho \gg 1$ into the gain functional \eqref{eq:gain_functional_alpha_impact_cost_trade_of}. The more concave the impact function is, the more aggressive the trades are while the amplitude of the price distortion decreases.

Moreover, the second row of Figure \ref{F:impact_concavity_parameter} compares the optimal trading inventories (bottom-left) and the resulting price distortions (bottom-right) for power-law and best-power-law-approximating exponential decays in both the linear and square-root impact function cases, i.e., taking $c=1$ and $c=0.5$ respectively in \eqref{eq:impact_function_specification}. Note that the trades are more aggressive in presence of the best-power-law-approximating exponential decay than for the power-law decay since accounting for a single exponential time-scale in the decay leads to an underestimation of the resulting price distortion actually caused by the power-law decay.

\begin{figure}[H]
    \begin{center}
    \includegraphics[width=5.5 in,angle=0]{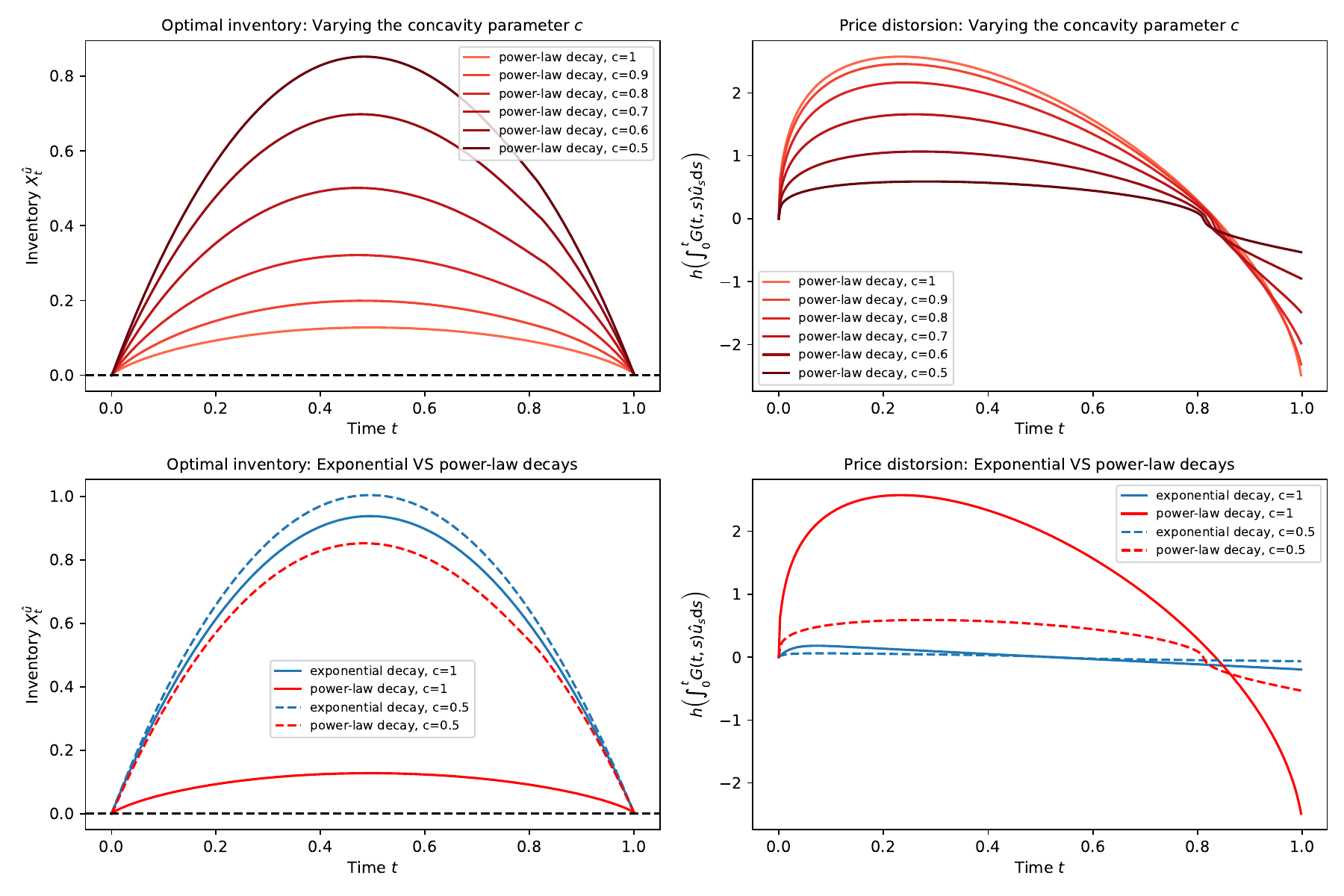}
    \caption{Impact of concavity with power-law versus one exponential decay on the optimal trading inventory and resulting price distortion in presence of a ``buy'' signal. For the impact function $h_{x_{0},c}$, we fix $x_{0} = 1e-2$ in \eqref{eq:impact_function_specification}; for the fractional kernel $G_{\nu, \epsilon}$, we fix $\xi = 10$, $\epsilon=0$ in \eqref{eq:def_shifted_frac_kernel}; for the one exponential time-scale decay, we fix $\xi_{1} = 39.07, x_{1} = 2.165$ in \eqref{eq:def_sum_of_exponentials} using Table~\ref{tab:optimal_weights_mean_reversions_approx_shifted_fractional_kernel}. We set $\varrho = 5e2$ in \eqref{eq:gain_functional_alpha_impact_cost_trade_of} to enforce full liquidation. We take $N = 400$ time steps and $n=100$ iterations for the scheme such that the numerical error $E^{N} \left( u^{[n]} \right)$ from \eqref{eq:empirical_error_metric} is less than the order of $1e-9$ in the square-root and linear cases for both power-law and exponential decays. The deterministic drift-signal's parameters are fixed to $\tilde \theta = 40, \; \tilde \kappa = 5, \; \tilde \xi = 0$ and $\tilde{I}_{0} = 10$ in \eqref{eq:drift_signal_specification}.}
    \label{F:impact_concavity_parameter}
    \end{center}
\end{figure}

\section{Proof of Theorem~\ref{T:mainnonlinear} } \label{S:proofofoptimality}

The main idea of the proof is to use general results of convex optimization in infinite dimensional spaces. We recall that, given a positive constant $\beta>0$, a map $\mathcal{T}\colon \mathcal{L}^{2}\to \mathbb{R}$ is said to be $\beta-$\emph{strongly concave} if the following inequality is satisfied for every $u,\,v\in \mathcal{L}^{2}$:
\begin{equation}\label{def:strongly_concave}
\mathcal{T}(\theta u+(1-\theta)v) \ge \theta\mathcal{T}(u)+(1-\theta)\mathcal{T}(v)
+\frac{\beta \theta(1-\theta)}{2}\norm{u-v}^2 ,\quad \theta\in(0,1).
\end{equation}
Supposing that $\mathcal{J}$ in \eqref{eq:gain_functional} is concave and G\^ateaux differentiable in $\mathcal{L}^{2}$,  we can study the existence of optimal strategies $\hat{u} \in \mathcal{L}^{2}$ in the sense of \eqref{eq:optimal_strategy} by solving the equation 
\begin{equation}\label{eq:FOC}
	\nabla \mathcal{J}(u)=0,\quad u\in \mathcal{L}^{2},
\end{equation} 
where $\nabla \mathcal{J}(u)\in \mathcal{L}^{2}$ is the G\^ateaux differential of $\mathcal{J}$, defined by
\begin{equation*}
    \langle \nabla \mathcal{J}(u), h \rangle = \lim_{\epsilon\to 0} \frac{\mathcal{J}(u+\epsilon h)-\mathcal{J}(u)}{\epsilon}, \quad h \in \mathcal{L}^{2}.
\end{equation*}
More precisely, the following theorem holds.

\begin{theorem}\label{thm:concave}
	Suppose that the performance functional $\mathcal{J}$ in \eqref{eq:gain_functional} is well-defined and G\^ateaux differentiable in $\mathcal{L}^{2}$. 
	\begin{enumerate}[label={(\roman*)}]
		\item\label{point:concave}If $\mathcal{J}$ is concave in $\mathcal{L}^{2}$, then the set of optimal strategies $\hat{u}\in \mathcal{L}^{2}$ satisfying \eqref{eq:optimal_strategy} coincide with the set of solutions of \eqref{eq:FOC}. In addition, if the functional $-\mathcal{J}$ is coercive, i.e., $\lim_{\norm{u}\to \infty}\mathcal{J}(u)=-\infty$, then there exists at least one optimal strategy that solves \eqref{eq:optimal_strategy}. If additionally $\mathcal{J}$ is strictly concave, then such an optimal strategy is unique.
		\item\label{point:strictly_concave} In particular, if $\mathcal{J}$ is strongly concave in the sense of \eqref{def:strongly_concave}, then there exists a unique optimal strategy $\hat{u}\in \mathcal{L}^{2}$ satisfying \eqref{eq:optimal_strategy}, which is also the unique solution of \eqref{eq:FOC}.
	\end{enumerate}
\end{theorem}
\begin{proof}
    The proof of the theorem statements are  consequences of well-known results of convex analysis in Hilbert spaces that can be found in \cite{bauschke2017} (see also \cite{ekeland1999convex}). We first notice that $\hat{u} \in \mathcal{L}^{2}$ satisfies \eqref{eq:optimal_strategy} if and only if 
    \begin{equation}\label{eq:minimization}
        -\mathcal{J}(\hat{u})=\inf_{u\in\mathcal{L}^{2}}(-\mathcal{J}(u)),
    \end{equation}
	and that $-\mathcal{J}$ is proper because, by the hypothesis of well-posedness of $\mathcal{J}$, $\mathcal{J}(\mathcal{L}^{2})\subset \mathbb{R}$. Since $\mathcal{J}$ is supposed to be G\^ateaux differentiable and concave in $\mathcal{L}^{2}$,  the fact that the set of optimal strategies $\hat{u} \in \mathcal{L}^{2}$ according to \eqref{eq:minimization} (or \eqref{eq:optimal_strategy}) coincides with the set of solutions to \eqref{eq:FOC} is ensured by \cite[Proposition 17.4]{bauschke2017}.	Moreover, by \cite[Proposition 17.39]{bauschke2017}, the opposite performance functional $-\mathcal{J}$ is lower semi-continuous in $\mathcal{L}^{2}$. Consequently, if we further require $-\mathcal{J}$ to be  coercive, an application of \cite[Proposition 11.14]{bauschke2017} yields the existence of at least one optimal strategy $\hat{u} \in \mathcal{L}^{2}$ satisfying \eqref{eq:minimization}, hence \eqref{eq:optimal_strategy}, too. Such an optimal strategy $\hat{u} \in \mathcal{L}^{2}$ is unique when $-\mathcal{J}$ is  strictly convex by \cite[Corollary 11.8]{bauschke2017}. This completes the proof of \ref{point:concave}.\\
	As for the assertion in  \ref{point:strictly_concave}, it follows by the same arguments as before combined with \cite[Corollary 11.16]{bauschke2017}. The theorem is now fully proved.
\end{proof}

Given an admissible kernel $G$ and an admissible impact function $h$, we can readily observe that assuming either assumption (i) or assumption (ii) ensures that $\mathcal{J}$ is well-defined in $\mathcal{L}^{2}$.

Therefore, it remains to prove that $\mathcal{J}$ is
\begin{enumerate}
    \item strongly concave,
    \item Gâteaux differentiable, and to compute its derivative.
\end{enumerate}
This is the objective of the next two subsections.

\subsection{\texorpdfstring{$\mathcal{J}$ is Gâteaux differentiable}{J is Gateaux differentiable}} \label{ss:gateau_diff}

	The next lemma shows that, under suitable requirements on the impact function $h$, recall also Remark \ref{rem_conditions_h}, the performance functional $\mathcal{J}$ is G\^ateaux differentiable in $\mathcal{L}^{2}$. In the sequel, we denote by $\mathbf{1}$ the linear integral operator in $\mathcal{L}^{2}$ defined by  
\begin{equation*}
	(\mathbf{1}v)_{t} := \int_{0}^{T} \mathbbm{1}_{\{s\le t\}} v_{s} \dd s = \int_{0}^{t} v_{s} \dd s, \quad t\in [0,T], \quad v\in \mathcal{L}^{2}.
\end{equation*}
\begin{lemma}\label{thm:FOC}
    Suppose that the impact function $h\colon \mathbb{R}\to \mathbb{R}$ satisfies assumption (i) from Definition \ref{def:admissible_impact_function}. Then the functional $\mathcal{J}$ in \eqref{eq:gain_functional_alpha_impact_cost_trade_of} is well-defined and G\^ateaux differentiable in $\mathcal{L}^{2}$, with G\^ateaux differential given by 
    \begin{equation}\label{Gat_diff_J}
         \nabla \mathcal{J}(u)=
         \alpha - X_{0} \left( \phi (T-\cdot) + \varrho \right) - \gamma u - \mathbf{A}(u) - \left( \mathbf{H}_{\phi,\varrho}(u) \right)_{\cdot}
         - \left( \mathbf{H}_{\phi,\varrho}^{*}(u) \right)_{\cdot},\quad u\in \mathcal{L}^{2},
    \end{equation}
    where $\mathbf{A}$ and $H_{\phi,\varrho}$ are respectively defined in \eqref{eq:def_A_operator} and \eqref{eq:kernel_H}.
\end{lemma}
  
\begin{proof}
Note that, for $u\in \mathcal{L}^{2}$, the definition of $\mathcal{J}$ in  \eqref{eq:gain_functional_alpha_impact_cost_trade_of} can be rewritten as follows
\begin{equation} \label{eq:J_def_rewrite}
	\mathcal{J}(u) = \langle \alpha, u \rangle
	- \frac{\gamma}{2} \norm{u}^2
    - \langle h(Z^u), u\rangle
    - \frac{\phi}{2}\norm{X^u}^2
    - \frac{\varrho}{2} \mathbb{E}\left[|X^u_T|^2\right]
    + X_{0}\E S_{T} =: \sum_{i=1}^{6}\mathcal{J}_i(u).
\end{equation}
We now proceed to G\^ateaux differentiate all the addends in \eqref{eq:J_def_rewrite}. Evidently, $\langle \nabla \mathcal{J}_1(u), f\rangle=\langle \alpha, f\rangle$, $\langle \nabla \mathcal{J}_2(u), f\rangle= \langle -\gamma u, f\rangle$ and $\langle \nabla \mathcal{J}_6(u), f\rangle=0$ for every $f\in \mathcal{L}^{2}.$ As for $\mathcal{J}_3$, we let $\epsilon\in (-1,1)\setminus\{0\}$ and compute 
the incremental ratio in a direction $f\in \mathcal{L}^{2}$, namely
\begin{align}\label{eq:J_3_diff1}
	\notag&\frac{\mathcal{J}_3(u+\epsilon f)-\mathcal{J}_3(u)}{\epsilon}=-\mathbb{E}\left[\int_{0}^{T}\frac{h(Z^{u+\epsilon f}_t)-h(Z^u_t)}{\epsilon}(u_t+\epsilon f_t)\,\dd t \right] - \mathbb{E}\left[\int_{0}^{T}h(Z^u_t) f_t \dd t\right]\\
	&\qquad =\notag - \mathbb{E}\left[\int_{0}^{T}\frac{h(Z^{u+\epsilon f}_t)-h(Z^u_t)}{\epsilon}u_t \dd t \right] - \mathbb{E}\left[\int_{0}^{T}{\left(h(Z^{u+\epsilon f}_t) - h(Z^u_t)\right)}f_t \dd t \right] - \langle h(Z^u), f\rangle \\
	& \qquad =:
	\mathbf{\upperRomannumeral{1}}^{u,f}(\epsilon) +
	\mathbf{\upperRomannumeral{2}}^{u,f}(\epsilon) - \langle h(Z^u), f\rangle.
\end{align}
We then study the limits of $\mathbf{\upperRomannumeral{1}}^{u,f}(\epsilon)$ and $\mathbf{\upperRomannumeral{2}}^{u,f}(\epsilon)$ as $\epsilon \to 0$. Notice that, from the definition in \eqref{eq:def_Z}, $Z^{u+\epsilon f}_t=Z^{u}_t + \epsilon \left( \mathbf{G} f\right)_{t},\, t\in [0,T]$, so that, by Lipschitz continuity of $h$ with constant $\norm{h'}_\infty$, we immediately get
\[
\left|h(Z^{u+\epsilon f}_t)-h(Z^u_t)\right||f_t|\le \epsilon \norm{h'}_\infty |f_t| \left| \left( \mathbf{G} f\right)_{t} \right| \underset{\epsilon \to 0}{\longrightarrow} 0, \quad \dd t \otimes \mathbb{P}-\text{a.e}.
\] 
Since $f, \; \mathbf{G} f \in \mathcal{L}^{2} \subset \mathcal{L}^{1}$, the dominated convergence theorem entails that $\lim_{\epsilon \to 0}\mathbf{\upperRomannumeral{2}}^{u,f}(\epsilon)=0$. Regarding $\mathbf{\upperRomannumeral{1}}^{u,f}(\epsilon)$, by the tower property of the conditional expectation and Fubini's theorem, 
\begin{align*}
    \notag \mathbf{\upperRomannumeral{1}}^{u,f}(\epsilon) &
    = - \mathbb{E}\left[\int_{0}^{T}\frac{h(Z^{u+\epsilon f}_t)-h(Z^u_t)}{Z^{u+\epsilon f}_t - Z^{u}_t} \left( \mathbf{G} f\right)_{t} \mathbbm{1}_{\left\{\left( \mathbf{G} f\right)_{t} \neq 0\right\}}u_t\,\dd t \right]
    \\&\underset{\epsilon\to 0}{\longrightarrow} \,
    -\mathbb{E}\left[\int_{0}^{T}h'(Z^u_t) \left( \mathbf{G} f\right)_{t} u_t \dd t\right]
    = - \mathbb{E} \left[ \int_{0}^{T} \left(\mathbf{G}^{\ast}\left(h'(Z^u)u\right)\right)_{t} f_t \dd t \right] \\
    & = - \mathbb{E}\left[ \int_{0}^{T} \left(\mathbf{G}^{\ast}\left(h'(Z^u)u\right)\right)_{t} f_t \dd t \right] = -
    \langle \mathbf{G}^{\ast}\left(h'(Z^u)u\right), f\rangle.
    \end{align*}
The passage to the limit under the integral sign in the previous step is again justified by the dominated convergence theorem, which can be applied since $\left(\dd t \otimes \mathbb{P}\right)-${a.e.},
\begin{equation*}
    \frac{\left|h(Z^{u+\epsilon f}_t)-h(Z^u_t)\right|}{\epsilon} \mathbbm{1}_{\left\{\left( \mathbf{G} f\right)_{t} \neq 0\right\}} |u_t|\le \norm{h'}_\infty\left|\left( \mathbf{G} f\right)_{t} \right||u_t|\in \mathcal{L}^1.
\end{equation*}
Therefore, thanks to the previous computations, we take the $\lim_{\epsilon \to 0}$ in \eqref{eq:J_3_diff1} to infer that
 \[
 	\big\langle \nabla \mathcal{J}_3(u), f \rangle = - \langle \mathbf{G}^\ast (h'(Z^u)u) + h(Z^u), f\big\rangle.
 \] 
Recall the definition of $X^u$ in \eqref{eq:running_inventory}, which in particular gives for any $f \in \mathcal{L}^{2}$,  $X_t^{u+\epsilon f}=X_t^u-\epsilon \int_{0}^{t}f_s\dd s, \; t \in [0,T], \; \mathbb{P}-$a.s. Then, expanding the squares, taking the limit $\epsilon \to 0$, and applying Fubini readily yields
\begin{equation*}
    \big\langle \nabla \left(\mathcal{J}_4(u) + \mathcal{J}_5(u) \right), f \rangle = - \langle \phi \mathbf{1}^{\ast} X^{u} + \varrho \E_{\cdot} \left[ X_{T}^{u} \right], f\big\rangle.
\end{equation*}
Notice that, by Fubini's theorem, for any $t \in [0,T]$,
\begin{align*}
    \phi \left(\mathbf{1}^{\ast} X^{u}\right)_{t} + \varrho \E_{t} \left[ X_{T}^{u} \right] & = \phi X_{0} (T-t) + \E_{t} \int_{0}^{T} \left( \int_{0}^{s} \mathbbm{1}_{\{s>t\}} u_{r} \dd r \right) \dd s + \varrho X_{0} + \E_{t} \int_{0}^{T} \varrho u_{r} \dd r \\
    & = X_{0} \left( \phi (T-t) + \varrho \right) + \E_{t} \int_{0}^{T} \left( \int_{r}^{T} \mathbbm{1}_{\{s>t\}} \dd s \right) u_{r} \dd r + \E_{t} \int_{0}^{T} \varrho u_{r} \dd r \\
    & = X_{0} \left( \phi (T-t) + \varrho \right) + \int_{0}^{t} \left( \int_{r}^{T} \mathbbm{1}_{\{s>t\}} \dd s \right) u_{r} \dd r \\
    & \quad + \int_{t}^{T} \left( \int_{r}^{T} \mathbbm{1}_{\{s>t\}} \dd s \right) \E_{t} u_{r} \dd r + \int_{0}^{t} \varrho u_{r} \dd r + \int_{0}^{t} \varrho \E_{t} u_{r} \dd r,
\end{align*}
and 

\begin{equation*}
    \begin{cases}
        \int_{0}^{t} \left( \int_{r}^{T} \mathbbm{1}_{\{s>t\}} \dd s \right) u_{r} \dd r = \int_{0}^{t} \left( \int_{t}^{T} \dd s \right) u_{r} \dd r = \int_{0}^{t} \left( T-t \right) u_{r} \dd r \\
        \int_{t}^{T} \left( \int_{r}^{T} \mathbbm{1}_{\{s>t\}} \dd s \right) \E_{t} u_{r} \dd r = \int_{t}^{T} \left( \int_{r}^{T} \dd s \right) \E_{t} u_{r} \dd r = \int_{t}^{T} \left( T-r \right) \E_{t} u_{r} \dd r
    \end{cases},
\end{equation*}
so that
\begin{align} \notag
    \phi \left(\mathbf{1}^{\ast} X^{u}\right)_{t} + \varrho \E_{t} \left[ X_{T}^{u} \right] & = X_{0} \left( \phi (T-t) + \varrho \right) + \int_{0}^{t} \left( \phi (T-t) + \varrho \right) u_{r} \dd r \\ \notag
    & \quad + \int_{t}^{T} \left( \phi (T-r) + \varrho \right) \E_{t} u_{r} \dd r \\ \label{eq:expression_H_phi_rho}
    & = X_{0} \left( \phi (T-t) + \varrho \right) + \left( \mathbf{H}_{\phi,\varrho}(u) \right)_{t} + \left( \mathbf{H}_{\phi,\varrho}^{*}(u) \right)_{t}.
\end{align}
Consequently,
\begin{equation*}
    \big\langle \nabla \left(\mathcal{J}_4(u) + \mathcal{J}_5(u) \right), f \rangle = - \langle X_{0} \left( \phi (T-t) + \varrho \right) + \left( \mathbf{H}_{\phi,\varrho}(u) \right)_{t}
    + \E_{t} \left( \mathbf{H}_{\phi,\varrho}^{*}(u) \right)_{t}, f\big\rangle,
\end{equation*}
where the admissible kernel $H_{\phi,\varrho}$ is defined in \eqref{eq:kernel_H}. Combining the previous computations, we obtain the G\^ateaux derivatives of all the terms in \eqref{eq:J_def_rewrite}, which  yields \eqref{Gat_diff_J} and completes the proof.  
\end{proof}

\subsection{\texorpdfstring{If $\mathbf{A}$ is monotone, then $\mathcal{J}$ is strongly concave}{If A is monotone, then J is strongly concave}}

Motivated by Theorem \ref{thm:concave}, we aim to determine suitable  conditions on the impact function $h$, the endogenous signal $g$ and the kernel $G$  to ensure the (strong) concavity of the performance functional $\mathcal{J}$.  This is done in the next lemma, where we focus on the operator $\mathbf{A}$ defined in \eqref{eq:def_A_operator}, which is indeed fully determined by $h, \; g$ and $G$.

\begin{lemma}\label{thm:conc_conditions}
	Suppose that the map $h \colon \mathbb{R} \to  \mathbb{R}$ is differentiable, with bounded derivative $h'$. The operator $\mathbf{H}_{\phi,\varrho}$ is positive semi-definite. As a consequence, if the operator  $\mathbf{A} \colon \mathcal{L}^{2} \to \mathcal{L}^{2}$ in \eqref{eq:def_A_operator} is monotone in the sense of \eqref{eq:definition_monotone}, then the functional $\mathcal{J}$ in \eqref{eq:gain_functional} is $\gamma-$strongly concave.
\end{lemma}

\begin{proof}
	We define the functional $\widetilde{\mathcal{J}}$ in $\mathcal{L}^{2}$  by
    \begin{equation*}
        \widetilde{\mathcal{J}}(u):= - \mathcal{J}(u) - \frac{\gamma}{2} \norm{u}^2,\quad u \in \mathcal{L}^{2}.
    \end{equation*}
	The aim of the proof consists in showing that $\widetilde{\mathcal{J}}$ is convex because, by \cite[Proposition 10.8]{bauschke2017}, this fact is equivalent to the $\gamma-$strong convexity of $-\mathcal{J}$.
	By Lemma~\ref{thm:FOC}, $\widetilde{\mathcal{J}}$ is G\^ateaux differentiable. Hence, \cite[Proposition 17.10]{bauschke2017} gives the equivalence between the convexity of $\widetilde{\mathcal{J}}$ and the monotonicity of the G\^ateaux gradient $\nabla\widetilde{\mathcal{J}}$, i.e.,
	\begin{equation*}
	\langle u-v, \nabla \widetilde{\mathcal{J}}(u)-\nabla \widetilde{\mathcal{J}}(v)\rangle\ge0,\quad u,\, v\in \mathcal{L}^{2}.
	\end{equation*}
    By \eqref{Gat_diff_J} while using the definition of $\mathbf A$ from \eqref{eq:def_A_operator}, we readily get for $u, \; v \in \mathcal{L}^{2}$
    \begin{equation*}
        \left\langle \nabla \widetilde{\mathcal{J}}(u)-\nabla \widetilde{\mathcal{J}}(v), u-v \right\rangle = \left \langle \mathbf{A}(u)-\mathbf{A}(v), u-v \right \rangle + \left \langle \left(\mathbf{H}_{\phi,\varrho} + \mathbf{H}_{\phi,\varrho}^{\ast} \right)(u-v), u-v \right \rangle.
    \end{equation*}
    It is then sufficient to prove that $\mathbf{H}_{\phi,\varrho}$ is positive semi-definite in the sense of \eqref{eq:sdp_operator_def}. Using \eqref{eq:expression_H_phi_rho}, we obtain
    \begin{align*}
        \left \langle \left(\mathbf{H}_{\phi,\varrho} + \mathbf{H}_{\phi,\varrho}^{\ast} \right)(u-v), u-v \right \rangle & = \left \langle \phi \mathbf{1}^{\ast} X^{u-v} + \varrho X_{T}^{u-v}, u-v \right \rangle \\
        & = \phi \left \langle X^{u-v}, \mathbf{1} \left(u-v\right) \right \rangle + \varrho \left \langle X_{T}^{u-v}, u-v \right \rangle \\
        & = \mathbb{E}\left[ \phi \int_{0}^{T}\left(\int_{0}^{t}(u_s-v_s)\dd s\right)^2\dd t  + \varrho \bigg(\int_{0}^{T}(u_t-v_t)\dd t\bigg)^2 \right] \ge 0.
    \end{align*}
    The proof is now complete.
\end{proof}

\subsection{Putting everything together}

We can now complete the proof of Theorem~\ref{T:mainnonlinear}.

\begin{proof}[Proof of Theorem~\ref{T:mainnonlinear}]
The assumptions of Theorem \ref{T:mainnonlinear} enable us to apply Lemmas \ref{thm:FOC}-\ref{thm:conc_conditions} to deduce that the performance functional $\mathcal{J}$ in \eqref{eq:gain_functional} is well-defined, G\^ateaux differentiable and $\gamma-$strongly concave in $\mathcal{L}^{2}$. Consequently, by Theorem \ref{thm:concave}\ref{point:strictly_concave}, there exists a unique optimal strategy $\hat{u} \in \mathcal{L}^{2}$ that satisfies \eqref{eq:optimal_strategy}, which is also the unique solution to \eqref{eq:FOC}. By the G\^ateaux derivative expression \eqref{Gat_diff_J} of $\mathcal{J}$, this is equivalent to solving the nonlinear stochastic Fredholm equation \eqref{eq:nonlinearfredholm}, completing the proof.    
\end{proof}

\section{Proof of Proposition \ref{P:one_exponential_case}} \label{s:sufficient_conditions_A_monotone}

Let  the admissible kernel $G$ be  such that $k: t \in [0,T] \mapsto G(t,t) > 0$ is a well-defined positive function\footnote{This condition excludes all kernels with a singularity at the origin such as the fractional kernel.}. Suppose that $G$ is differentiable with respect to its first argument, and denote by $\partial_{x}G(t,s), \; t,s \in [0,T]$ the resulting kernel. The associated operator will be denoted by $\mathbf{\partial_{x}G}$  and  assumed to be admissible\footnote{Any combination of exponential kernels with positive weights satisfies such requirements.}.

For $u \in \mathcal{L}^{2}$, the dynamics of $Z^u$ from \eqref{eq:def_Z} are given by 
\begin{equation} \label{eq:dynamics_Z_u}
    \dd Z^u_t= \dd g_{t} + \left( k(t) u_t +(\mathbf{\partial_{x}G}u)_{t} \right) \dd t, \quad Z^u_0 = g(0).
\end{equation}
Using \eqref{eq:dynamics_Z_u}, and then integrating by parts the second addend in the definition \eqref{eq:def_A_operator} of $\mathbf{A}$, we obtain
\begin{align*}
    \left(\mathbf{A}(u)\right)_{t} & = h(Z_t^u) +
    \mathbb E_t \left[ \int_t^T \frac{G(s,t)}{k(s)} h'\left(Z^u_s\right) \dd Z^u_s\right]
    - \mathbb E_t \left[ \int_t^T \frac{G(s,t)}{k(s)} h'\left(Z^u_s\right) \dd g_{s} \right] \\
    & \quad -  \mathbb E_t \left[ \int_t^T \frac{G(s,t)}{k(s)} h'\left(Z^u_s\right) (\mathbf{\partial_{x}G}u)_{s} \dd s \right] \\
    & = \mathbb{E}_t\left[\frac{G(T,t)}{k(T)}h(Z^u_T)\right] -\mathbb{E}_t\left[ \int_{t}^{T}\frac{\partial_{x}G(s,t)k(s)-G(s,t)k'(s)}{k(s)^2}h(Z^u_s) \dd s \right] \\
    & - \mathbb E_t\bigg[\int_t^T  \frac{G(s,t)}{k(s)} h'\left(Z^u_s\right) \dd g_{s} \bigg]
    -\mathbb E_t\bigg[\int_t^T  \frac{G(s,t)}{k(s)} h'\left(Z^u_s\right) \left(\mathbf{\partial_x G}u\right)_{s} \,\dd s\bigg]
    \\&=:(\left( \mathbf{\upperRomannumeral{1}+\upperRomannumeral{2}+\upperRomannumeral{3}}+\mathbf{\upperRomannumeral{4}})(u)\right)_{t}, \quad t \in [0,T].
\end{align*}
Observe that $\mathbf{\upperRomannumeral{1}} \; \colon \; \mathcal{L}^{2}\to\mathcal{L}^{2}$ is a monotone operator as soon as $h$ is a nondecreasing map, since
\begin{align*}
	\left\langle
	\mathbf{\upperRomannumeral{1}}(u)-\mathbf{\upperRomannumeral{1}}(v),u-v
	\right\rangle
&	=\frac{1}{k(T)}
	\mathbb{E}\bigg[\left(h(Z^u_T)-h(Z^v_T)\right)\int_{0}^{T}G(T-t)(u_t-v_t) \dd t\bigg]
	\\&=
	\frac{1}{k(T)}
	\mathbb{E}\Big[\left(h(Z^u_T)-h(Z^v_T)\right)(Z^u_T-Z^v_T)\Big]\ge 0,\quad u,v\in \mathcal{L}^{2}.
\end{align*}
Then,  applying Fubini's Theorem we get
\begin{align}
    \left\langle
    \mathbf{\upperRomannumeral{3}}(u)-\mathbf{\upperRomannumeral{3}}(v),u-v
    \right\rangle
    =-\mathbb{E}\bigg[\int_{0}^{T}(h'(Z^u_t)-h'(Z^v_t))(Z^u_t-Z^v_t) \frac{\dd g_{t}}{k(t)}\bigg], \quad u,v\in\mathcal{L}^{2},
\end{align}
so that the monotonicity of the operator $\mathbf{\upperRomannumeral{3}}$ may not true in general and depends on the dynamics of the endogenous signal $g$. The particular case where $\dd g_{t} = g'(t) dt$ with $g' \ge 0$ (i.e., $g$ is a nondecreasing input curve) and $h'$ is nonincreasing (i.e., the impact function $h$ is concave on the real line) would yield the monotonicity of $\mathbf{\upperRomannumeral{3}}$, but such assumptions may be too restrictive. 

Furthermore, verifying the monotonicity property of the  operators $\mathbf{\upperRomannumeral{2}}$ and $\mathbf{\upperRomannumeral{4}}$ in general is not obvious and depends on the form of the kernel $G$ and its first-argument derivative $\partial_x G$. 

\paragraph*{Case of one exponential for the impact decay.} For this reason, in what follows we assume (i) $h$ to be nondecreasing and we restrict our attention to the exponential Volterra kernel, i.e.,
\[
    G(t,s) = \mathbbm{1}_{\{t\ge s\}}ae^{-b(t-s)}, \quad t,s\in [0,T], \quad a > 0, \; b \ge 0,
\]
to study the monotonicity property of the operator $\mathbf{\upperRomannumeral{5}} := \mathbf{\upperRomannumeral{2}} + \mathbf{\upperRomannumeral{4}}$. In this way, we conveniently use the relation $\partial_xG(t,s)=-bG(t,s)$, for $0\le s <t\le T,$ and the fact that $k \equiv a$. Specifically, by Fubini's Theorem and straightforward calculus, we get for any $u,\; v \in \mathcal{L}^{2}$
\begin{align*}
    \left\langle \mathbf{\upperRomannumeral{5}}(u)-\mathbf{\upperRomannumeral{5}}(v),u-v \right\rangle & = \frac{b}{a}\left\langle h(Z^u)-h(Z^v),\,Z^u-Z^v\right\rangle \\
    & + \quad \frac{b}{a} \left\langle Z^uh'(Z^u)-Z^vh'(Z^v),Z^u-Z^v \right\rangle \\
    & - \frac{b}{a} \mathbb{E}\bigg[ \int_{0}^{T} g_{t} (h'(Z^u_t)-h'(Z^v_t))(Z^u_t-Z^v_t)\dd t \bigg].
\end{align*}
Concluding on the nonnegativity of the above quantity in general depends on the endogenous signal $g$ as was the case for $\mathbf{\upperRomannumeral{3}}$. But if we additionally assume that (ii) $g\equiv 0$ and (iii) $x\mapsto xh'(x)$ is nondecreasing, then clearly $\mathbf{\upperRomannumeral{5}}$ is monotone, which proves Proposition \ref{P:one_exponential_case}.
  
\paragraph*{Case of a sum of exponential time scales for the impact decay.} More generally, one may wonder whether we could derive in a similar way sufficient conditions to get the monotonicity property of the operator $\mathbf{\upperRomannumeral{5}}\colon\mathcal{L}^{2}\to\mathcal{L}^{2}$ in the case of a Volterra kernel given by a finite sum of exponentials, i.e., 
\[
G(t,s)=\mathbbm{1}_{\{t\ge s\}}\sum_{i=1}^{n}a_ie^{-b_i(t-s)}, \quad t,s\in [0,T], \quad 
a_i > 0, \;  b_i \ge 0, \quad i \in \{ 1, \cdots, n \}, \quad n\in\mathbb{N}.
\]
The answer turns out to be negative. Indeed, assume without loss of generality that the mean--reversion speeds $(b_i)_i$ differ from one another. Then, in this case, $k \equiv \sum_{i=1}^{N}a_i=: A$ and 
\[
	\partial_{x}G(t,s)=-\sum_{i=1}^n a_ib_i\exp\{-b_i(t-s)\}, \quad 0 \le s < t \le T.
\]
For any $i=1,\dots, n$, we also define the Volterra kernels $G_i(t,s)=\mathbbm{1}_{\{t\ge s\}}a_ie^{-b_i(t-s)}$, so that $G(t,s)=\sum_{i=1}^n G_i(t,s), \; s,t \in [0,T]$, and $\partial_xG(t,s)=-\sum_{i=1}^n b_iG_i(t,s),\,0<s<t<T$. Considering the initial input curve $g\equiv 0$, we can write 
$$
    Z^u = \sum_{i=1}^n \mathbf{G_i}u =: \sum_{i=1}^n Z^{i,u}, \quad u \in \mathcal{L}^{2}.
$$ 
Direct computations based on the definitions of $\mathbf{\upperRomannumeral{2}}$ and $\mathbf{\upperRomannumeral{4}}$ yield, for every $u\in\mathcal{L}^{2}$, 
\begin{equation*}
	\left(\mathbf{\upperRomannumeral{6}}(u)\right)_{t} = \frac{1}{A}\sum_{i=1}^{N} \bigg(\mathbf{G_i}^\ast\bigg(b_{i}h(Z^u) + h'(Z^u)\bigg(\sum_{j=1}^n b_j Z^{j,u}\bigg)\bigg)\bigg)_{t}, \quad t \in [0,T].
\end{equation*}
It follows by Fubini's Theorem and putting the finite sum inside the integrals that
\begin{align*}
    \left\langle \mathbf{\upperRomannumeral{6}}(u)-\mathbf{\upperRomannumeral{6}}(v),u-v \right\rangle & = \frac{1}{A}\sum_{i=1}^n \mathbb{E}\bigg[ \int_{0}^{T} \left(\mathbf{G_i}(u-v)\right)_{t} \bigg(b_i(h(Z^u_t)-h(Z^v_t)) \\
    &\quad + \bigg(h'(Z^u_t)\bigg(\sum_{j=1}^n b_j Z_{t}^{j,u} \bigg) - h'(Z^v_t)\bigg(\sum_{j=1}^n b_j Z_{t}^{j,v} \bigg)\bigg)\bigg) \dd t \bigg] \\
    & = \frac{1}{A} \mathbb{E}\bigg[\int_{0}^{T}\bigg((h(Z^u_t)-h(Z^v_t))\sum_{i=1}^n b_i(Z^{i,u}_t-Z^{i,v}_t) \\
    & \quad + (Z^u_t-Z^v_t) \bigg(h'(Z^u_t)\bigg(\sum_{j=1}^n b_jZ^{j,u}_t \bigg) - h'(Z^v_t)\bigg(\sum_{j=1}^n b_jZ^{j,v}_t \bigg)\bigg)\bigg)\dd t
    \bigg].
\end{align*}
Therefore, a sufficient condition on $h$ that guarantees the monotonicity of the operator $\mathbf{\upperRomannumeral{6}}$ is 
\begin{multline} \label{suff_condition_f}
\bigg(h\bigg(\sum_{i=1}^n x_i\bigg)-h\bigg(\sum_{i=1}^{n} y_i\bigg)\bigg)\bigg(\sum_{i=1}^n b_i(x_i-y_i)\bigg) \\
+ \bigg(\sum_{i=1}^n (x_i-y_i)\bigg) \bigg( h'\bigg(\sum_{i=1}^n x_i\bigg)\bigg(\sum_{i=1}^n b_i x_i \bigg) - h'\bigg(\sum_{i=1}^n y_i\bigg)\bigg(\sum_{i=1}^n b_i y_i \bigg) \bigg)\ge 0,
\end{multline}
for every $(x_1,\dots,x_n),\,(y_1,\dots, y_n)\in \mathbb{R}^n$. 

\begin{remark}
    Observe that condition \eqref{suff_condition_f} is too strong, as it is not even satisfied for $h(x)=x$. Indeed, in this case, choosing $n=2$, \eqref{suff_condition_f} reads
\[
	\big(x_1+x_2-y_1-y_2\big)(b_1(x_1-y_1)+b_2(x_2-y_2))\ge 0,\quad x_i,\,y_i\in\mathbb{R},
\]
which can be rewritten as 
\[
	b_1(x_1-y_1)^2+b_2(x_2-y_2)^2\ge -(b_1+b_2)(x_1-y_1)(x_2-y_2).
\]
However, this inequality is not satisfied on the entire plane $\mathbb{R}^2.$
\end{remark}

\section{Proof of Theorem \ref{T:existence_beyond_monotonicity} } 
\label{s:proof_existence_beyond_monotonicity}
\subsection{\texorpdfstring{$-\mathcal{J}$ is coercive}{-J is coercive}} \label{ss:coercive}

\begin{lemma} \label{lemma:coercive}
    Let $G$ be an admissible kernel, and let $h \colon \mathbb{R} \to \mathbb{R}$ satisfy Definition \ref{def:admissible_impact_function}(i). 
    Additionally, suppose that one of the following conditions holds.
    \begin{enumerate}
        \item[a)] $h:x \mapsto x$ and $G$ is positive semi-definite;
        \item[b)] Definition \ref{def:admissible_impact_function}(ii) holds;
        \item[c)] the inequality
        \begin{equation*}
            \frac{\gamma}{2} > \norm{h'}_{\infty} \sqrt{T C_{G}}
        \end{equation*}
     is satisfied.
    \end{enumerate}
    Then $- \mathcal{J} \colon \mathcal{L}^2\to \mathbb{R}$ is coercive, i.e., $\lim_{\norm{u}\to \infty } \mathcal{J}(u) = -\infty$. 
\end{lemma}
\begin{proof}
    First of all, if Assumption a) holds, starting from the expression \eqref{eq:gain_functional_alpha_impact_cost_trade_of} of the gain functional $\mathcal{J}$, we readily have
    \begin{align*}
        \mathcal{J}(u) & \le \norm{\alpha-g} \norm{u} + X_{0} \E S_{T} - \frac{\gamma}{2} \norm{u}^{2} - \langle \mathbf{G}u, u \rangle \\
        & \le \norm{\alpha-g} \norm{u} + X_{0} \E S_{T} - \frac{\gamma}{2} \norm{u}^{2},
    \end{align*}
    where the second inequality is consequence of the positive semi-definite property of $G$. Coercivity follows immediately.
    
    Furthermore, starting again from \eqref{eq:gain_functional_alpha_impact_cost_trade_of}, and applying twice the Cauchy-Schwarz inequality, we have
    \begin{align} \notag
        \mathcal{J}(u) & \le \norm{\alpha} \norm{u} + X_{0} \E S_{T} - \frac{\gamma}{2} \norm{u}^{2} - \langle h \left( g_{t} + \mathbf{G}u \right), u \rangle \\ \label{eq:estimate_on_J_for_coercivity}
        & \le \norm{\alpha} \norm{u} + X_{0} \E S_{T} - \frac{\gamma}{2} \norm{u}^{2} + \norm{h \left( g +  \mathbf{G}u \right)} \norm{u}.
    \end{align}

    Suppose Assumption b) holds. By sublinearity of $h$ in the sense of Definition \ref{def:admissible_impact_function}(ii), we fix $0 \leq \zeta < 1$ such that
    \begin{align} \notag
        \norm{h \left( g + \mathbf{G}u \right)} & \leq C_{h} \mathbb{E}\bigg[ \int_{0}^{T} \Big( 1+|g_{t}+\left(\mathbf{G}u\right)_{t}|^\zeta \Big)^2 \dd t \bigg]^{\frac{1}{2}} \\ \notag
        & \leq \sqrt{2} C_{h} \mathbb{E}\bigg[ \int_{0}^{T} \left(1+|g_{t}+\left(\mathbf{G}u\right)_{t}|^{2\zeta} \right) dt \bigg]^{\frac{1}{2}} \\ \notag
        & \leq \sqrt{2} C_{h} \Big( \sqrt{T} + \mathbb{E}[\|g+\mathbf{G}u\|_{L^{2\zeta}}]^{\zeta} \Big) \\ \notag
        & \leq \sqrt{2} C_{h} \Big( \sqrt{T} + T^{\frac{1-\zeta}{2}}\mathbb{E}[\|g+\mathbf{G}u\|_{L^{2}}]^{\zeta} \Big) \\ \notag
        & \leq \sqrt{2} C_{h} \Big( \sqrt{T} + T^{\frac{1-\zeta}{2}}(\|g\| +\|\mathbf{G}u\|)^{\zeta} \Big) \\ \notag
        & \leq \sqrt{2} C_{h} \left( \sqrt{T} + T^{\frac{1-\zeta}{2}} \left( \|g\| + \sqrt{T C_{G}} \|u\| \right)^{\zeta} \right) \\ \label{eq:estimate_on_h}
        & \leq \sqrt{2} C_{h} \left( \sqrt{T} + T^{\frac{1-\zeta}{2}} \left( \|g\|^{\zeta} + \left(T C_{G} \right)^{\frac{\zeta}{2}} \|u\|^{\zeta} \right) \right).
    \end{align}
   Here we obtain the fourth inequality by using the fact that $\|h\|_{L^{2\zeta}}\leq T^{\frac{1-\zeta}{2\zeta}}\|h\|_{L^{2}},\; h\in L^2$ as a direct consequence of Hölder's inequality (with conjugate exponents $\frac{1}{\zeta}, \frac{1}{1-\zeta}$), while the fifth inequality follows by Minkowski's inequality, the sixth one is a consequence of the estimate \eqref{eq:estimate_on_admissible_G_operator} and the last one follows by sub-additivity of $z \mapsto |z|^{\zeta}, \; z \ge 0$.

    Following \eqref{eq:estimate_on_h}, we have
    \begin{equation} \label{eq:asymptotic_behavior_h_of_Z}
        \norm{h \left( g + \left( \mathbf{G}u\right) \right)} = O(\|u\|^{\zeta}) \quad  \text{as } \norm{u} \to \infty,
    \end{equation}
    so that combining \eqref{eq:asymptotic_behavior_h_of_Z} and \eqref{eq:estimate_on_J_for_coercivity}, we get
    \begin{equation*}
        \mathcal{J}(u) \underset{\norm{u} \to \infty}{\sim} - \frac{\gamma}{2} \norm{u}^{2} \underset{\norm{u} \to \infty}{\longrightarrow} - \infty,
    \end{equation*}
    which yields the desired coercivity.

    Finally, by Lipschitz continuity of $h$ and using the estimate \eqref{eq:estimate_on_admissible_G_operator}, we readily obtain
    \begin{align} \notag
        \norm{h \left( g + \mathbf{G}u \right)} & \le \norm{h \left( g \right)} + \norm{h'}_{\infty} \norm{\mathbf{G}u} \\ \label{eq:interm_coercivity_regularize}
        & \le \norm{h \left( g \right)} + \norm{h'}_{\infty} \sqrt{T C_{G}} \norm{u},
    \end{align}
    so that injecting \eqref{eq:interm_coercivity_regularize} into \eqref{eq:estimate_on_J_for_coercivity}, while noting that $1+\zeta<2$, readily yields
    \begin{equation*}
        \mathcal{J}(u) \underset{\norm{u} \to \infty}{\sim} \left( \norm{h'}_{\infty} \sqrt{T C_{G}} - \frac{\gamma}{2} \right) \norm{u}^{2} \underset{\norm{u} \to \infty}{\longrightarrow} - \infty,
    \end{equation*}
    where the divergence holds thanks to Assumption c).
\end{proof}

\subsection{\texorpdfstring{$-\mathcal{J}$ is sequentially weakly lower semi-continuous}{-J is sequentially weakly lower semi-continuous}} \label{ss:wlsc}

We are now interested in proving the sequential weak lower semi-continuity of $-\mathcal{J}$ on $\mathcal{L}^{2}$, that is, given any sequence $(u^n)_n \in \left(\mathcal{L}^{2}\right)^{\mathbb{N}}$ such that $(u^n)_n$ weakly converges to some $u \in \mathcal{L}^{2}$, denoted by $u^{n} \rightharpoonup u$, we aim to show that 
	\begin{equation*}
    	\liminf_{n\to \infty} (-\mathcal{J}(u^{n})) \ge -\mathcal{J}(u).
    \end{equation*}
This property is established in the next lemma, which requires the probability space $\Omega$ to be countable.

\begin{lemma} \label{lemma_wlsc}
	Suppose that $(\Omega, 2^{\Omega},\mathbb{P})$ is a countable probability space, where $2^{\Omega}$ denotes the power set of $\Omega$. If $h \colon \mathbb{R}\to\mathbb{R}$ is Lipschitz continuous and satisfies the strict sublinearity condition in Definition \ref{def:admissible_impact_function} (ii), then $- \mathcal{J}$ is sequentially weakly lower semi-continuous (denoted hereafter by s.w.l.s-c.).
\end{lemma}

\begin{proof}
    Recall from \eqref{eq:gain_functional_alpha_impact_cost_trade_of} that
    \begin{equation*}
        - \mathcal{J}(u) = - \langle \alpha, u \rangle + \frac{\gamma}{2} \norm{u}^{2} + \langle h \left( Z^{u} \right), u \rangle + \frac{\phi}{2} \norm{X^{u}}^{2} + \frac{\varrho}{2} \E \left| X_{T}^{u} \right|^{2} - X_{0} \E S_T, \quad u \in \mathcal{L}^{2}.
    \end{equation*}
   To complete the proof, it is then sufficient to demonstrate that the first five addends in the right-hand side of the equation above are s.w.l.s-c.
    
    First, $\norm{\cdot}^{2}$ is convex and (strongly) continuous, hence s.w.l.s-c., and $\langle \alpha, \cdot \rangle$ is weakly continuous. 
    Furthermore,  $u \mapsto X^{u}$ is a linear bounded operator in $\mathcal{L}^2$, hence it is weakly continuous. Since the operator $\norm{\cdot}^{2}$ is s.w.l.s-c., the corresponding penalization term $\frac{\phi}{2} \norm{X^{u}}^{2}$ is s.w.l.s-c. 
      Analogously, the mapping $u\mapsto X^u_T$ is weakly continuous from $\mathcal{L}^2$ to the space of square integrable, $\mathcal{F}_T-$measurable random variables endowed with the usual norm $(\mathbb{E}[|\cdot|^2])^{1/2}$.  It follows that $\frac{\varrho}{2} \E \left| X_{T}^{u} \right|^{2}$ is s.w.l.s-c., as well. In particular, notice that the four operators discussed so far are weakly lower semi-continuous.

    It remains to study the functional $\left\langle h(Z^\cdot), \cdot \right\rangle$. Fix $(u^{n})_n \in \left(\mathcal{L}^{2}\right)^{\mathbb{N}}$ such that $u^{n}\rightharpoonup u\in\mathcal{L}^2$ and consider, for every $n\in\mathbb{N}$,
    \begin{equation*}
        \left| \left\langle h(Z^{u^{n}}),u^{n}\right\rangle-
    	\left\langle h(Z^u),u\right\rangle\right|\le 
    	\left| \left\langle h(Z^{u}),u^{n}-u\right\rangle\right|
    	+ \left| \left\langle h(Z^{u^{n}}) -h(Z^{u}), u^{n} \right\rangle\right|=:\mathbf{\upperRomannumeral{1}}_n + \mathbf{\upperRomannumeral{2}}_n.
    \end{equation*}
    It is clear that $\lim_{n\to\infty}\mathbf{\upperRomannumeral{1}}_n=0$, hence we only have to focus on $\mathbf{\upperRomannumeral{2}}_n$. By hypothesis, the space $\Omega$ is countable and the probability measure $\mathbb{P}$ is defined on the $\sigma$-algebra $2^{\Omega}$. Therefore, for every $\bar{\omega}\in\Omega$ such that $\mathbb{P}(\{\bar{\omega}\})>0$, from the  weak convergence of $u^{n}$ to $u$ in $\mathcal{L}^{2}$ we infer that, while using the tower property of the conditional expectation, Fubini's Theorem and the fact that $u^{n}$ and $u$ are adapted,  
    \begin{align*}
    \int_{0}^{T} u_{t}^{n}(\bar{\omega}) \psi(t) \dd t & = \frac{1}{\mathbb{P}( \{ \bar{\omega}\})} \mathbb{E}\bigg[ \mathbbm{1}_{\{\bar{\omega}\}}(\omega)\int_{0}^{T} u_{t}^{n}(\bar{\omega}) \psi(t)\dd t \bigg] \\
    & = \frac{1}{\mathbb{P}(\{\bar{\omega}\})} \int_{0}^{T} \mathbb{E}\bigg[ \mathbb{E}_t[\mathbbm{1}_{\{\bar{\omega}\}} \psi(t)](\omega) u_{t}^{n}(\bar{\omega})\bigg] \dd t \\ 
    & = \frac{1}{\mathbb{P}(\{\bar{\omega}\})}  \mathbb{E}\bigg[ \int_{0}^{T} \mathbb{E}_t[\mathbbm{1}_{\{\bar{\omega}\}} \psi(t)](\omega) u_{t}^{n}(\bar{\omega}) \dd t \bigg] \\    
    & \underset{n\to\infty}{\longrightarrow}
    \frac{1}{\mathbb{P}( \{ \bar{\omega}\})}\mathbb{E}\bigg[ \mathbbm{1}_{\{\bar{\omega}\}}(\omega)\int_{0}^{T} u_{t}(\omega) \psi(t)\dd t \bigg]
    = \int_{0}^{T} u_{t}(\bar{\omega}) \psi(t)\dd t, \quad \psi \in L^{2}.
    \end{align*}
    Note that, in the previous step, we consider a progressively measurable version of the process $(t,\omega) \mapsto \mathbb{E}_t[\mathbbm{1}_{\{\bar{\omega}\}}](\omega),$ which exists by, e.g.~the optional projection theorem, hence $\mathbb{E}_.[\mathbbm{1}_{\{\bar{\omega}\}} \phi(.)](\omega) \in \mathcal{L}^{2}$ is a proper test function. This demonstrates that $u^{n}(\omega,\cdot)\rightharpoonup u(\omega,\cdot)$ in $L^{2}$ for $\mathbb{P}-$a.e. $\omega\in \Omega$. Moreover, the restriction of $\mathbf{G}$ to the separable Hilbert space $L^{2}$ is a Hilbert-Schmidt operator. Indeed, denoting by $(e_n)_n$ an orthonormal basis of $L^{2}$, by the monotone convergence theorem and Parseval's identity
     \begin{align*}
     	\sum_{n=1}^{\infty} \norm{\mathbf{G}e_n}_{L^{2}}^2 & = \sum_{n=1}^\infty\int_{0}^{T} \left|\left(\mathbf{G}e_n\right)(t)\right|^2 \dd t \\
        & = \sum_{n=1}^\infty \int_{0}^{T}\bigg|\int_{0}^{t}G(t,s) e_n(s)\dd s\bigg|^2 \dd t \\
        & = \int_{0}^{T} \sum_{n=1}^\infty\left|\langle G(t,\cdot), e_n \rangle_{L^{2}} \right|^2\dd t \le T C_{G} < \infty.
     \end{align*}
     As a result, the restriction of $\mathbf{G}$ to $L^2$ is a compact operator, which in particular maps weakly convergent sequences into strongly convergent ones. Combining the two previous observations with the Lipschitz continuity of $h$ we conclude that
      \begin{equation}\label{strong_Gcompact}
     \lim_{n\to\infty}\norm{h(g_{\cdot}(\omega) + \mathbf{G}u_{\cdot}^{n}(\omega))-h(g_{\cdot}(\omega) + \mathbf{G}u_{\cdot}(\omega))}_{L^{2}}=0,\quad \text{for $\mathbb{P}-$a.e. $\omega\in\Omega$}.
     \end{equation}
     By strict sublinearity of $h$, fix $\zeta \in [0,1)$ such that 
     \begin{equation} \label{eq_strict_sub_for_proof}
        |h(x)| \le C' \left( 1+|x|^\zeta \right), \quad x \in \R, 
    \end{equation}
    is satisfied for some constant $C'>0$. Assume for now $\zeta > 0$ and let us show that the sequence 
    \begin{equation*}
        (\norm{h(g + \mathbf{G}u^{n})-h(g + \mathbf{G}u)}_{L^{2}}^2)_{n} \in \mathbb{R}^{\mathbb{N}}
    \end{equation*} 
     is uniformly integrable. To achieve this, it is enough to show that 
    \begin{equation} \label{est_Vitali}
    \sup_{n\in \mathbb{N}}\mathbb{E}\left[\norm{h(g + \mathbf{G}u^{n})-h(g + \mathbf{G}u)}_{L^{2}}^{\frac{2}{\zeta}}\right]<\infty.
    \end{equation}
    Fix $\omega \in \Omega$. Then, denoting by $C'$ some positive constant, independent of $\omega$ and possibly dependent on $T,h,\zeta$, allowed to change from line to line, we have
    \begin{align} \notag
        \norm{h(g(\omega)+\mathbf{G}u^n(\omega))}_{L^2}^{\frac{2}{\zeta}}
         & \le C'\left(\int_0^T \left(1+|g_t(\omega)|^{2\zeta} + |\left(\mathbf{G}u^n\right)_t(\omega)|^{2\zeta}\right)\dd t \right)^{\frac{1}{\zeta}} \\ \notag
         & \le C'\Bigg( T^{\frac{1}{\zeta}} +\left(\int_0^T |g_t(\omega)|^{2\zeta} \dd t\right)^{\frac{1}{\zeta}}
         +
         \left(\int_0^T |\left(\mathbf{G}u^n\right)_t(\omega)|^{2\zeta} \dd t\right)^{\frac{1}{\zeta}}\Bigg) \\ \label{eq:estimate_intermediate_proof_existence_countable}
         & \le 
         C'\left(1+ \norm{g(\omega)}^2_{L^2}+\norm{\left(\mathbf{G}u^n\right)(\omega)}^2_{L^2}\right),\quad n\in \mathbb{N},
    \end{align}
    where in particular, we used \eqref{eq_strict_sub_for_proof} to get the first inequality, while the last one holds by Jensen's inequality, noting that $1/\zeta>1$. Taking the expected value on \eqref{eq:estimate_intermediate_proof_existence_countable}, given that any weakly convergent sequence is bounded i.e., $\sup_{n\in\mathbb{N}}\norm{u^{n}}^2<\infty$ and applying the estimate \eqref{eq:estimate_on_admissible_G_operator}, we conclude that 
    \[
        \sup_{n\in\mathbb{N}}\mathbb{E}\left[
        \norm{h(g+\mathbf{G}u^n)}_{L^2}^{2\frac{1}{\zeta}}
        \right]
        \le 
        C'\left( 1 + \norm{g}^2 + T C_{G} \sup_{n\in\mathbb{N}}\| u ^n\|^2\right)<\infty.
    \]
    An analogous argument demonstrates that $\mathbb{E}\Big[
    \norm{h(g+\mathbf{G}u)}_{L^2}^{\frac{2}{\zeta}}
    \Big]<\infty$, hence \eqref{est_Vitali} holds. 
    
    Combining \eqref{strong_Gcompact} and \eqref{est_Vitali}, Vitali's convergence theorem yields 
    \begin{equation} \label{eq:h_Z_n_convergence}
        \lim_{n\to\infty}\norm{h(g + \mathbf{G}u^{n})-h(g + \mathbf{G}u)}= 0.
    \end{equation}
    If $\zeta = 0$, then $h$ is bounded and \eqref{eq:h_Z_n_convergence} is immediate by the dominated convergence theorem. Therefore, by the Cauchy-Schwarz inequality,
    \begin{equation*}
        \mathbf{\upperRomannumeral{2}}_n \le \left(\sup_{n\in\mathbb{N}}\norm{u^n}\right)\norm{h(Z^{u^{n}})-h(Z^u)}\underset{n\to \infty}{\longrightarrow}0,
    \end{equation*}
    which proves the desired sequential weak continuity of $\left\langle h(Z^\cdot),\cdot\right\rangle$ and concludes the proof.
\end{proof}

\subsection{Putting everything together} \label{ss:everything_together_existence}

\begin{proof}[Proof of Theorem \ref{T:existence_beyond_monotonicity}]
    Since $\mathcal{L}^{2}$ is a Hilbert space and $-\mathcal{J}$ is coercive and s.w.l.s-c. by Lemma \ref{lemma:coercive}-\ref{lemma_wlsc}, respectively,  an application of \cite[Theorem 1.2, Chapter 1]{Struwe2008} yields the existence of a solution $\hat{u} \in \mathcal{L}^{2}$ satisfying \eqref{eq:optimal_strategy}. Additionally, if we assume  the differentiability of $h$ according to Definition \eqref{def:admissible_impact_function}(i), Fermat's rule for G\^ateaux differentiable functions in Hilbert spaces yields that $\hat{u}$ satisfies the FOC \eqref{eq:nonlinearfredholm}.
\end{proof}

\section{Proof of Proposition \ref{pp:conv_scheme} } \label{ss:proof_conv_scheme}

Recalling the definition of $\tilde{\mathbf{A}}$ in \eqref{eq:Atilde}, the iterates $(u^{[n]})_n$ in \eqref{scheme_update} satisfy, $\left(\dd t \otimes \mathbb{P}\right)-$ a.e.,
    \begin{multline*}
        \gamma u_{t}^{[n]} + \mathbf{A}(u^{[n-1]})_{t} + (\mathbf{G}+
        \mathbf{G}^*)(u^{[n]}-u^{[n-1]})_t
       \\+\mathbf{H}_{\phi,\varrho}u^{[n]}_{t} + \mathbf{H}_{\phi,\varrho}^{*}u^{[n]}_{t} = \alpha_t - X_{0} \left(\phi(T-t)+\varrho\right), \quad n \in \mathbb{N},
    \end{multline*}
    where the kernel $H_{\phi,\varrho}$ is given in \eqref{eq:kernel_H}.
    If we now subtract \eqref{eq:nonlinearfredholm} to the previous equation, and then apply the scalar product of the resulting processes against $u^{[n]}-\hat{u}$, we obtain
    \begin{align*}
        \gamma{\lVert{u^{[n]}-\hat{u}}\rVert}^2= &\left\langle u^{[n]}-\hat{u},\mathbf{A}(\hat{u})-\mathbf{A}(u^{[n-1]})\right\rangle
        -
        \left\langle u^{[n]}-\hat{u},(\mathbf{G}+\mathbf{G}^*)({u}^{[n]}-u^{[n-1]})\right\rangle
        \\ &
        - \left\langle u^{[n]}-\hat{u}, \left(\mathbf{H}_{\phi,\varrho} + \mathbf{H}^\ast_{\phi,\varrho}\right)(u^{[n]}-\hat{u})
        \right\rangle,\quad \mathbb{P}-\text{a.s.}
    \end{align*}
Since, by  hypothesis and  Lemma \ref{thm:conc_conditions}, the operators $\mathbf{G}$ and  $\mathbf{H}_{\phi,\varrho}$ are positive semi-definite,  an application of the Cauchy-Schwarz inequality, \eqref{eq:estimate_on_admissible_G_operator} and \eqref{eq:estimate_on_adjoint}, as well as the triangle inequality with \eqref{eq:def_A_operator}, yields
\begin{align}
&\notag 
\gamma{\lVert{u^{[n]}-\hat{u}}\rVert}  \le 
    {\lVert\mathbf{A}(\hat{u})-\mathbf{A}(u^{[n-1]})\rVert} +2\sqrt{T C_{G}} \lVert\hat{u}-u^{[n-1]}\rVert \\	\notag
    &\qquad  \le {\Big\lVert{h\Big(Z^{\hat{u}}\Big)-h\Big(Z^{u^{[n-1]}}\Big)}\Big\rVert}
        + {\Big\lVert{\mathbf{G}^\ast\Big(h'\Big(Z^{\hat{u}}\Big)\hat{u}-h'\Big(Z^{u^{[n-1]}}\Big)u^{[n-1]}\Big)\Big\rVert}}
        +2\sqrt{T C_{G}} \lVert\hat{u}-u^{[n-1]}\rVert
    \\ \label{eq:conve_proof1}
    &\qquad  =:\mathbf{\upperRomannumeral{1}}_n + \mathbf{\upperRomannumeral{2}}_n
    +
    \mathbf{\upperRomannumeral{3}}_n,\quad n \in \mathbb{N}.
\end{align}
By the Lipschitz continuity of $h$ with constant $\norm{h'}_{\infty}$, recalling the definition \eqref{eq:def_Z} and using the estimate \eqref{eq:estimate_on_admissible_G_operator},
\begin{align*}
\mathbf{\upperRomannumeral{1}}^2_n & \le \norm{h'}_{\infty}^2 \norm{\mathbf{G}(u^{[n-1]}-\hat{u})}
 \le T C_{G} \norm{h'}_{\infty}^2
{\lVert{u^{[n-1]}-\hat{u}}\rVert}^2.
\end{align*}
Next, by the triangle inequality, 
\begin{align*}
    \mathbf{\upperRomannumeral{2}}_n & \leq
    {\Big\lVert{\mathbf{G}^\ast\Big(h'\Big(Z^{u^{[n-1]}}\Big)(\hat{u}-u^{[n-1]})\Big) \Big\rVert}}
    + {\Big\lVert{\mathbf{G}^\ast\Big(\Big(h'\Big(Z^{\hat{u}}\Big)-h'\Big(Z^{u^{[n-1]}}\Big)\Big)\hat{u}\Big) \Big\rVert}} \\
    & =:\mathbf{\upperRomannumeral{2}}_{1,n} + \mathbf{\upperRomannumeral{2}}_{2,n}.
\end{align*} 
Notice that, by boundedness of $h'$ and using the estimate \eqref{eq:estimate_on_adjoint}, we get
\begin{align*}
    \mathbf{\upperRomannumeral{2}}^2_{1,n} \le T C_{G} \norm{h'}_\infty^2 {\lVert{u^{[n-1]}-\hat{u}}\rVert}^2.
\end{align*}
Moreover, letting $L>0$ be the Lipschitz constant of $h'$,  we have
\begin{align*}
    \mathbf{\upperRomannumeral{2}}^2_{2,n} & \le T C_{G} L^2 \norm{\left(\mathbf{G}(\hat{u}-u^{[n-1]})\right) \hat{u}}^{2} \\
    & \le T C_{G}^{2} L^2 \E \bigg[ \int_{0}^{T} \hat{u}_{t}^{2} \dd t \int_{0}^{T} \left(\hat{u}_{s}^{[n-1]} - \hat{u}_{s} \right)^{2} \dd s \bigg]  \le T C_{G}^{2} L^2 M_{\gamma}(\hat{u}) \lVert \hat{u}^{[n-1]} - \hat{u}\rVert^{2},
\end{align*}
where we apply \eqref{eq:estimate_on_adjoint} for the first inequality, Cauchy-Schwarz's inequality as in \eqref{eq:bound_product_term} for the  second, and \eqref{eq:definition_ess_sup} for the last one.

Combining all the above and coming back to \eqref{eq:conve_proof1}, we obtain
\begin{equation*}
    \gamma{\lVert{u^{[n]}-\hat{u}}\rVert} \leq \widetilde{C} {\lVert{u^{[n-1]}-\hat{u}}\rVert},
\end{equation*}
where $\widetilde{C}>0$ is given in \eqref{eq:common_ratio_geometric_sequence}.  Since $\widetilde{C}/\gamma<1$ by \eqref{eq:common_ratio_geometric_sequence}, the previous inequality coincides with \eqref{eq:convergence_rate}, hence 
the proof is complete.

\section{Proofs of Propositions~\ref{pp:stability_with_respect_to_kernel} and \ref{pp:stability_with_respect_to_signals} } \label{s:proofs_stability_results}

\begin{proof}[Proof of Proposition~\ref{pp:stability_with_respect_to_kernel}]
    To start with, we denote by 
    \begin{equation*}
        C_n := \sup_{t\in [0,T]} \int_{0}^{t}\left|G_n(t,s)-G(t,s)\right|^2\dd s, \quad n \in \mathbb{N},
    \end{equation*}
    such that $C_n\to 0$ as $n\to \infty$ by assumption \eqref{eq:condition_stability}. By applying the estimate \eqref{eq:estimate_on_admissible_G_operator}, we readily have
    \begin{equation}\label{estimate_Gnorm}
     	\norm{(\mathbf{G}-\mathbf{G}_n)h}
     	\le \sqrt{TC_n}\norm{h}, \quad h \in \mathcal{L}^{2}.
    \end{equation}
    The required hypotheses enable us to apply Theorem \ref{T:mainnonlinear}, which characterizes the optimal strategies $\hat{u}$ and $\left(\hat{u}_n\right)_{n\in\mathbb{N}}$ as the unique solutions of the nonlinear stochastic Fredholm equation \eqref{eq:nonlinearfredholm}, replacing accordingly $\mathbf{A}$ by $\left(\mathbf{A}_n\right)_{n\in\mathbb{N}}$. Consequently, for every $n \in \mathbb{N}$, $\dd t \otimes \mathbb{P}-$a.e., 
    \begin{equation*}
        \gamma (\left(\hat{u}_{n}\right)_{t}-\hat{u}_t)=\mathbf{A}(\hat{u})(t)-\mathbf{A}_n(\hat{u}_{n})(t)-\left(\mathbf{H}_{\phi,\varrho}\left(\hat{u}_n-\hat{u}\right) \right)_{t}
        - \left(\mathbf{H}^\ast_{\phi,\varrho}\left(\hat{u}_n-\hat{u}\right)\right)_{t}.
    \end{equation*}
    Taking the scalar product in $\mathcal{L}^{2}$ of the previous equation with $\hat{u}_{n}-\hat{u}$ yields
    \begin{equation*}
    	\gamma \norm{\hat{u}_n-\hat{u}}^2
    	= \left\langle\mathbf{A}(\hat{u})-\mathbf{A}_n(\hat{u}_{n}), \hat{u}_{n}-\hat{u}\right\rangle
    	- \left\langle \left(\mathbf{H}_{\phi,\varrho}+\mathbf{H}_{\phi,\varrho}^{\ast}\right) (\hat{u}_{n}-\hat{u}),\hat{u}_{n}-\hat{u} \right\rangle.
    \end{equation*}
    Observing again that $\mathbf{H}_{\phi,\varrho}$ is a positive semi-definite operator in the sense of \eqref{eq:sdp_operator_def} by Lemma \ref{thm:conc_conditions}, recalling also the monotonicity of $\mathbf{A}_n$ in \eqref{eq:definition_monotone} we deduce that 
    \begin{align*}
    	\gamma \norm{\hat{u}_n-\hat{u}}^2
    	& \le \left\langle(\mathbf{A}(\hat{u})-\mathbf{A}_n(\hat{u}))+(\mathbf{A}_n(\hat{u})-\mathbf{A}_n(\hat{u}_{n})), \hat{u}_{n}-\hat{u}\right\rangle
    	\\
        & \le \left\langle\mathbf{A}(\hat{u})-\mathbf{A}_n(\hat{u}), \hat{u}_{n}-\hat{u}\right\rangle \\
        & = \left\langle h(g+\mathbf{G}(\hat{u}))-h(g+\mathbf{G}_{n}(\hat{u})),  \hat{u}_{n}-\hat{u} \right\rangle + \left\langle h'(g+\mathbf{G}(\hat{u}))\hat{u},  \left( \mathbf{G} - \mathbf{G}_{n}\right) \left( \hat{u}_{n}-\hat{u} \right)\right\rangle \\
        & \quad + \left\langle \left(h'(g+\mathbf{G}(\hat{u}))-h'(g+\mathbf{G}_n(\hat{u}))\right) \hat{u}, \mathbf{G}_n(\hat{u}_n-\hat{u}) \right\rangle,
    \end{align*}
    where the last equality follows from the definitions of the operators $(\mathbf{A}_n)_n$ and $\mathbf{A}$ in \eqref{eq:def_A_operator} and an application of Fubini.
    Now, we have by the Cauchy-Schwarz inequality
    \begin{align} \notag 
        \gamma \norm{\hat{u}_n-\hat{u}}^2
    	& \le \norm{h(g+\mathbf{G}(\hat{u}))-h(g+\mathbf{G}_{n}(\hat{u}))}\norm{\hat{u}_n-\hat{u}} \\ \notag
        & \quad + \norm{h'(g+\mathbf{G}(\hat{u}))\hat{u}}\norm{\left( \mathbf{G} - \mathbf{G}_{n}\right) \left( \hat{u}_{n}-\hat{u} \right)} \\ \notag
        & \quad + \left\langle
    	\left(h'(g+\mathbf{G}(\hat{u}))-h'(g+\mathbf{G}_n(\hat{u}))\right) \hat{u}, \mathbf{G}_n(\hat{u}_n-\hat{u}) \right\rangle \\ \label{est_proof_stability}
        &=: \mathbf{\upperRomannumeral{1}}_n + \mathbf{\upperRomannumeral{2}}_n + \mathbf{\upperRomannumeral{3}}_n.
    \end{align}
    Using the fact that $h$ is Lipschitz continuous with constant $\norm{h'}_\infty$ and the estimate \eqref{estimate_Gnorm}, we have
    \begin{equation} \label{est_1}
        \mathbf{\mathbf{\upperRomannumeral{1}}}_n \le \norm{h'}_\infty \norm{(\mathbf{G}-\mathbf{G}_n)(\hat{u})} \norm{\hat{u}_n-\hat{u}} \le \norm{h'}_\infty \sqrt{TC_n}
        \norm{\hat{u}} \norm{\hat{u}_n-\hat{u}}.
    \end{equation}
    As for $\mathbf{\mathbf{\upperRomannumeral{2}}}_n$ in \eqref{est_proof_stability}, by boundedness of $h'$, it is immediate to deduce from \eqref{estimate_Gnorm} that 
    \begin{equation} \label{est_2}
    	\mathbf{\upperRomannumeral{2}}_n \le \norm{h'}_\infty \sqrt{TC_n} \norm{\hat{u}}\norm{\hat{u}_n-\hat{u}}.
    \end{equation}
    Finally, applying Hölder's inequality, the Lipschitz continuity of $h'$ and using again the estimate \eqref{estimate_Gnorm}, we get
    \begin{align*}
        \mathbf{\upperRomannumeral{3}}_n & = \left\langle
    	\left(h'(g+\mathbf{G}(\hat{u}))-h'(g+\mathbf{G}_n(\hat{u}))\right) \hat{u}, \left(\mathbf{G}_n-\mathbf{G}\right)(\hat{u}_n-\hat{u})
    	\right\rangle \\
        & \quad + \left\langle
    	\left(h'(g+\mathbf{G}(\hat{u}))-h'(g+\mathbf{G}_n(\hat{u}))\right) \hat{u}, \mathbf{G}(\hat{u}_n-\hat{u})
    	\right\rangle \\
        & =: \mathbf{\upperRomannumeral{3}}_{1,n} + \mathbf{\upperRomannumeral{3}}_{2,n}
    \end{align*}
    Observe that
    \begin{equation} \label{estimate:rough}
    	\norm{\left(h'(g+\mathbf{G}(\hat{u}))-h'(g+\mathbf{G}_n(\hat{u}))\right) \hat{u}} \le \sqrt{2}\norm{h'}_\infty\norm{\hat{u}}.
    \end{equation}
    Then, by the Cauchy--Schwarz inequality and using the estimates \eqref{estimate_Gnorm} and \eqref{estimate:rough}, we obtain
    \begin{equation} \label{est_3_1}
    	\mathbf{\mathbf{\upperRomannumeral{3}}}_{1,n} \le \sqrt{2TC_n} \norm{h'}_\infty \norm{\hat{u}} \norm{\hat{u}_{n}-\hat{u}}.
    \end{equation}
    Furthermore, since $G$ is admissible according to Definition \ref{def:admissible_kernel}, by the Cauchy-Schwarz inequality, \eqref{eq:estimate_on_admissible_G_operator} and \eqref{estimate:rough} we infer that 
    \begin{align} \label{eq:est_to_be_refined}
    	\mathbf{\mathbf{\upperRomannumeral{3}}}_{2,n} & \le \sqrt{TC_{G}} \norm{\left(h'(g+\mathbf{G}(\hat{u}))-h'(g+\mathbf{G}_n(\hat{u}))\right) \hat{u}} \norm{\hat{u}_{n}-\hat{u}} \\ \label{est_3_2_rough}
        & \le \sqrt{2T C_{G}} \norm{h'}_\infty \norm{\hat{u}} \norm{\hat{u}_{n}-\hat{u}}.
    \end{align}
    Combining all the above estimates \eqref{est_1}--\eqref{est_2}--\eqref{est_3_1}--\eqref{est_3_2_rough} into \eqref{est_proof_stability}, we eventually obtain
    \begin{equation} \label{eq:rough_estimate_norm_diff}
        \norm{\hat{u}_n-\hat{u}} \le \frac{\sqrt{2T} \norm{h'}_\infty \norm{\hat{u}}}{\gamma} \left( \sqrt{C_n} \left( 1 + \sqrt{2} \right) + \sqrt{C_{G}} \right) < \infty,
    \end{equation}
    hence the sequence $(\hat{u}_n)_n$ is bounded in $\mathcal{L}^{2}$, which is a necessary but not sufficient condition to ensure $\norm{\hat{u}_n-\hat{u}} \underset{n \to \infty}{\longrightarrow} 0$.
    To achieve such convergence, first notice that, by the Cauchy-Schwarz inequality,  
    \begin{align} \notag
        \norm{\left( \left( \mathbf{G} - \mathbf{G}_{n} \right) \hat{u} \right) \hat{u}}^{2} & = \E \bigg[ \int_{0}^{T} \left( \int_{0}^{t} \left( G(t,s) - G_{n}(t,s) \right) \hat{u}_{s} \dd s \right)^{2} \hat{u}_{t}^{2} \dd t \bigg] \\ \notag
        & \leq \E \bigg[ \int_{0}^{T} \left( \int_{0}^{t} \left| G(t,s) - G_{n}(t,s) \right|^{2} \dd s \int_{0}^{t} \hat{u}_{s}^{2} \dd s \right) \hat{u}_{t}^{2} \dd t \bigg] \\ \label{eq:bound_product_term}
        & \le C_{n} \mathbb{E} \bigg[ \bigg(\int_{0}^{T} \hat{u}_{t}^{2} \dd t \bigg)^{2} \bigg].
    \end{align}
  We now refine the estimate \eqref{eq:est_to_be_refined}. Leveraging the Lipschitz continuity of $h'$ with constant $L>0$ instead of using the rough estimate \eqref{estimate:rough}, it readily follows from  \eqref{eq:bound_product_term} that 
    \begin{align} \notag
        \mathbf{\mathbf{\upperRomannumeral{3}}}_{2,n} & \le L \sqrt{T C_{G}} \norm{\left( \left( \mathbf{G} - \mathbf{G}_{n} \right) \hat{u} \right) \hat{u}} \norm{\hat{u}_{n}-\hat{u}} \\ \label{est_3_2_refined}
        & \le L \sqrt{T C_{G} C_{n} \mathbb{E} \bigg[ \bigg(\int_{0}^{T} \hat{u}_{t}^{2} \dd t \bigg)^{2} \bigg]} \norm{\hat{u}_{n}-\hat{u}}.
    \end{align}
    Then, using \eqref{est_3_2_refined} instead of \eqref{est_3_2_rough} for $\mathbf{\mathbf{\upperRomannumeral{3}}}_{2,n}$, \eqref{eq:rough_estimate_norm_diff} becomes 
    \begin{equation*}
        \norm{\hat{u}_n-\hat{u}} \le \frac{\sqrt{2TC_{n}} \norm{h'}_\infty \norm{\hat{u}}}{\gamma} \Bigg( 1 + \sqrt{2} + L \sqrt{ \frac{1}{2} C_{G} \mathbb{E} \bigg[ \bigg(\int_{0}^{T} \hat{u}_{t}^{2} \dd t \bigg)^{2} \bigg]} \Bigg) \underset{n \to \infty}{\longrightarrow} 0,
    \end{equation*}
    where the convergence holds due to the assumption \eqref{eq:assumption_on_optimal_strategy_for_stability}, concluding the proof.
\end{proof}

\begin{proof}[Proof of Proposition~\ref{pp:stability_with_respect_to_kernel}]
    By definition of the gain functional in \eqref{eq:gain_functional}, the Cauchy-Schwarz inequality, and the Lipschitz continuity of $h$, there is a constant $C'>0$ such that 
    \begin{equation}\begin{aligned}\label{eq:J_n-J}
        \big|\mathcal{J}_n(u)-\mathcal{J}(u)\big|&=\langle\alpha_n-\alpha,u\rangle+\langle h(g+\mathbf{G}u)-h(g_n+\mathbf{G}u)\alpha,u\rangle\\
        &\leq \|\alpha_n-\alpha\|\|u\|+C'\|g-g_n\|\|u\|.
    \end{aligned}\end{equation}
    Next, due to the convergence assumption, $(\alpha_n)_n$ and $(g_n)_n$ are bounded  in $\mathcal{L}^2$, that is,
    $$ \sup_n \|\alpha_n\|<\infty, \quad  \sup_n \|g_n\|<\infty.$$
    Therefore, it follows from \eqref{eq:estimate_on_J_for_coercivity} and \eqref{eq:estimate_on_h} with $\alpha$ and $g$ replaced by $\alpha_n$ and $g_n$, that there is a constant $M>0$ such that 
    $$
    \sup_{u\in\mathcal{L}^2}\mathcal{J}(u)=\sup_{\substack{u\in\mathcal{L}^2\\\|u\|\leq M}}\mathcal{J}(u),\quad \sup_{u\in\mathcal{L}^2}\mathcal{J}_n(u)=\sup_{\substack{u\in\mathcal{L}^2\\\|u\|\leq M}}\mathcal{J}_n(u),\quad\text{for all }n\in\mathbb{N}.
    $$
    Hence, \eqref{eq:J_n-J} implies that 
    \begin{align*}
        \Big|\sup_{u\in\mathcal{L}^2}\mathcal{J}_n(u)-\sup_{u\in\mathcal{L}^2}\mathcal{J}(u)\Big| & = \Big|\sup_{\substack{u\in\mathcal{L}^2 \\|u\|\leq M}}\mathcal{J}_n(u)-\sup_{\substack{u\in\mathcal{L}^2\\\|u\|\leq M}}\mathcal{J}(u)\Big|\\
        & \leq \sup_{\substack{u\in\mathcal{L}^2\\\|u\|\leq M}}\Big|\mathcal{J}_n(u)-\mathcal{J}(u)\Big|\\
        &\leq  M\big(\|\alpha_n-\alpha\|+C''\|g-g_n\|\big)\underset{n\to\infty}{\longrightarrow}0,
    \end{align*}
    which completes the proof.
\end{proof}

\section{Proof of Theorem \ref{T:existence_uniqueness_beyond_monotonicity} } \label{s:proof_existence_uniqueness_beyond_monotonicity}

 In the following lemma, we establish an \emph{a priori} estimate on the $L^2$-norm of any solution to \eqref{eq:nonlinearfredholm}, which will enable us to deduce properties of a solution to \eqref{eq:nonlinearfredholm} from the signals $\alpha$ and $g$.

\begin{lemma} \label{a_priori_lemma}
Suppose that 
\begin{equation*}
    C_{G}^{\ast} := \sup_{t \in [0,T]} \int_{t}^{T} \left| G(s,t) \right|^{2} \dd s < \infty.
\end{equation*}
Let $h\colon\mathbb{R}\to\mathbb{R}$ be differentiable with bounded derivative $h'$, and  assume that the following bound holds:
\begin{equation}\label{a_priori_precise}
    \gamma > 2\sqrt{T} \max \left\{\widetilde{C}_{H,G}, \; \sqrt{C_{H}} + \frac{1}{2}\norm{h'}_{\infty}
    \left(\sqrt{C_{G}} + \sqrt{T C_{G}^{\ast}}
     \left( \gamma - 2\sqrt{T} \widetilde{C}_{H,G} \right)^{-1} \widetilde{C}_{H,G}
    \right) \right\},
\end{equation}
with
\begin{equation} \label{eq:def_c_h_k_constant}
    \widetilde{C}_{H,G} := \sqrt{C_{H}} + \norm{h'}_\infty \sqrt{C_{G}},
\end{equation}
where the constants $C_{H},\,C_{G}$ associated with the admissible kernels $H_{\phi,\varrho}, \; G$, respectively, are defined in \eqref{eq:constant_norm_of_G}.
Then a solution $\hat{u}\in\mathcal{L}^2$ of \eqref{eq:nonlinearfredholm} satisfies $\mathbb{P}-$a.s. the a priori estimate
\begin{align} \notag
    \norm{\hat{u}}_{L^2} \le \frac{1}{\widetilde{C}_\gamma}
    \Bigg( & \norm{\tilde \alpha}_{L^2} + \norm{h(g)}_{L^2} + \norm{h'}_\infty \sqrt{C_{G}^{\ast}} \left( \gamma - 2 \sqrt{T} \widetilde{C}_{H,G} \right)^{-1} \\ \label{est_apriori}
    & \times \bigg( \bigg(\int_{0}^{T}\mathbb{E}_t\Big[ \norm{\tilde \alpha}_{L^2}^2 \Big]\dd t \bigg)^{\frac{1}{2}} + \bigg( \int_{0}^{T}\mathbb{E}_t\Big[ \norm{h(g)}_{L^2}^2 \Big]\dd t \bigg)^{\frac{1}{2}} \bigg) \Bigg),
\end{align}
where 
\begin{equation} \label{eq:effective_signal_with_soft_penalization}
    \tilde \alpha_v := \alpha_v - X_{0} \left( \phi(T-v) + \varrho \right), \quad v \in [0,T],
\end{equation}
and
\begin{equation} \label{constant_Ctilde}
    \widetilde{C}_\gamma := \gamma - \sqrt{T} \left( 2\sqrt{C_{H}} + \norm{h'}_{\infty} \left( \sqrt{C_{G}} + \sqrt{T C_{G}^{\ast}} \left( \gamma - 2 \sqrt{T} \widetilde{C}_{H,G} \right)^{-1} \widetilde{C}_{H,G} \right) \right).
\end{equation}
\end{lemma}

\begin{proof}
    Fix a solution $\hat{u} \in \mathcal{L}^{2}$ to the nonlinear Fredholm equation \eqref{eq:nonlinearfredholm}. Note that \eqref{eq:nonlinearfredholm} re-writes as
    \begin{equation*}
        \gamma \hat{u}_v = \tilde \alpha_v - \left(\mathbf A(\hat{u})\right)_{v} - \left( \mathbf{H}_{\phi,\varrho}\hat{u} \right)_{v}
        - \left( \mathbf{H}_{\phi,\varrho}^{*}\hat{u} \right)_{v} \quad \left( \dd t \otimes \mathbb{P} \right)-a.e.,
    \end{equation*}
    where $\tilde \alpha$ is given by \eqref{eq:effective_signal_with_soft_penalization}, so that applying the absolute value together with the triangular inequality yields
    \begin{equation} \label{eq:estimate_by_abs_value}
        \gamma \left|\hat{u}_v\right| \le \left|\tilde \alpha_v \right| + \left|\left(\mathbf A(\hat{u})\right)_{v}\right| + \left( \mathbf{H}_{\phi,\varrho}|\hat{u}| \right)_{v}
        + \left( \mathbf{H}_{\phi,\varrho}^{*}|\hat{u}| \right)_{v} \quad \left( \dd t \otimes \mathbb{P} \right)-a.e.
    \end{equation}
    By boundedness of $h'$, $h$ is Lipschitz continuous with constant $\norm{h'}_{\infty}$ satisfying
    \begin{equation} \label{eq:lipschitz_continuity_h}
        \left| h \left( g + \mathbf{G}u \right) \right| \le \left| h \left( g \right) \right| + \norm{h'}_{\infty} \left| \mathbf{G}u \right|,
    \end{equation}
    such that using the definition of $\mathbf{A}$ from \eqref{eq:def_A_operator}, we readily get
    \begin{equation} \label{eq:estimate_on_A_lipschitz}
        \left|\left(\mathbf A(\hat{u})\right)_{v}\right| \le \left|h\left( g_{v} \right)\right| + \norm{h'}_{\infty} \left( \left(\mathbf{G} |\hat{u}| \right)_{v} + \left(\mathbf{G}^{\ast} |\hat{u}| \right)_{v} \right) \quad \left( \dd t \otimes \mathbb{P} \right)-a.e.
    \end{equation}
    Now, combining \eqref{eq:estimate_by_abs_value} and \eqref{eq:estimate_on_A_lipschitz} yields
    \begin{align} \notag
        \gamma \left|\hat{u}_v\right| \le & \left|\tilde \alpha_v\right| + \left|h\left( g_{v} \right)\right| + \norm{h'}_{\infty} \left( \left(\mathbf{G} |\hat{u}| \right)_{v} + \left(\mathbf{G}^{\ast} |\hat{u}| \right)_{v} \right) \\ \label{eq:estimate_before_massacre}
        & + \left( \mathbf{H}_{\phi,\varrho}|\hat{u}| \right)_{v}
        + \left( \mathbf{H}_{\phi,\varrho}^{*}|\hat{u}| \right)_{v} \quad \left( \dd t \otimes \mathbb{P} \right)-a.e.
    \end{align}
   Fix $t \in [0,T]$ and define 
    \begin{equation*}
        m_{v}^{(t)} := \mathbbm{1}_{\{v>t\}}\mathbb{E}_t{[|\hat{u}_v|]},\quad v\in[0,T].
    \end{equation*}
    Then, taking the conditional expectation $\mathbbm{1}_{\{v>t\}} \mathbb{E}_t[\cdot]$ on the last estimate \eqref{eq:estimate_before_massacre} gives, thanks to the tower property,
    \begin{align} \notag
        \gamma m_{v}^{(t)} \le \; & Y_{v}^{(t)} + \mathbbm{1}_{\{v>t\}} \left(\left(\norm{h'}_{\infty} \mathbf{G}^{(t)} + \mathbf{H}^{(t)}_{\phi,\rho} \right) m^{(t)}\right)_{v} \\ \label{eq_estimate_intermediary_m_t}
        & + \left(\left( \norm{h'}_{\infty} \left(\mathbf{G}^{(t)}\right)^{\ast} + \left(\mathbf{H}^{(t)}_{\phi,\rho}\right)^{\ast} \right) m^{(t)}\right)_{v}, \quad \left( \dd t \otimes \mathbb{P} \right)-a.e.
    \end{align}
    Here $\mathbf{G}^{(t)}$ (resp.~$\mathbf{H}^{(t)}_{\phi,\rho}$) is the linear integral $\mathcal{L}^2-$operator associated with the kernel 
    $$G^{(t)}(v,s) := \mathbbm{1}_{\{s>t\}}G(v,s)\quad  \text{(resp.~$H^{(t)}_{\phi,\rho}(v,s) := \mathbbm{1}_{\{s>t\}}H_{\phi,\rho}(v,s)$)},$$ for $s,v\in [0,T]$, and $Y^{(t)}=(Y^{(t)}_v)_{v\ge0}\in \mathcal{L}^2$ is the auxiliary process defined by 
    \begin{align} \notag
        Y^{(t)}_v := \mathbbm{1}_{\{v>t\}} \Big( & \mathbb{E}_t[|\tilde \alpha_v|]+\mathbb{E}_t[|h(g_v)|] + \int_{0}^{t} H_{\phi,\rho}(v,s)|\hat{u}_s|\dd s \\ \label{def:aux_Y}
        & + \norm{h'}_\infty \int_{0}^{t} G(v,s)|\hat{u}_s|\dd s
        \Big),\quad v\in [0,T].
    \end{align}
    Taking the $L^2([0,T])-$norm in the previous estimate \eqref{eq_estimate_intermediary_m_t}, similarly to  \eqref{eq:estimate_on_admissible_G_operator} and \eqref{eq:dual_operator_def},
    \begin{equation*}
        \gamma\lVert{m^{(t)}}\rVert_{L^2}\le \lVert{Y^{(t)}}\rVert_{L^2} + 
        2 \sqrt{T} \widetilde{C}_{H,G} \lVert{m^{(t)}}\rVert_{L^2}, \quad \mathbb{P}-\text{a.s.},
    \end{equation*}
    where $\widetilde{C}_{H,G}$ is given by \eqref{eq:def_c_h_k_constant}. Considering also the joint measurability in $[0,T] \times \Omega$ of the processes $(t,\omega) \mapsto \norm{m^{(t)}(\omega)}_{L^2}$ and $(t,\omega) \mapsto \norm{Y^{(t)}(\omega)}_{L^2}$, it follows from assumption \eqref{a_priori_precise} that
    \begin{equation}\label{est_mt}
        \lVert{m^{(t)}}\rVert_{L^2} \le \left(\gamma-2\sqrt{T} \widetilde{C}_{H,G} \right)^{-1} \lVert{Y^{(t)}}\rVert_{L^2}, \quad \text{for a.e. }t\in [0,T], \quad \mathbb{P}-\text{a.s.}
    \end{equation}
    Furthermore, from \eqref{def:aux_Y} we compute, also employing conditional Jensen's and Jensen's inequalities
    \begin{equation} \label{eq:estimate_on_Y}
        \lVert{Y^{(t)}}\rVert_{L^2} \le \left( \mathbb{E}_t \left[ \norm{\tilde \alpha}_{L^2}^2 \right] \right)^{\frac{1}{2}} + \left( \mathbb{E}_t \left[ \norm{h(g)}_{L^2}^2 \right] \right)^{\frac{1}{2}} + \sqrt{T} \widetilde{C}_{H,G} \norm{\hat{u}}_{L^2} \quad \left(\dd t\otimes \mathbb{P}\right)-\text{a.e.}
    \end{equation}
    In particular, injecting the estimate \eqref{eq:estimate_on_Y} into \eqref{est_mt}, together with the triangular inequality yields
    \begin{align} \notag
        \left(\int_{0}^{T} \lVert{m^{(t)}}\rVert_{L^2}^2\dd t \right)^{\frac{1}{2}} \le & \left(\gamma-2\sqrt{T} \widetilde{C}_{H,G} \right)^{-1} \Bigg( \left( \int_{0}^{T}\mathbb{E}_t \Big[ \norm{\tilde \alpha}_{L^2}^2 \Big]\dd t \right)^{\frac{1}{2}} \\ \label{eq:estimate_integral_m_t}
        & + \left( \int_{0}^{T}\mathbb{E}_t\Big[ \norm{h(g)}_{L^2}^2 \Big]\dd t \right)^{\frac{1}{2}} + T \widetilde{C}_{H,G} \norm{\hat{u}}_{L^2} \Bigg).
    \end{align}
    Furthermore, by the estimate \eqref{eq:lipschitz_continuity_h} and the Cauchy--Schwarz inequality, from \eqref{eq:def_A_operator} we obtain $\mathbb{P}$ almost surely the following estimate:
    \begin{equation} \label{eq:estimate_on_A_L_2_norm}	
        \norm{\mathbf A (\hat{u})}_{L^2} \le \norm{h(g)}_{L^2} + \sqrt{T C_{G}}\norm{h'}_\infty \norm{\hat{u}}_{L^2} + \norm{h'}_\infty \sqrt{C_{G}^{\ast}} \left( \int_{0}^{T} \lVert{m^{(t)}}\rVert^2_{L^2} \dd t \right)^{\frac{1}{2}}.
    \end{equation}
    Starting from \eqref{eq:nonlinearfredholm}, taking the $L^2$-norm, injecting \eqref{eq:estimate_on_A_L_2_norm} and then \eqref{eq:estimate_integral_m_t} yields successively
    \begin{align*}
        \gamma \norm{\hat{u}}_{L^2} & \le \norm{\tilde \alpha}_{L^2} + \norm{\mathbf A (\hat{u})}_{L^2} + 2 \sqrt{T C_{H}} \norm{\hat{u}}_{L^2} \\
        & \le \norm{\tilde \alpha}_{L^2} + \norm{h(g)}_{L^2} + \sqrt{T} \left( 2\sqrt{C_{H}} + \norm{h'}_{\infty} \sqrt{C_{G}} \right) \norm{\hat{u}}_{L^2} \\
        & \quad + \norm{h'}_\infty\sqrt{C_{G}^{\ast}}	\bigg(\int_{0}^{T} \lVert{m^{(t)}}\rVert_{L^2}^2\dd t\bigg)^{\frac{1}{2}} \\ 
        & \le \norm{\tilde \alpha}_{L^2} + \norm{h(g)}_{L^2} + \norm{h'}_\infty \sqrt{C_{G}^{\ast}} \left( \gamma - 2 \sqrt{T} \widetilde{C}_{H,G} \right)^{-1} \\ 
        & \quad \times \bigg( \bigg(\int_{0}^{T}\mathbb{E}_t\Big[ \norm{\tilde \alpha}_{L^2}^2 \Big]\dd t \bigg)^{\frac{1}{2}} + \bigg( \int_{0}^{T}\mathbb{E}_t\Big[ \norm{h(g)}_{L^2}^2 \Big]\dd t \bigg)^{\frac{1}{2}} \bigg) \\
        & \quad + \sqrt{T} \bigg( 2\sqrt{C_{H}} + \norm{h'}_{\infty} \Big( \sqrt{C_{G}} + \sqrt{T C_{G}^{\ast}} \left( \gamma - 2 \sqrt{T} \widetilde{C}_{H,G} \right)^{-1} \widetilde{C}_{H,G} \Big) \bigg) \norm{\hat{u}}_{L^2},
    \end{align*}
    which holds $\mathbb{P}-$a.s. By \eqref{a_priori_precise} and \eqref{constant_Ctilde}, this inequality gives the desired a priori estimate \eqref{est_apriori}, completing the proof.
\end{proof}

Thanks to Lemma \ref{a_priori_lemma}, the following corollary is immediate.
\begin{corollary} \label{cor_Linfinity}
    Under the hypotheses of Lemma \ref{a_priori_lemma}, if $\alpha,\,g\in\mathcal{L}^{\infty} ( \Omega, L^{2}([0,T]))$, that is,
    \begin{equation*}
        \esssup_{\omega\in\Omega}\norm{\alpha(\omega)}_{L^2}<\infty \quad \text{ and }\quad  \esssup_{\omega\in\Omega}\norm{g(\omega)}_{L^2}<\infty,
    \end{equation*}
    then any solution of \eqref{eq:nonlinearfredholm} belongs to the space $ \mathcal{L}^\infty(\Omega;L^2([0,T]))$. 
\end{corollary}

\begin{proof}[Proof of Theorem \ref{T:existence_uniqueness_beyond_monotonicity}]
    Theorem \ref{T:existence_beyond_monotonicity} ensures the existence of an optimal trading strategy $\hat{u}\in\mathcal{L}^2$ that satisfies the nonlinear stochastic Fredholm equation \eqref{eq:nonlinearfredholm}.
    For every  $\gamma$ sufficiently large, following the a priori estimate \eqref{est_apriori} in Lemma \ref{a_priori_lemma}, and an application of Corollary \ref{cor_Linfinity} under assumption \eqref{eq:alpha_infinity}, any solution of \eqref{eq:nonlinearfredholm} belongs to $\mathcal{L}^{\infty} \left( \Omega, L^{2}([0,T]) \right)$ and satisfies \eqref{eq:common_ratio_geometric_sequence}. Hence, Proposition \ref{pp:conv_scheme} applies and guarantees the convergence of the iterative scheme \eqref{scheme_initialization}-\eqref{scheme_update} to $\hat{u}$, which is then unique solution in $\mathcal{L}^2$. This completes the proof.
\end{proof}

\appendix

\section{Empirical PnL and error with the Least Squares Monte Carlo (LSMC)} \label{s:pnl_error_scheme}

For this section, unless stated otherwise,  in the gain functional $\mathcal{J}$ from \eqref{eq:gain_functional_alpha_impact_cost_trade_of} we set: $\gamma = 1$, $\phi = \varrho = 0$ and $G := G_{\nu, \epsilon}$, the fractional kernel from \eqref{eq:def_shifted_frac_kernel} with $\xi = 1$, $\epsilon=0$ and $\nu=0.6$. The impact function is  $h_{x_{0},c}$ from \eqref{eq:impact_function_specification}  with $(x_{0}, c) = (0.1, 0.6)$, and the stochastic signal is parametrized as the integral of an Ornstein-Uhlenbeck process (see \eqref{eq:drift_signal_specification}) with volatility $\tilde \xi = 5$, mean long-term level $\tilde \theta = 40$, mean-reversion speed $\tilde \kappa=5$ and initial value $\tilde I_0 = 10$. We consider $M=10 000$ trajectories (including antithetic paths) and  $100$ time-steps.

We also define the following state variables: 
\begin{align*}
    X^{\tilde\kappa} & := \int_{0}^{t} e^{-\tilde\kappa (t-s)} \alpha_{s} \dd s \\
    X^{i} & := \xi_{i} \int_{0}^{t} e^{-x_{i} (t-s)} \alpha_{s} \dd s, \; i \in \{1,\cdots,5\}, \\
    X^{frac, ref} & := \int_{0}^{t} G_{\nu, \epsilon}(t,s) \alpha_{s} \dd s
\end{align*}
where $(\xi_{i}, x_{i})_{i \in \{1,\cdots,5\}}$ are given in Table~\ref{tab:optimal_weights_mean_reversions_approx_shifted_fractional_kernel}.

Notice that both solving the linear Fredholm equation \eqref{scheme_update} as well as computing the empirical error metric $E^{N}\left( u^{[n]} \right), \; n \geq 1$ \eqref{eq:empirical_error_metric} depend in general on how well the conditional expectations in these expressions are respectively estimated. In the particular case of the LSMC technique, the quality of estimation depends on:
\begin{enumerate}
    \item the choice of the regression variables: the regression basis $\left( \alpha, \int_{0}^{.} \alpha_{s} \dd s, X^{\tilde{\kappa}} \right)$ displays the best PnL maximization and error minimization results in Figure \ref{fig:impact_regression_basis}; 
    
    \item the number $M \in \mathbb{N}^{*}$ of sample trajectories: the more trajectories, the better the convergence as illustrated in Figure \ref{fig:impact_nber_trajs};
    
    \item the choice of the expansion basis, which we specify as a family of orthonormal polynomials, and has no significant impact on the quality of convergence in our case, when testing Chebyshev, Legendre, Laguerre, Hermite polynomials. The maximum degree $d$ of the basis from \eqref{eq:basis_expansion_definition} has also a minor impact in all the provided examples: by default we use $d=2$.
\end{enumerate}

\begin{figure}[H]
    \centering
    \begin{subfigure}[t]{0.48\textwidth}
        \centering
        \includegraphics[width=\textwidth, trim=0cm 0cm 0cm 0cm, clip]{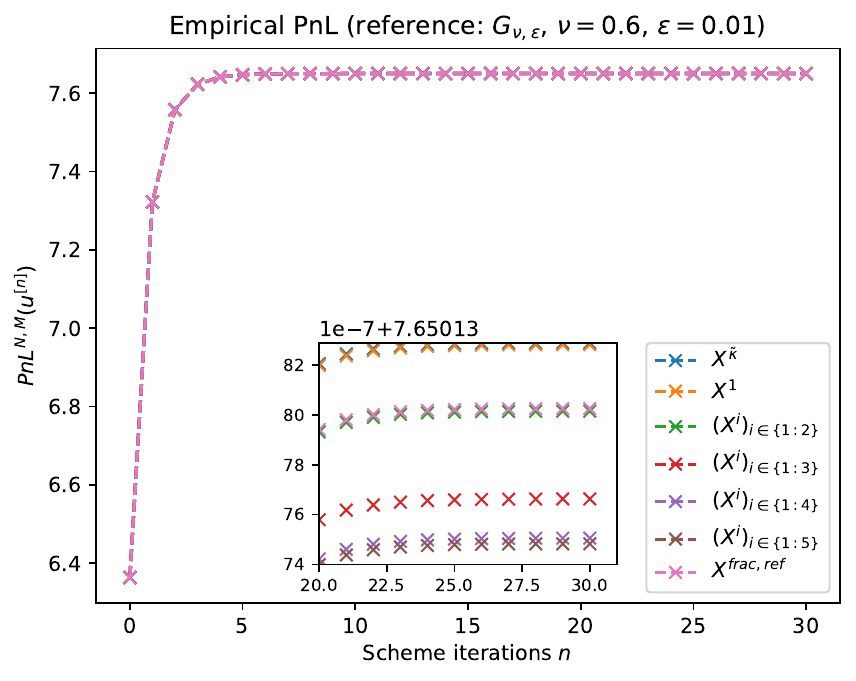} 
    \end{subfigure}
    \hfill 
    \begin{subfigure}[t]{0.48\textwidth}
        \centering
        \includegraphics[width=\textwidth, trim=0cm 0cm 0cm 0cm, clip]{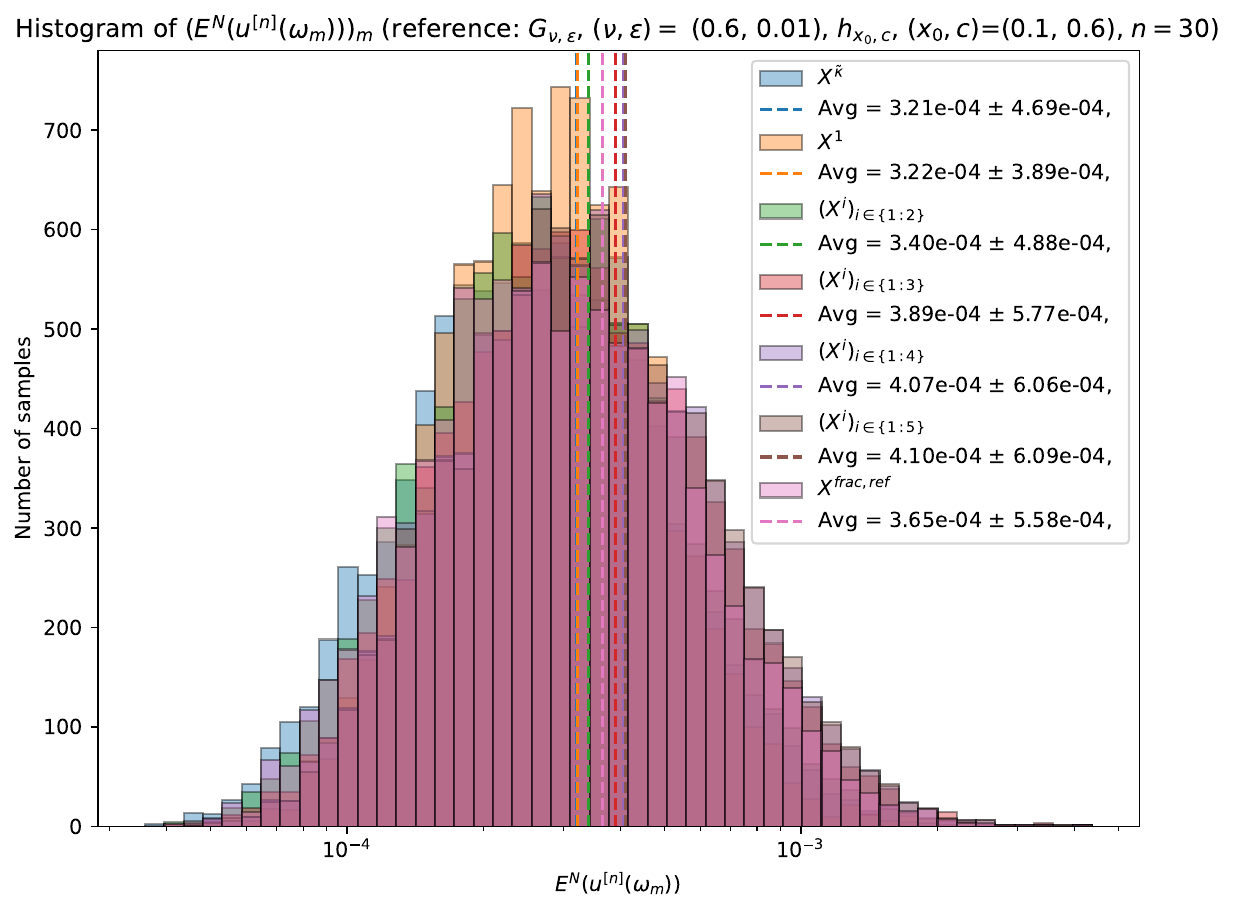} 
    \end{subfigure}
    
    \caption{Impact of the choice of regression variables in the Least Square Monte Carlo on the empirical $PnL^{N,M}$ from \eqref{eq:empirical_pnl_functional_alpha_impact_cost_trade_of} through the scheme \eqref{scheme_initialization}--\eqref{scheme_update} iterations $n \in \mathbb{N}$ (left) and the histograms of the empirical errors $\left( E^{N}(u^{[n]}(\omega_{m}) \right)_{m}$ of the scheme from \eqref{eq:empirical_error_metric_per_omega} (the respective empirical averages $E^{N,M}(u^{[n]})$ from \eqref{eq:empirical_error_metric} are displayed in dashed line) after $n=30$ iterations (right). In each case, we estimate the conditional expectations \eqref{eq:conditional_expectations_to_estimate} required to solve the linear Fredholm equation at each iteration by LSMC with the regression variables $\left( \alpha, \int_{0}^{.} \alpha_{s} \dd s \right)$ and additionally those mentioned respectively in the legend. Laguerre polynomials up to degree $2$ are used for the basis expansion. For a fair comparison, in all the cases, the conditional expectations in the error metric \eqref{eq:empirical_error_metric_per_omega} are estimated using the regression basis $\left( \alpha, \int_{0}^{.} \alpha_{s} \dd s, X^{\tilde{\kappa}} \right)$ with a Laguerre polynomial basis expansion up to degree $3$. We apply a Ridge regularization with intensity $1e-6$ for all linear regressions.}
    \label{fig:impact_regression_basis}
\end{figure}

\begin{figure}[H]
    \centering
    \begin{subfigure}[t]{0.48\textwidth}
        \centering
        \includegraphics[width=\textwidth, trim=0cm 0cm 0cm 0cm, clip]{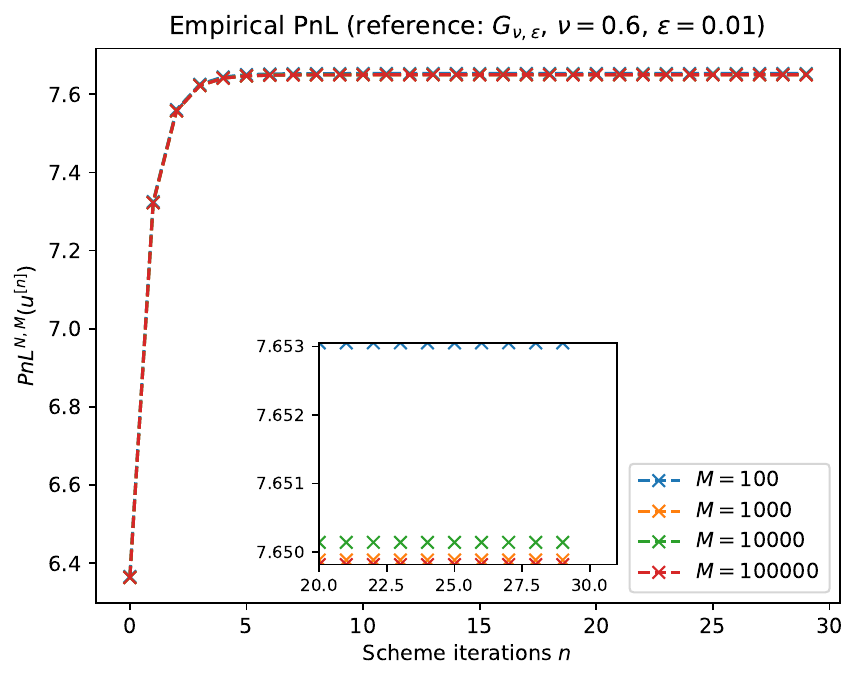} 
    \end{subfigure}
    \hfill 
    \begin{subfigure}[t]{0.48\textwidth}
        \centering
        \includegraphics[width=\textwidth, trim=0cm 0cm 0cm 0cm, clip]{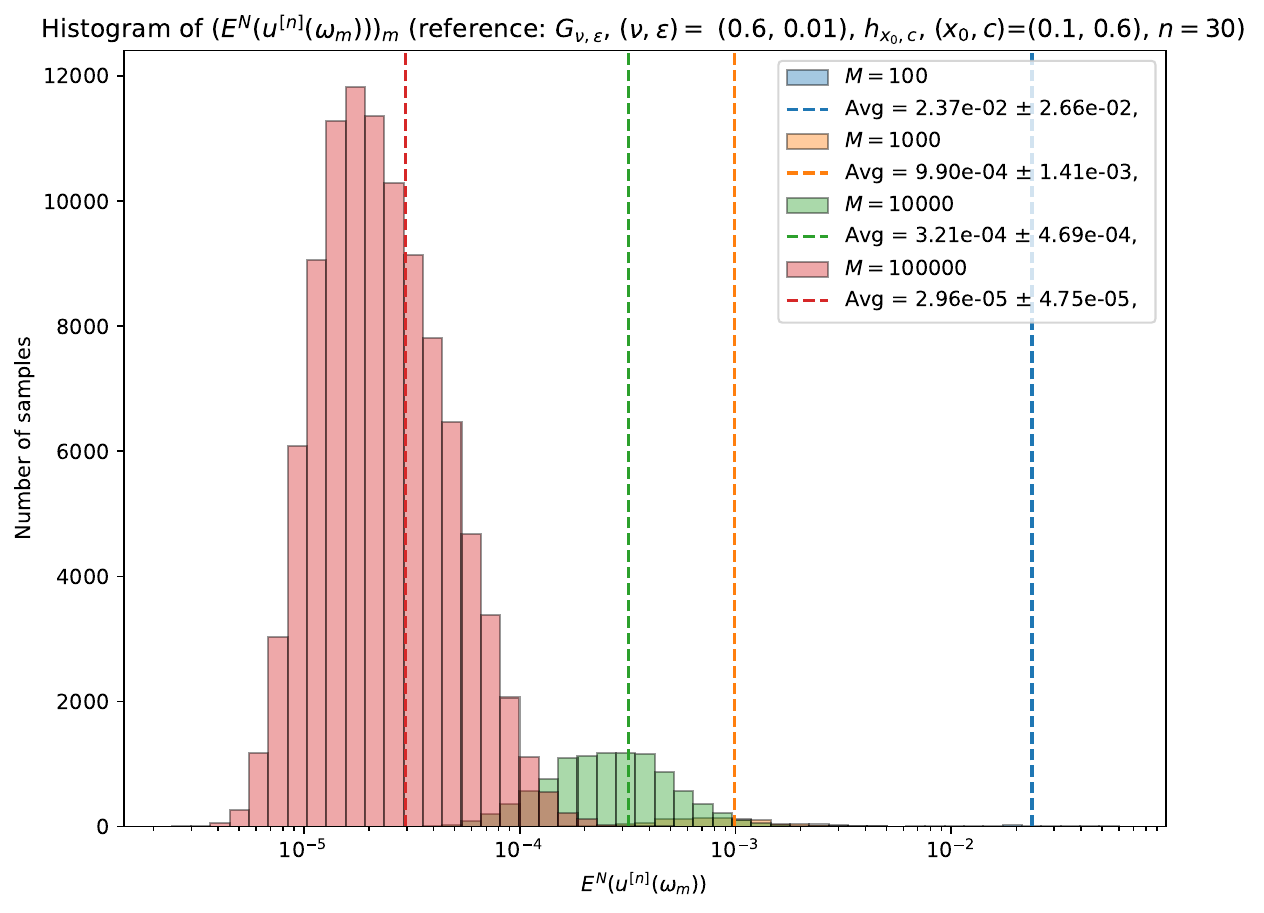} 
    \end{subfigure}
    
    \caption{Impact of the number of sample trajectories $M \in \mathbb{N}^{*}$ of the signal in the Least Square Monte Carlo on the empirical $PnL^{N,M}$ from \eqref{eq:empirical_pnl_functional_alpha_impact_cost_trade_of} through the scheme \eqref{scheme_initialization}--\eqref{scheme_update} iterations $n \in \mathbb{N}$ (left) and the histograms of the empirical errors $\left( E^{N}(u^{[n]}(\omega_{m}) \right)_{m}$ of the scheme from \eqref{eq:empirical_error_metric_per_omega} (the respective empirical averages $E^{N,M}(u^{[n]})$ from \eqref{eq:empirical_error_metric} are displayed in dashed line) after $n=30$ iterations (right). In each case, we estimate the conditional expectations \eqref{eq:conditional_expectations_to_estimate} required to solve the linear Fredholm equation at each iteration by LSMC with the regression variables $\left( \alpha, \int_{0}^{.} \alpha_{s} \dd s, X^{\tilde{\kappa}} \right)$. Laguerre polynomials up to degree $2$ are used for the basis expansion. For a fair comparison, in all the cases, the conditional expectations in the error metric \eqref{eq:empirical_error_metric_per_omega} are estimated using the regression basis $\left( \alpha, \int_{0}^{.} \alpha_{s} \dd s, X^{\tilde{\kappa}} \right)$ with a Laguerre polynomial basis expansion up to degree $3$. We apply a Ridge regularization with intensity $1e-6$ for all linear regressions.}
    \label{fig:impact_nber_trajs}
\end{figure}

\end{document}